\documentclass[onecolumn]{IEEEtran}

\usepackage{latexsym}
\usepackage{graphicx}
\usepackage{array}
\usepackage{amsmath}
\usepackage{amsfonts}
\usepackage{amssymb}
\usepackage{amsthm}
\usepackage{epsfig}
\usepackage{epstopdf}
\usepackage{cite}
\usepackage{algorithm2e}
\usepackage{mathtools}
\usepackage{etoolbox}
\usepackage{lipsum}
\usepackage{caption}
\usepackage{subcaption}
\usepackage{rotating,graphics,psfrag}
\let\bbordermatrix\bordermatrix
\patchcmd{\bbordermatrix}{8.75}{4.75}{}{}
\patchcmd{\bbordermatrix}{\left(}{\left[}{}{}
\patchcmd{\bbordermatrix}{\right)}{\right]}{}{}
 
\numberwithin{equation}{section}

\newtheorem{theorem}{Theorem}[section]
\newtheorem{lemma}[theorem]{Lemma}
\newtheorem{definition}[theorem]{Definition}

\newtheorem{remark}[theorem]{Remark}

\newtheorem{assumption}[theorem]{Assumption}

\newcommand{\sr}{\stackrel}

\newcommand{\tri}{\sr{\triangle}{=}}

\makeatletter
\def\rightharpoonupfill@{\arrowfill@\relbar\relbar\rightharpoonup}
\newcommand{\overrightharpoonup}{%
   \mathpalette{\overarrow@\rightharpoonupfill@}}
\makeatother
\newcommand{\be}{\begin{equation}}
\newcommand{\ee}{\end{equation}}
\newcommand{\bea}{\begin{eqnarray}}
\newcommand{\eea}{\end{eqnarray}}
\newcommand{\bes}{\begin{eqnarray*}}
\newcommand{\ees}{\end{eqnarray*}}
\newcommand{\bi}{\begin{itemize}}
\newcommand{\ei}{\end{itemize}}
\newcommand{\ben}{\begin{enumerate}}
\newcommand{\een}{\end{enumerate}}
\newcommand{\bp}{\begin{problem}}
\newcommand{\ep}{\end{problem}}
\newcommand{\hso}{\hspace{.1in}}
\newcommand{\hst}{\hspace{.2in}}

\newcommand{\noi}{\noindent}

\begin{document}


\title{Information Nonanticipative Rate Distortion 
\\Function and its Applications}
%
%
\author{\IEEEauthorblockN{ Photios~A.~Stavrou\IEEEauthorrefmark{1},  Christos K. Kourtellaris\IEEEauthorrefmark{2}, and Charalambos~D.~Charalambous\IEEEauthorrefmark{1}}
\thanks{This work was financially supported by a medium size University of Cyprus grant entitled ``DIMITRIS"  and by QNRF, a member of Qatar Foundation, under the project NPRP 6-784-2-329. Part of the material in this paper was presented in IEEE International Symposium on Information Theory, 2014 \cite{stavrou-kourtellaris-charalambous2014isit} and in book series ``Lecture Notes in Control and Information Sciences" \cite{stavrou-kourtellaris-charalambous2015}.}
\thanks{\IEEEauthorblockA{\IEEEauthorrefmark{1} Department of Electrical and Computer Engineering, University of Cyprus, Nicosia, Cyprus}
\IEEEauthorblockA{\IEEEauthorrefmark{2} Department of Electrical
and Computer Engineering, Texas A\&{M} University at Qatar} 
{\it Email:\{stavrou.fotios,chadcha\}@ucy.ac.cy, c.kourtellaris@qatar.tamu.edu}}}
%


\maketitle

\begin{abstract}
This paper investigates applications of nonanticipative Rate Distortion Function (RDF) in a) zero-delay Joint Source-Channel Coding (JSCC)  design based on average and excess distortion probability, b) in bounding the Optimal Performance Theoretically Attainable (OPTA) by noncausal and causal codes, and computing the Rate Loss (RL) of zero-delay and causal codes with respect to noncausal codes. These applications are described using two running examples, the Binary Symmetric Markov Source with parameter $p$, (BSMS($p$)) and the multidimensional partially observed Gaussian-Markov source.\\
For the multidimensional Gaussian-Markov source with square error distortion, the solution of the nonanticipative RDF is derived, its operational meaning using JSCC design via a noisy coding theorem  is shown by providing the optimal encoding-decoding scheme over a vector Gaussian channel, and the RL of causal and zero-delay codes with respect to noncausal codes is computed.\\
\noi For the BSMS($p$) with Hamming distortion, the solution of the nonanticipative RDF is derived, the RL of causal codes with respect to noncausal codes is computed, and an uncoded noisy coding theorem based on excess distortion probability is shown. 
\noi The information nonanticipative RDF is shown to be equivalent to the nonanticipatory $\epsilon$-entropy, which corresponds to the classical RDF with an additional causality or nonanticipative condition imposed on the optimal reproduction conditional distribution. 
\end{abstract}


\section{Introduction}\label{introduction}
\par In lossy compression source coding with fidelity \cite{berger,han2003}, the sequence of real-valued symbols $x^\infty\tri\{x_0,x_1,\ldots\}$, ${x}_j\in{\cal X}$, $\forall{j}\geq{0}$, generated by a source with distribution $P_{X^\infty}$, is transformed by the encoder into a sequence of symbols, the compressed representation $z^\infty\tri\{z_0,z_1,\ldots\}$ (taking values in a finite alphabet set), which is then transmitted over a noiseless channel. The decoder at the channel output upon observing the compressed representation symbols produces the reproduction sequence $y^\infty\tri\{y_0,y_1,\ldots\}$, ${y}_j\in{\cal Y}$, $\forall{j}\geq{0}$. Such a compression system is called causal \cite{neuhoff1982} if the reproduction symbol $y_n$ of the source symbol $x_n$, depends on the present and past source symbols $\{x_0,\ldots,x_n\}$ but not on the future source symbols $\{x_{n+1},x_{n+2},\ldots\}$, i.e., $y_n\tri{f}_n(x_0,\ldots,x_n)$, for some function $f_n(\cdot),~n\geq{0}$. The Optimal Performance Theoretically Attainable (OPTA) by noncausal codes is given by the Rate Distortion Function (RDF) \cite{berger,han2003}, while that of causal codes is given by the entropy rate of the reproduction sequence  \cite{neuhoff1982}. Zero-delay source codes are a sub-class of causal codes, with the additional constraint that the compressed representation symbol $z_n$, depends on the past and present source symbols $x^n\tri\{x_0,x_1,\ldots,x_n\}$, while the reproduction at the decoder $y_n$ of the present source symbol $x_n$, depends only on the compressed representation $z^n\tri\{z_0,z_1,\ldots,z_n\}$, i.e., $z_n=h_n(\{x_j:~j=0,1,\ldots,n\})$ and $y_n=f_n(\{z_j:~j=0,1,\ldots,n\})$, $\forall{n}\geq{0}$ \cite{ericson1979,gaarder-slepian1982,tatikonda2000,linder-lagosi2001,akyol-viswanatha-rose-ramstad2014}.

\noi Joint Source-Channel Coding (JSCC) design based on nonanticipative processing (i.e., the encoder-channel-decoder mappings at each time instant do not depend on future symbols),   is perhaps, the most efficient zero-delay coding system,  in the sense of optimal performance of matching the  source characteristics to the channel  characteristics, coded or uncoded   \cite{gastpar2003,berger2003,goblick1965}. Two such fascinating examples of matching the RDF of an IID source to the capacity of a memoryless channel are the following.\\
\noi {\it Example-IID-BSS:} The Independent Identically Distributed (IID) Binary Source (BS) with Hamming distortion transmitted uncoded over a symmetric memoryless channel (the distortion and channel parameter are made equal);\\
{\it Example-IID-GS:} The IID Gaussian Source (GS) with average squared-error distortion, transmitted over a memoryless Additive White Gaussian Noise (AWGN) channel without feedback, with the encoder and decoder scaling their inputs \cite{goblick1965}.\\ 
\noi In both examples, symbol-by-symbol transmission is optimal, that is, the optimal code rate is ``1" in terms of achieving average end-to-end distortion. These examples demonstrate the simplicity of the JSCC design often called probabilistic matching, in operating optimally with zero-delay, in complexity, when this is compared to the asymptotic performance of optimally separating the encoder/decoder to the source and channel encoders/decoders which may cause long processing delays. Recently,  {\it Example-IID-BSS} and {\it Example-IID-GS} are revisited in \cite{kostina-verdu2012itw}, using symbol-by-symbol (zero-delay) codes.

\noi In general, very little is known about JSCC design or probabilistic matching based on nonanticipative processing, that of matching the source characteristics to the channel characteristics, as in {\it Example-IID-BSS} and in {\it Example-IID-GS}, because the RDF of sources with memory is generally not known. Moreover, only bounds are known on the OPTA by causal and zero-delay codes \cite{berger2003}. For the latter, these bounds are often introduced to quantify the Rate Loss (RL) due to causality and zero-delay of the coding systems compared to that of the noncausal coding systems. Such bounds are elaborated recently in \cite{derpich-ostergaard2012}, for Gaussian sources with square-error distortion function. 

\noi In this paper, we consider the information nonanticipative RDF \cite{charalambous-stavrou-ahmed2014a,charalambous-stavrou2013bb}, a variant of the classical information RDF, and we describe its applications in 
\begin{description}
\item[\bf{(a)}] zero-delay JSCC design or probabilistic matching of the source to the channel, with respect to nonanticipative processing (zero delay or symbol-by-symbol codes), based on average and excess distortion probability;
\item[\bf{(b)}] bounding the OPTA by noncausal and causal codes, and computing the RL of zero-delay and causal codes with respect to noncausal codes.
\end{description}
\noi For {\bf(a)}, the design of encoder and decoder for a given $\{$source, channel$\}$ pair with respect to a given $\{$distortion function, transmission cost function$\}$ pair is constructive. Given a source, first we compute the nonanticipative RDF. Second, we realize the nonanticipative RDF (see Fig.~\ref{realization_nonanticipative_RDF}) with respect to the channel by providing an $\{$encoder, decoder$\}$ pair, so that the end-to-end average distortion is achieved, the channel operates at its capacity, and the nonanticipative RDF is equal to the channel capacity. For the multidimensional (vector) stationary Gaussian source with memory and square-error distortion, we show that the solution of the nonanticipative RDF is equivalent to JSCC design of a vector AWGN channel with and without feedback encoding in which the encoder and the decoder are linear. For the special case of IID Gaussian sequence, our JSCC design reduces to the JSCC design of {\it Example IID-GS} \cite{goblick1965,gastpar2003}. Moreover, from the solution of the nonanticipative RDF we recover the Schalkwijk-Kailath's coding scheme \cite{schalkwijk-kailath1966}.\\
\noi For {\bf(b)}, we use the nonanticipative RDF to compute the RL or gap between the OPTA by causal \cite{neuhoff1982} and zero-delay codes with respect to the OPTA by noncausal codes, for both finite alphabet and continuous alphabet valued sources.\\
\noi Moreover, we show equivalence of the information nonanticipative RDF and its rate  to Gorbunov and Pinsker  nonanticipatory $\epsilon-$entropy and message generation rates \cite{gorbunov-pinsker,pinsker-gorbunov1987,gorbunov-pinsker1991}, which corresponds to Shannon information RDF with an additional causality constraint imposed on the optimal reproduction distribution. Gorbunov and Pinsker \cite{gorbunov-pinsker} appear to be the first who recognized the importance of nonanticipative RDF in real-time applications \cite[p. 1, lines 4-5]{gorbunov-pinsker}. 
\vspace*{0.2cm}\\
Next, we introduce the nonanticipative RDF, classical RDF and Gorbunov and Pinsker's nonanticipatory $\epsilon$-entropy to clarify certain relations between them.\\
\noi{\bf Nonanticipative RDF.} The information nonanticipative RDF is defined as follows. Consider a source distribution $P_{X^n}(dx^n)$, a causal sequence of reproduction distributions $\{P_{Y_i|Y^{i-1},X^i}(dy_i|y^{i-1},x^i):~i=0,1,\ldots,n\}$, a measurable distortion function $d_{0,n}(x^n,y^n):{\cal X}_{0,n}\times{\cal Y}_{0,n}\longmapsto[0,\infty]$ and an average fidelity set\footnote{$\otimes$ denotes convolution of distributions.}
\begin{align}
\overrightarrow{\cal Q}_{0,n}(D)&\tri\bigg\{\overrightarrow{P}_{Y^n|X^n}(dy^n|x^n)\tri\otimes_{i=0}^n{P}_{Y_i|Y^{i-1},X^i}(dy_i|y^{i-1},x^i):\nonumber\\
&\quad~\frac{1}{n+1}\int_{{\cal X}_{0,n}\times{\cal Y}_{0,n}}d_{0,n}(x^n,y^n)(\overrightarrow{P}_{Y^n|X^n}\otimes{P}_{X^n})(dx^n,dy^n)\leq{D}\bigg\}.\label{eq.12}
\end{align}
The information nonanticipative RDF is defined by
\begin{align}
R^{na}_{0,n}(D)&\tri\inf_{\overrightarrow{P}_{Y^n|X^n}(\cdot|x^n)\in\overrightarrow{\cal Q}_{0,n}(D)}\int_{{\cal X}_{0,n}\times{\cal Y}_{0,n}}\log\Big(\frac{\overrightarrow{P}_{Y^n|X^n}(dy^n|x^n)}{P_{Y^n}(dy^n)}\Big)(\overrightarrow{P}_{Y^n|X^n}\otimes{P}_{X^n})(dx^n,dy^n)\nonumber\\
&=\inf_{\overrightarrow{P}_{Y^n|X^n}(\cdot|x^n)\in\overrightarrow{\cal Q}_{0,n}(D)}\mathbb{I}_{X^n\rightarrow{Y^n}}(P_{X^n},\overrightarrow{P}_{Y^n|X^n}).\label{equation1h}
\end{align}
\noi Here, $\mathbb{I}_{X^n\rightarrow{Y^n}}(\cdot,\cdot)$ is used to denote the functional dependence of $R_{0,n}^{na}(D)$ on the two distributions $\{P_{X^n},\overrightarrow{P}_{Y^n|X^n}\}$. The information nonanticipative RDF rate is defined by
\begin{align}
R^{na}(D)\tri\lim_{n\longrightarrow\infty}\frac{1}{n+1}R_{0,n}^{na}(D).\label{equation1q}
\end{align}

\noi{\bf Classical RDF.} The classical information RDF (often called OPTA by noncausal codes)\cite{berger,csiszar74,han2003} is defined by
\begin{align}
R(D)\tri\lim_{n\longrightarrow\infty}\frac{1}{n+1}R_{0,n}(D),~~R_{0,n}(D)\tri\inf_{P_{Y^n|X^n}(\cdot|x^n)\in{\cal Q}_{0,n}(D)}I(X^n;Y^n)\label{equation1b}
\end{align}
where the fidelity set of reproduction conditional distributions is defined by
\begin{align}
{\cal Q}_{0,n}(D)\tri\bigg\{P_{Y^n|X^n}(dy^n|x^n):\frac{1}{n+1}\int_{{\cal X}_{0,n}\times{\cal Y}_{0,n}}d_{0,n}(x^n,y^n)(P_{Y^n|X^n}\otimes{P}_{X^n})(dx^n,dy^n)\leq{D}\bigg\}.\label{equation1a}
\end{align}
Goblick in \cite{goblick1965} applied the $R(D)$ of an IID Gaussian distributed  ${N}(0;\sigma^2_X)$ sequence  with $R(D)=\frac{1}{2}\log(\frac{\sigma^2_X}{D}),$~$0\leq{D}\leq\sigma^2_X$ to provide the JSCC design over a memoryless AWGN channel without feedback. The extension of Goblick's JSCC design to a memoryless AWGN channel with feedback is treated in \cite{ihara1993}, and makes use of the Schalkwijk-Kailath's coding scheme.\\
\noi{\bf Nonanticipatory $\epsilon$-Entropy.} Gorbunov and Pinsker in \cite{gorbunov-pinsker} introduced the so-called nonanticipatory $\epsilon$-entropy, $R^{\varepsilon}_{0,n}(D)$, and message generation rate $R^{\varepsilon}(D)$ as follows
\begin{align}
R^{\varepsilon}(D)&\tri\lim_{n\longrightarrow\infty}\frac{1}{n+1}R^{\varepsilon}_{0,n}(D),\label{equation1kk}\\ R^{\varepsilon}_{0,n}(D)&\tri\inf_{\substack{P_{Y^n|X^n}(\cdot|x^n)\in{\cal Q}_{0,n}(D)\\ X_{i+1}^n\leftrightarrow{X}^i\leftrightarrow{Y^i},~i=0,1,\ldots,n-1}}I(X^n;Y^n),~\forall{n}\geq{0}.\label{equation1k}
\end{align}
The same authors in \cite{pinsker-gorbunov1987,gorbunov-pinsker1991}, showed that for scalar stationary Gaussian sources (using power spectral densities) that $R^{\varepsilon}(D)$ tends to Shannon's RDF $R(D)$, as $D\longrightarrow{0}$.

\noi We show in Theorem~\ref{equivalent_rdf} that $R^{na}_{0,n}(D)=R^\varepsilon_{0,n}(D)$. The fundamental difference between $R_{0,n}^{na}(D)$ or $R^\varepsilon_{0,n}(D)$ and $R_{0,n}(D)$ is the following. For sources with memory, the optimal reproduction distribution of the information nonanticipative RDF, $R^{na}_{0,n}(D)$ at each time $i=0,\ldots,n$, is a causal (or nonanticipative) sequence of conditional distributions $P^*_{Y_i|Y^{i-1},X^i}(dy_i|y^{i-1},x^i)$, that is, it depend only on the past and present source symbols and past reproduction symbols $\{x^i,y^{i-1}\}$. On the other hand, the optimal reproduction distribution of the classical RDF, $R_{0,n}(D)$ at each time $i=0,\ldots,n$, is a sequence of noncausal conditional distributions $P^*_{Y_i|Y^{i-1},X^n}(dy_i|y^{i-1},x^n)$, which depends not only on $\{x^i,y^{i-1}\}$ but also on future source symbols $\{x_{i+1},\ldots,x_n\}$. For independent sources, $R_{0,n}^{na}(D)$ and $R_{0,n}(D)$ coincide.
\par Next, we discuss certain limitations of the classical information RDF $R(D)$ with respect to applications, which motivated our interest to  investigate the information nonanticipative RDF (\ref{equation1h}).

\subsection{Motivation for Nonanticipative RDF}\label{introduction:motivation}

\noi The first limitation of the classical information RDF is the lack of examples of a source with memory for which the exact expression of $R_{0,n}(D)$ for finite $n$, and that of $R(D)\tri\lim_{n\longrightarrow\infty}\frac{1}{n+1}R_{0,n}(D)$, are computed, aside from memoryless or Gaussian sources. For example, the RDF of the BSMS($p$) with single letter Hamming distortion function is currently unknown (see Section~\ref{example:bsms} for more discussion).\\

\noi The second limitation of the classical information RDF $R_{0,n}(D)$ is the anticipative (noncausal) nature of the optimal reproduction distributions $\{P^*_{Y_i|Y^{i-1},X^n}(dy_i|y^{i-1},x^n):~i=0,1,\ldots,n\}$, because they depend on future source symbols even for single letter distortion functions $d_{0,n}(x^n,y^n) \tri \sum_{i=0}^n\rho(x_i,y_i)$. This anticipative structure of the optimal reproduction distribution corresponding to $R_{0,n}(D)$ implies that, in general, the classical information RDF cannot be used in JSCC design using nonanticipative processing or uncoded transmission (i.e., the $\{$encoder, decoder$\}$ pair processes at each time instant source symbol without dependence on future source symbols), also called probabilistic matching of the source to the channel \cite{berger2003}, unless the source is memoryless, such as, the {\it Example-IID-BS} and the {\it Example-IID-GS} discussed earlier \cite{gastpar2003,kostina-verdu2012itw,kostina-verdu2012,kostina-verdu2013}. Indeed, to perform the JSCC design using nonanticipative processing it is necessary that {\it 1)} the exact expression of the RDF of the source is known or computed, and {\it 2)} its corresponding optimal reproduction distribution is realizable using nonanticipative processing via $\{$encoder, channel, decoder$\}$ mappings processing symbols without dependence on future symbols, as shown in Fig.~\ref{realization_nonanticipative_RDF}. This means that for general sources with memory, the optimal conditional reproduction distribution of $R_{0,n}(D)$ should satisfy the Markov chain (MC) (i.e., conditional independence)
\begin{figure}[ht]
\centering
\includegraphics[scale=0.70]{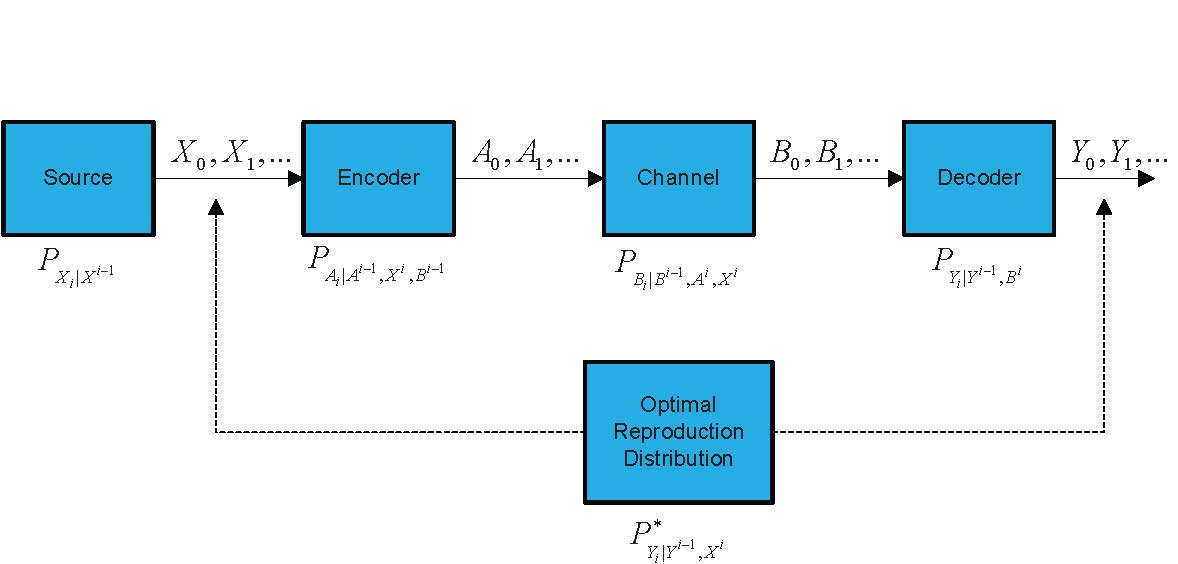}
\caption{ Probabilistic matching based on nonanticipative transmission.}
\label{realization_nonanticipative_RDF}
\end{figure}
\begin{align}
{P}^*_{Y_i|Y^{i-1},X^n}(dy_i|y^{i-1},x^\infty)={P}^*_{Y_i|Y^{i-1},X^i}(dy_i|y^{i-1},x^i)\Longleftrightarrow{X}_{i+1}^\infty\leftrightarrow(X^i,Y^{i-1})\leftrightarrow{Y_i},~i=0,1,\ldots.\label{equation1f}
\end{align}
\noi Alternatively stated, if (\ref{equation1f}) does not hold, then nonanticipative processing of information in the sense of Fig.~\ref{realization_nonanticipative_RDF} is violated.

\subsection{Summary of Main Results and Related Literature}

\par Next, we present a summary of the contributions in this paper and we discuss relations to existing literature.

\noi{\bf Relations to Nonanticipatory $\epsilon$-Entropy.} In subsection~\ref{equivalence_nrdf_epsilon_entropy} we show that the information nonanticipative RDF, $R^{na}_{0,n}(D)$ and its rate $R^{na}(D)$ (see (\ref{equation1h}), (\ref{equation1q})), are equivalent to Gorbunov and Pinsker's nonanticipatory $\epsilon$-entropy and message generation rate \cite{gorbunov-pinsker}, $R^{\varepsilon}_{0,n}(D)$ and $R^{\varepsilon}(D)$, respectively. In subsection~\ref{infinite_horizon}, we combine the existence of solution to the nonanticipative RDF shown in \cite[Theorem III.4]{charalambous-stavrou-ahmed2014a}, and the main theorems in \cite{gorbunov-pinsker} to establish that for general stationary sources, the limit exists, and the limit and infimum operations can be interchanged, that is,
\begin{align}
R^{na}(D)
&=\lim_{n\longrightarrow\infty}\inf_{\overrightarrow{P}_{Y^n|X^n}\in\overrightarrow{\cal Q}_{0,n}(D)}\frac{1}{n+1}\mathbb{I}_{X^n\rightarrow{Y^n}}(P_{X^n},\overrightarrow{P}_{Y^n|X^n})\nonumber\\
&=\inf_{\overrightarrow{P}_{Y^\infty|X^\infty}\in\overrightarrow{\cal Q}_{0,\infty}(D)}\lim_{n\longrightarrow\infty}\frac{1}{n+1}\mathbb{I}_{X^n\rightarrow{Y^n}}(P_{X^n},\overrightarrow{P}_{Y^n|X^n})\equiv\overrightharpoonup{R}^{na}(D)<\infty.\label{equation705}
\end{align}
\noi and that the optimal reproduction distribution of $R^{na}(D)$  corresponds to jointly stationary source-reproduction pair $\{(X_i,Y_i):~i=0,1,\ldots\}$. 
\vspace*{0.2cm}\\
\noi{\bf Closed Form Expression of Information Nonanticipative RDF (Stationary Solution), Properties, and Examples.} In Section~\ref{optimal_reproduction_distribution_section}, we characterize the stationary solution of the nonanticipative RDF and some of its properties. One of the properties is analogous to the Shannon lower bound of the classical RDF \cite{berger}. In subsections~\ref{example:bsms} and \ref{example:gaussian}, we use these to compute the information theoretic nonanticipative RDF in closed form for the following two running examples.\\
{\it Example 1:~BSMS($p$) with single letter Hamming distortion.} For the Binary Symmetric Markov Source with parameter $p$ (BSMS($p$)) and   single letter Hamming distortion, we show that
\begin{align}
{R}^{na}(D) = \left\{ \begin{array}{ll} H(m)-H(D) & \mbox{if $D \leq \frac{1}{2}$}\\
        0 &  \mbox{otherwise} \end{array} \right., m=1-p-D+2pD.\label{BSMS_I1}
\end{align}
\noi Note that, for $p=\frac{1}{2}$, then $R^{na}(D)$ reduces to the classical RDF, $R(D)$, of IID memoryless source, as expected. Recently,  $R_{0,n}^{na}(D)$ of BSMS($p$) is utilized in \cite{johnston-modiano-polyanskiy2014isit} to demonstrate an example of an opportunistic scheduling system. In section~\ref{bounds_opta}  we use (\ref{BSMS_I1}) to obtain tighter lower bounds on the OPTA by causal and noncausal codes, compared to those given in \cite{berger1977}.
\vspace*{0.2cm}\\
\noi {\it Example 2:~Multidimensional Gaussian-Markov source with square-error  distortion.} For the multidimensional partially observed stationary Gauss-Markov source, described in state space form (this includes autoregressive models) by 
\begin{align}
\left\{ \begin{array}{ll} Z_{t+1}=AZ_t+BW_t,~Z_0\sim{N}(0,\bar{\Sigma}_0),~t=0,1,\ldots\\
X_t=CZ_t+NV_t,~t=0,1,\ldots\end{array} \right.\label{equation1o}
\end{align}
\begin{figure}[ht]
\centering
\includegraphics[scale=0.70]{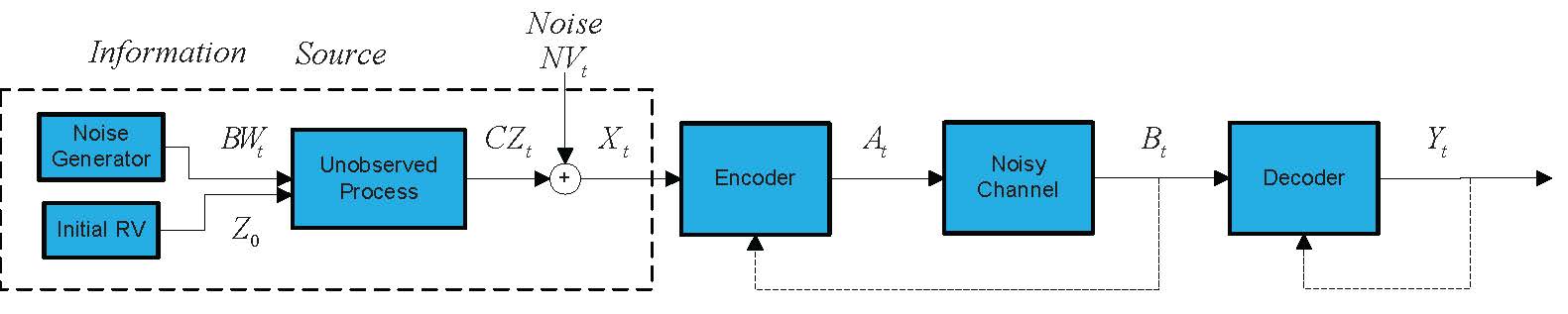}
\caption{JSCC design system of partially observed multidimensional Gaussian-Markov source.}
\label{communication_system1}
\end{figure}
\noi where $Z_t\in\mathbb{R}^m$ is the state (unobserved) process driven by a multidimensional Gaussian noise process $\{W_t:~t=0,1,\ldots\}$, and $X_t\in\mathbb{R}^p$ is the observed source process corrupted by a multidimensional additive Gaussian noise process $\{V_t:~t=0,1,\ldots\}$ (see\footnote{In this application, the objective is to reconstruct $\{X_t:~t=0,1,\ldots\}$ by $\{Y_t:~t=0,1\ldots\}$ with respect to a certain fidelity of reproduction. Under certain assumptions \cite{caines1988} $\{X_t:~t=0,1,\ldots,n\}$ is ergodic (although the $A$ matrix of $\{Z_t:~t=0,1,\ldots,n\}$ may have unstable eigenvalues).} Fig.~\ref{communication_system1}),  we utilize the characterization of the solution of $R_{0,n}^{na}(D)$ obtained in Theorem~\ref{alternative_expression}, to show that for a square error distortion function, the nonanticipative RDF is given by
\begin{align}
R^{na}(D)=\frac{1}{2}\sum_{i=0}^p\log\big(\frac{\lambda_{\infty,i}}{\delta_{\infty,i}}\big)=\frac{1}{2}\log\frac{|\Lambda_\infty|}{|\Delta_\infty|}\label{equation1r}
\end{align}
where
\begin{align*}
\Lambda_{\infty}\tri\lim_{n\longrightarrow\infty}\Lambda_t, \Lambda_t\tri\mathbb{E}\Big\{\Big(X_t-\mathbb{E}\big(X_t|Y^{t-1}\big)\Big)\Big(X_t-\mathbb{E}\big(X_t|Y^{t-1})\Big)^{tr}\Big\}\equiv\mbox{(\ref{lambda_infty})},~\Delta_\infty=diag\{\delta_{\infty,1},\ldots,\delta_{\infty,p}\}
\end{align*}
Here, $\{\lambda_{\infty,1},\ldots,\lambda_{\infty,p}\}$ are the steady state eigenvalues of $\Lambda_{\infty}$, and 
\begin{align}
\delta_{\infty,i} \tri\left\{ \begin{array}{ll} \xi_\infty & \mbox{if} \quad \xi_\infty\leq\lambda_{\infty,i} \\
\lambda_{\infty,i} &  \mbox{if}\quad\xi_\infty>\lambda_{\infty,i} \end{array} \right.,~i=1,\ldots,p,~\sum_{i=1}^p\delta_{\infty,i}=D.\nonumber
\end{align}
\noi In addition, in subsection~\ref{example:gaussian}, we recover from (\ref{equation1r}), several special cases. These include the expression of the nonanticipative RDF  for the {\it scalar stationary fully observed Gaussian-Markov source} (i.e., corresponding to $p=1$, $C=1$, $N=0$, $A=\alpha$, $B=\sigma_W$) also obtained in \cite[Theorem 3]{derpich-ostergaard2012} via alternative methods, and for {\it IID Gaussian sources} $\{X_t:~t=0,1\ldots\}\sim{N}(0;\sigma^2_X)$ the known expression $R(D)=\frac{1}{2}\log\frac{\sigma^2_X}{D},~\sigma^2_X\geq{D}\geq{0}$ .\\

\begin{figure}[ht]
\centering
\includegraphics[scale=0.70]{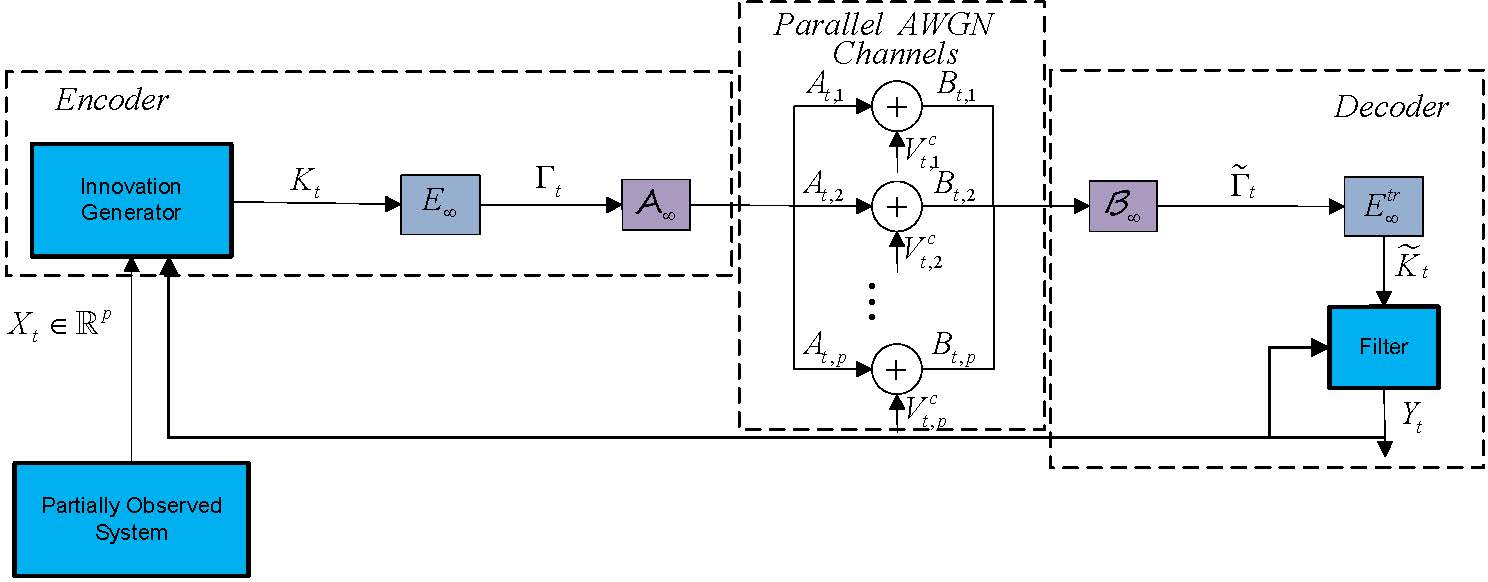}
\caption{Realization of the optimal reproduction distribution: $K_t=X_t-\mathbb{E}\{X_t|Y^{t-1}\}$, $\tilde{K}_t=Y_t-\mathbb{E}\{X_t|Y^{t-1}\}$, $\{E_{\infty}, E^{tr}_{\infty}\}$=unitary transformations, $\{{\cal A}_{\infty},{\cal B}_{\infty}\}$=scaling matrices, $\{\Gamma_t,\tilde{\Gamma}_t\}$= scaling matrices.}
\label{communication_system}
\end{figure}
\noi We note that recently in \cite{tanaka-kim-parrilo-mitter2014} the importance of example (\ref{equation1o}) is elaborated for the special case of fully observed Gaussian-Markov sources, i.e., $N=0$, $C=I$ (i.e., $X_t=Z_t$) in the context of the so-called sequential rate distortion theory \cite{tatikonda2000}, which is equivalent to the nonanticipative RDF. The authors in \cite{tanaka-kim-parrilo-mitter2014} provide an approximation solution by utilizing a semidefinite programming approach, while we give the exact closed form expression (\ref{equation1r}) for the more general model (\ref{equation1o}).
\vspace*{0.2cm}\\
\noi {\bf Joint Source-Channel Coding:~Symbol-by-Symbol Transmission.} In Section~\ref{joint_source_channel_coding},  we show that (\ref{equation1r}) is achievable by the JSCC design system depicted in Fig.~\ref{communication_system}, where the encoder operates using symbol-by-symbol codes, achieves channel capacity $C(P)$ ($P$ is the power allocated for transmission), while  the end-to-end average distortion is met, and $R^{na}(D)=C(P)$. \\
Moreover, we demonstrate that the general JSCC design system of Fig.~\ref{communication_system} gives as degraded cases (corresponding to $p=1$, $C=1$, $N=0$, $A=\alpha$, $B=\sigma_W$),  the following new examples of JSCC designs operating based on symbol-by-symbol codes.
\begin{description}
\item[{\bf(i)}] The optimal coding scheme of a {\it scalar Gaussian-Markov Source over a memoryless AWGN with feedback} (see Fig.~\ref{jscc_fully_feedback});
\item[{\bf(ii)}] the optimal coding scheme of a {\it scalar Gaussian-Markov Source over a memoryless AWGN channel without feedback} (see Fig.~\ref{jscc_fully_nofeedback}).
\end{description}
\begin{figure}[ht]
\centering
\includegraphics[scale=0.70]{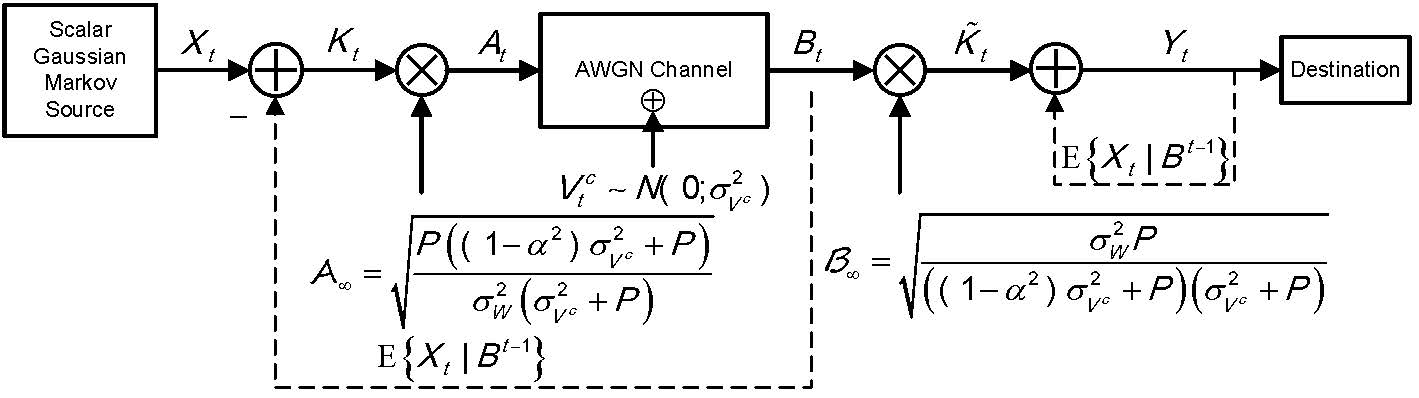}
\caption{Reduction of Fig.~\ref{communication_system}  to $p=1$, $C=1$, $N=0$, $A=\alpha, |\alpha|<1$, $B=\sigma_W$, which corresponds to the JSCC design system of a scalar Gaussian Markov Source given by (\ref{equation111}) over a memoryless AWGN channel with feedback, $K_t=X_t-\mathbb{E}\{X_t|B^{t-1}\}$.}\label{jscc_fully_feedback}
\end{figure}
\begin{figure}
\centering
\includegraphics[scale=0.70]{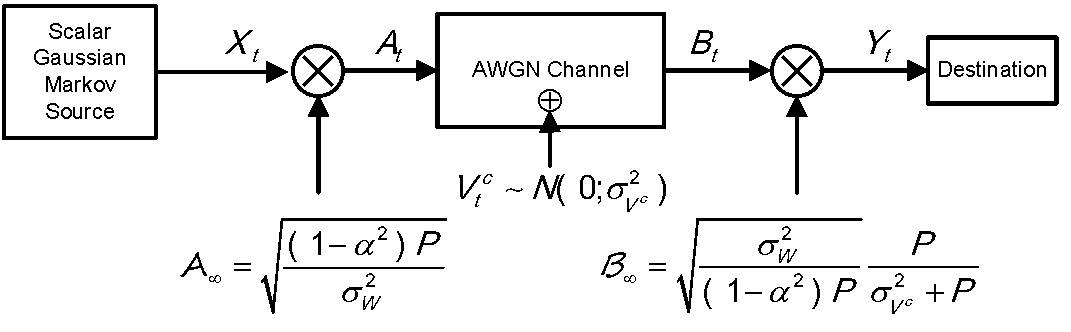}
\caption{Reduction of Fig.~\ref{communication_system}  to $p=1$, $C=1$, $N=0$, $A=\alpha, |\alpha|<1$, $B=\sigma_W$, which corresponds the  to JSCC design system of a scalar Gaussian Markov source given by (\ref{equation111}) over a memoryless AWGN channel without feedback.}\label{jscc_fully_nofeedback}
\end{figure}
Note that Fig.~\ref{jscc_fully_nofeedback} with $\alpha=0$, reduces to Goblick's JSCC design system of the {\it scalar IID Gaussian process transmitted over a memoryless AWGN channel without feedback}. Finally, we demonstrate how the realization scheme of Fig.~\ref{communication_system} recovers the Schalkwijk-Kailath's coding scheme which achieves the capacity of a memoryless AWGN channel with feedback \cite{schalkwijk-kailath1966}.\\
\noi When $\{X_t:~t=0,1\ldots,n\}$ is a vector with independent components $\{X^1_t:~t=0,1\ldots,n\},\ldots,\{X^n_t:~t=0,1\ldots,n\}$ then the analogues of Fig.~\ref{jscc_fully_feedback}, Fig.~\ref{jscc_fully_nofeedback} are easily obtained, and involve JSCC design of a vector source over a vector AWGN channel.\\
\vspace*{0.2cm}
\noi{\bf Bounds on Noncausal and Causal Codes.}~In Section~\ref{bounds_opta} we use $R^{na}(D)$ to derive the following bounds for $R(D)$, and the OPTA by causal codes \cite{neuhoff1982}, denoted by $r^{c,+}(D)$:
\begin{align}
r^{c,+}(D)\geq {R}^{na}(D) \geq{R}(D),~R_{0,n}^{na}(D)\geq{R}_{0,n}(D),~\forall{n}\geq{0}.\label{equation1m}
\end{align}
\noi Based on these bounds we evaluate the RL due to causality by using $R^{na}(D)$. This part compliments previous work by Linder and Zamir \cite{linder2006}, who showed using the results in \cite{neuhoff1982}, that at high resolution (small distortion), an optimal causal code for stationary source with finite differential entropy and square error distortion, consists of uniform quantizers followed by a sequence of entropy coders, and that the RL due to causality is given by the so-called space-filling loss of quantizers, which is at most $\frac{1}{2}\ln\big(\frac{2\pi{e}}{12}\big)\simeq{0.254}$~bits/sample.\\
\noi For arbitrary Gaussian stationary sources with memory defined by (\ref{equation1o}) and square-error distortion, we use the explicit expression of $R^{na}(D)$ (i.e., (\ref{equation1r})) in (\ref{equation1m})  to evaluate the RL due to causality. Our analysis of vector Gaussian stationary sources (\ref{equation1o}) compliment Gorbunov and Pinsker computation in \cite{pinsker-gorbunov1987,gorbunov-pinsker1991}, where it is shown (for scalar Gaussian sources), using power spectral densities that $R^{\varepsilon}(D)$ tends to Shannon's RDF $R(D)$, as $D\longrightarrow{0}$.
\vspace*{0.2cm}\\
\noi Moreover, for stationary Gaussian processes with square-error distortion, our results compliment recent contributions obtained in \cite{derpich-ostergaard2012}, on the gap between the OPTA by causal codes, nonanticipatory $\epsilon$-entropy, $R^{\varepsilon}(D)$, and classical RDF, $R(D)$. Specifically, in \cite{derpich-ostergaard2012} it is shown that for zero-mean Gaussian sources with square-error distortion (and bounded differential entropy rate), the OPTA by causal codes exceeds $R^{\varepsilon}(D)$ by less than approximately {0.254}~bits/sample. The analysis in \cite{derpich-ostergaard2012} includes a closed-form expression for $R^{\varepsilon}(D)$ only when the source is first-order Markov with square-error distortion, while for general Gaussian sources it is shown that $r^{c,+}(D)\leq{R}^{\varepsilon}(D)+\frac{1}{2}\log_2(2\pi{e})$ bits/sample, but no expression is given for $R^{\varepsilon}(D)$. Since we show  $R^{\varepsilon}(D)=R^{na}(D)$, then we can use this equivalence to evaluate the gap between the OPTA by causal codes and $R(D)$, for arbitrary multidimensional stationary Gaussian sources (partially observed), via the bound $R^{na}(D)-R(D)\leq{r}^{c,+}(D)-R(D)$. Note that the RL due to zero-delay  and noncausal codes does not exceed $\frac{1}{2}\log\frac{|\Lambda_{\infty}|}{|\Delta_{\infty}|}-R(D)$, where $R(D)$ can be computed using power spectral densities \cite{berger}.\\
\noi For the case of BSMS($p$), we show that the RL due to causality, i.e., $R^{na}(D)-R(D)$, does not exceed $H(m)-H(p)$, for the region $0\leq{D}\leq{D_c}$,~$m=1-p-D+2pD$, while for the region $D_c\leq{D}\leq\frac{1}{2}$, we compare the upper bound given in \cite{berger1977} with the upper bound obtained via $R^{na}(D)$. Our simulations and properties of $R^{na}(D)$ illustrate that the proposed upper bound on $R(D)$ based on $R^{na}(D)$ is more reliable compared to the one given  in \cite{berger1977}.
\vspace*{0.2cm}\\
\noi{\bf Noisy and Noiseless Coding Theorems.} In Section~\ref{noisy_coding_theorem_unmatched}, we show achievability of nonanticipative RDF using the JSCC design system (Fig.~\ref{realization_nonanticipative_RDF}) based on nonanticipative transmission, with respect to the excess distortion probability between the source symbols and their reproductions. That is, we show achievability based on probabilistic matching of the source and the channel with respect to excess distortion probability. For the multidimensional stationary Gaussian source, we use the JSCC design of Fig.~\ref{communication_system}, and we apply Chernoff's bound to compute the error exponent of the excess distortion probability. For the BSMS($p$), we apply uncoded transmission, and we invoke Hoeffding's inequality and large deviations theory to compute the error exponent of the excess distortion probability. In Section~\ref{sequential_causal_rdf}, we briefly point out that the classical noiseless coding theorem derived in \cite[Chapter 5]{tatikonda2000} for two dimensional sources (such as video coding applications) with per-sample or average distortion function is applicable, giving an alternative operational meaning to $R_{0,n}^{na}(D)$.
   
\section{Information Nonanticipative RDF}\label{equivalent_definitions}

\par In this section, we introduce the definition of information nonanticipative RDF, for general source/reproduction alphabets modelled by complete separable metric spaces (Polish spaces) so that our analysis holds for both finite and continuous alphabets (because we present examples for both). Then we state some properties of information nonanticipative RDF which follow directly from the general properties of directed information found in \cite{charalambous-stavrou2013aa}. 
\subsection{Information Nonanticipative RDF}\label{equivalence_of_distributions}
Let $\mathbb{N} \tri \{0,1,2,\ldots\},$ and $\mathbb{N}^n \tri \{0,1,2,\ldots,n\}.$ Introduce the source spaces $\{({\cal X}_n,{\cal B}({\cal X }_n)):n\in\mathbb{N}\}$ and the reproduction spaces $\{({\cal Y}_n,{\cal B}({\cal Y}_n)):n\in\mathbb{N}\},$ where ${\cal X}_n,{\cal Y}_n, n\in\mathbb{N}$ are Polish spaces, and ${\cal B}({\cal X}_n)$ and ${\cal B}({\cal Y}_n)$ are Borel $\sigma-$algebras of subsets of ${\cal X}_n$ and ${\cal Y}_n,$ respectively. Points in ${\cal X}^{\mathbb{N}}\tri{{\times}_{n\in\mathbb{N}}}{\cal X}_n,$ ${\cal Y}^{\mathbb{N}}\tri{\times_{n\in\mathbb{N}}}{\cal Y}_n$ are denoted by ${\bf x}\tri\{x_0,x_1,\ldots\}\in{\cal X}^{\mathbb{N}},$ ${\bf y}\tri\{y_0,y_1,\ldots\}\in{\cal Y}^{\mathbb{N}},$ respectively, while their restrictions to finite coordinates for any $n\in\mathbb{N}$ are denoted by $x^n\tri\{x_0,x_1,\ldots,x_n\}\in{\cal X}_{0,n},$ $y^n\tri\{y_0,y_1,\ldots,y_n\}\in{\cal Y}_{0,n}$. Let ${\cal B}({\cal X}^{\mathbb{N}})\tri\odot_{i\in\mathbb{N}}{\cal B}({\cal X}_i)$ denote the $\sigma-$algebra on ${\cal X}^{\mathbb{N}}$ generated by cylinder sets $\{{\bf x}=(x_0,x_1,\ldots)\in{\cal X}^{\mathbb{N}}:x_0\in{A}_0,x_1\in{A}_1,\ldots,x_n\in{A}_n\}, A_i\in{\cal B}({\cal X}_i), i=0,1,\ldots,n, n\in\mathbb{N}$, and similarly for ${\cal B}({\cal Y}^{\mathbb{N}})\tri\odot_{i\in\mathbb{N}}{\cal B}({\cal Y}_i)$. Thus, ${\cal B}({\cal X}_{0,n})$ and ${\cal B}({\cal Y}_{0,n})$ denote the $\sigma-$algebras of cylinder sets in ${\cal X}^{\mathbb{N}}$ and ${\cal Y}^{\mathbb{N}},$ respectively, with bases over $A_i\in{\cal B}({\cal X}_i)$, and $B_i\in{\cal B}({\cal Y}_i),~i=0,1,\ldots,n$, respectively. The set of probability distributions on any measurable space $({\cal X},{\cal B}({\cal X}))$ is denoted by ${\cal M}({\cal X})$.
\vspace*{0.2cm}\\
\noi{\bf Source Distribution.} The source is a collection of conditional probability distributions $\{P_{X_n|X^{n-1}}(dx_n|x^{n-1}):n\in\mathbb{N}\}$. For a cylinder set $B\tri\big\{{\bf x}\in{\cal X}^{\mathbb{N}}:x_0\in{B_0},x_1\in{B_1},\ldots,x_n\in{B_n}\big\},~B_i\in{\cal B}({\cal X }_i),~i=0,\ldots,n$, we can define the probability distribution ${P}_{X^n}(B_{0,n})$ on ${\cal B}({\cal X}_{0,n})$ by
\begin{align}
{P}_{X^n}(B_{0,n})\tri\int_{B_0}P_{X_0}(dx_0)\ldots\int_{B_n}P_{X_n|X^{n-1}}(dx_n|x^{n-1}),~B_{0,n}=\times_{i=0}^n{B_i}.\label{equation2}
\end{align}
For each $n\in\mathbb{N}$, we use the notation ${\cal Q}_n({\cal X}_n|{\cal X}_{0,n-1})\tri\{P_{X_n|X^{n-1}}(\cdot|x^{n-1})\in{\cal M}({\cal X}_n):~x^{n-1}\in{\cal X}_{0,n-1}\}$ to emphasize that $P_{X_n|X^{n-1}}(\cdot|x^{n-1})$ is a probability distribution on ${\cal X}_n$ for fixed $x^{n-1}\in{\cal X}_{0,n-1}$ and $P_{X^n|X^{n-1}}(dx_n|\cdot)$ is a measurable function on $x^{n-1}\in{\cal X}_{0,n-1}$ for fixed $dx_n$.
\vspace*{0.2cm}\\
\noi{\bf Reproduction Distribution.} The reproduction distribution is a collection of conditional probability distributions $\{Q_{Y_n|Y^{n-1},X^n}$ \\$(dy_n|y^{n-1},x^n):n\in\mathbb{N}\}$. By our notation, for each $n\in\mathbb{N}$, $Q_{Y_n|Y^{n-1},X^n}(\cdot|\cdot,\cdot)\in{\cal Q}_n({\cal Y}_n|{\cal Y}_{0,n-1}\times{\cal X}_{0,n})$. For a cylinder set $C\tri\Big\{{\bf y}\in{\cal Y}^{\mathbb{N}}:y_0{\in}C_0,y_1{\in}C_1,\ldots$, $y_n{\in}C_n\Big\}$,~$C_i\in{\cal B}({\cal Y}_i),~i=0,\ldots,n$,~$C_{0,n}=\times_{i=0}^n{C_i}$, we define the family of conditional probability distributions ${\overrightarrow{Q}}_{Y^n|X^n}(C_{0,n}|x^n)$ on ${\cal B}({\cal Y}_{0,n})$ by 
\begin{align}
{\overrightarrow{Q}}_{Y^n|X^n}(C_{0,n}|x^n)&\tri\int_{C_0}Q_{Y_0|X_0}(dy_0|x_0)\ldots\int_{C_n}Q_{Y_n|Y^{n-1},X^n}(dy_n|y^{n-1},x^n).\label{equation4}
\end{align}
Throughout the paper, we assume with abuse of notation that $Q_{Y_0|X_0}(dy_0|x_0)\equiv{Q}_{Y_0|X_0,Y^{-1}}(dy_0|x_0,y^{-1})$ and that either $P_{Y^{-1}}(dy^{-1})$ is fixed or $Y^{-1}=y^{-1}$ is fixed.\\
\noi We denote the family of probability distributions (\ref{equation4}) by the set 
\begin{align}
\overrightarrow{\cal Q}({\cal Y}_{0,n}|{\cal X}_{0,n})\tri\Big\{\overrightarrow{Q}_{Y^n|X^n}(\cdot|x^n)\in{\cal M}({\cal Y}_{0,n}):~x^n\in{\cal X}_{0,n},~\overrightarrow{Q}_{Y^n|X^n}(\cdot|{x}^n)~\mbox{is defined by (\ref{equation4})}\Big\}.\nonumber
\end{align}
\noi Thus, any element $\overrightarrow{Q}_{Y^n|X^n}(\cdot|\cdot)\in\overrightarrow{\cal Q}({\cal Y}_{0,n}|{\cal X}_{0,n})$ is represented by (\ref{equation4}).\\
\noi The following theorem is derived in \cite{charalambous-stavrou2013aa}.
\begin{theorem}\cite{charalambous-stavrou2013aa}$($Convexity$)$\label{convexity_of_sets}{\ \\}
The set of probability distributions $\overrightarrow{Q}_{Y^n|X^n}(\cdot|x^n)$ (for a fixed $x^n\in{\cal X}_{0,n}$) defined by (\ref{equation4}) is a convex subset of ${\cal M}({\cal Y}_{0,n})$.
\end{theorem}
\vspace*{0.2cm}
\noi Theorem~\ref{convexity_of_sets} states that for any $\lambda\in[0,1]$, and $\overrightarrow{Q}^1_{Y^n|X^n}(\cdot|{x}^n)$, $\overrightarrow{Q}_{Y^n|X^n}^2(\cdot|{x}^n)$ two probability measures on $({\cal Y}_{0,n},{\cal B}({\cal Y}_{0,n}))$ of the form (\ref{equation4}), then $\lambda\overrightarrow{Q}^1_{Y^n|X^n}(\cdot|{x}^n)+(1-\lambda)\overrightarrow{Q}^2_{Y^n|X^n}(\cdot|{x}^n)$ is also a probability measure on $({\cal Y}_{0,n},{\cal B}({\cal Y}_{0,n}))$ of the form (\ref{equation4}).  
\vspace*{0.2cm}\\
Given the source distribution ${P}_{X^n}(\cdot)\in{\cal M}({\cal X}_{0,n})$ and reproduction distribution $\overrightarrow{Q}_{Y^n|X^n}(\cdot|\cdot)\in\overrightarrow{\cal Q}({\cal Y}_{0,n}|{\cal X}_{0,n})$ define the following measures.\\
{\bf P1}: The joint distribution on ${\cal X}_{0,n}\times{\cal Y}_{0,n}$ defined uniquely for $A_i\in{\cal B}({\cal X}_i)$,~$B_i\in{\cal B}({\cal Y}_i)$,~$\forall{i}\in\mathbb{N}^n$, by
\begin{align}
P_{X^n,Y^n}(dx^n,dy^n)&\tri(P_{X^n}\otimes{\overrightarrow Q}_{Y^n|X^n})\big(\times^n_{i=0}(A_i{\times}B_i)\big)\nonumber\\
&\tri\int_{A_0}P_{X_0}(dx_0)\int_{B_0}Q_{Y_0|X_0}(dy_0|x_0)\ldots
\int_{A_n}P_{X_n|X^{n-1}}(dx_n|x^{n-1})\int_{B_n}Q_{Y_n|Y^{n-1},X^n}(dy_n|y^{n-1},x^n).\nonumber
\end{align}
{\bf P2}: The marginal distributions on ${\cal Y}_{0,n}$ defined uniquely for $B_i\in{\cal B}({\cal Y}_i)$, $\forall{i}\in\mathbb{N}^n$, by
\begin{align}
P_{Y^n}(\times^n_{i=0}B_i)&\tri(P_{X^n}\otimes{\overrightarrow Q}_{Y^n|X^n})\big(\times^n_{i=0}({\cal X}_i\times{B}_i)\big)\nonumber.
\end{align}
{\bf P3}: The product probability distribution ${\overrightarrow\Pi}_{0,n}:{\cal B}({\cal X}_{0,n})\odot{\cal B}({\cal Y}_{0,n})\longmapsto[0,1]$ defined uniquely for $A_i\in{\cal B}({\cal X}_i)$,~$B_i\in{\cal B}({\cal Y}_i)$,~$\forall{i}\in\mathbb{N}^n$, by
\begin{align}
{\overrightarrow\Pi}_{0,n}\big(\times^n_{i=0}(A_i{\times}B_i)\big)&\tri(P_{X^n}\times{P}_{Y^n})\big(\times^n_{i=0}(A_i{\times}B_i)\big)\nonumber\\
&=\int_{A_0}P_{X_0}(dx_0)\int_{B_0}P_{Y_0}(dy_0)\ldots
\int_{A_n}P_{X_n|X^{n-1}}(dx_n|x^{n-1})\int_{B_n}P_{Y_n|Y^{n-1}}(dy_n|y^{n-1})\nonumber
\end{align}
where the notation is $P_{Y_0}(dy_0)\equiv{P}_{Y_0|Y^{-1}}(dy_0|y^{-1})$.\\
\noi The information theoretic measure associated with nonanticipative RDF is a special case of directed information\footnote{Directed information corresponds to $\{P_{X_n|X^{n-1}}(dx_i|x^{n-1}):~n\in\mathbb{N}\}$ and $P_{X^n}(\cdot)$ replaced by $\{P_{X_n|X^{n-1},Y^{n-1}}(dx_n|x^{n-1},y^{n-1}):n\in\mathbb{N}\}$ and $\overleftarrow{P}_{X^n|Y^{n-1}}(dx^n|y^{n-1})\tri\otimes_{i=0}^n{P_{X_i|X^{i-1},Y^{i-1}}(dx_|x^{i-1},y^{i-1})}$, respectively, in the construction of measures {\bf P1}-{\bf P3}, and (\ref{equation33}).} \cite{tatikonda2000}, defined via relative entropy $\mathbb{D}(\cdot||\cdot)$ as follows.
\begin{align}
I_{P_{X^n}}(X^n\rightarrow{Y}^n)&\tri\mathbb{D}({P}_{X^n} \otimes {\overrightarrow Q}_{Y^n|X^n}||{\overrightarrow\Pi}_{0,n})\label{equation33}\\
&=\int_{{\cal X}_{0,n}\times{\cal Y}_{0,n}} \log \Big( \frac{{\overrightarrow Q}_{Y^n|X^n}(dy^n|x^n)}{P_{Y^n}(dy^n)}\Big)({P}_{X^n}\otimes {\overrightarrow Q}_{Y^n|X^n})(dx^n,dy^n)\label{equation203}\\
&\equiv{\mathbb{I}}_{X^n\rightarrow{Y^n}}({P}_{X^n}, {\overrightarrow Q}_{Y^n|X^n}).\label{equation7a}
\end{align}
The RHS of (\ref{equation203}) is obtained by using the chain rule of relative entropy \cite{dupuis-ellis97},\cite[Chapter 3]{pinsker-book}. In (\ref{equation7a}) we use the notation ${\mathbb{I}}_{X^n\rightarrow{Y^n}}(\cdot,\cdot)$ to indicate the functional dependence of $I_{P_{X^n}}(X^n\rightarrow{Y^n})$ on $\{P_{X^n}, {\overrightarrow Q}_{Y^n|X^n}\}$.\\
\noi We recall the following convexity result derived in \cite{charalambous-stavrou2013aa} for subsequent use.
\begin{theorem}\cite{charalambous-stavrou2013aa}$(${Convexity of information nonanticipative RDF}$)$\label{convexity_properties}{\ \\}
${\mathbb I}_{X^n\rightarrow{Y^n}}(P_{X^n},\overrightarrow{Q}_{Y^n|X^n})$ is a convex functional of $\overrightarrow{Q}_{Y^n|X^n}({\cdot|x^n})\in{\cal M}({\cal Y}_{0,n})$ for a fixed $P_{X^n}(\cdot)\in{\cal M}({\cal X}_{0,n})$, and a concave functional of $P_{X^n}(\cdot)\in{\cal M}({\cal X}_{0,n})$ for a fixed $\overrightarrow{Q}_{Y^n|X^n}({\cdot|x^n})\in{\cal M}({\cal Y}_{0,n})$.
\end{theorem}
\noi For each $n\in\mathbb{N}$, let $d_{0,n}:{\cal X}_{0,n}\times{\cal Y}_{0,n}\longmapsto[0,\infty]$ be a measurable distortion function. The fidelity of reproduction of $y^n\in{\cal Y}_{0,n}$ by $x^n\in{\cal X}_{0,n}$ is defined by the set of conditional distributions
\begin{align}
\overrightarrow{\cal Q}_{0,n}(D)&\tri\Big\{\overrightarrow{Q}_{Y^n|X^n}(\cdot|x^n)\in{\cal M}({\cal Y}_{0,n}):~\frac{1}{n+1}\int_{{\cal X}_{0,n}\times{\cal Y}_{0,n}} d_{0,n}({x^n},{y^n})(P_{X^n}\otimes\overrightarrow{Q}_{Y^n|X^n})(d{x}^{n},d{y}^{n})\leq D\Big\}\label{eq2}
\end{align}
for some $D\geq{0}$. Denote by $\overrightarrow{\cal Q}_{0,\infty}(D)$ the corresponding set in (\ref{eq2}), when the fidelity is replaced by $\lim_{n\longrightarrow\infty}\frac{1}{n+1}\int_{{\cal X}_{0,n}\times{\cal Y}_{0,n}}$\\ $d_{0,n}(x^n,y^n)(P_{X^n}\otimes\overrightarrow{Q}_{Y^n|X^n})(d{x}^{n},d{y}^{n})\leq D$.\\
\noi The information nonanticipative RDF is defined as follows.
\begin{definition}$(${Information nonanticipative RDF}$)${\ \\}\label{nonanticipative_rdf1}
\begin{itemize} 
\item[(1)] The information nonanticipative RDF is defined by 
\begin{align}
{R}^{na}_{0,n}(D) \tri  \inf_{{\overrightarrow{Q}_{Y^n|X^n}(\cdot|x^n)\in\overrightarrow{\cal Q}_{0,n}(D)}}\mathbb{I}_{X^n\rightarrow{Y^n}}(P_{X^n},{\overrightarrow Q}_{Y^n|X^n})
\label{ex12}
\end{align}
provided the infimum over $\overrightarrow{\cal Q}_{0,n}(D)$ in (\ref{ex12}) exists; if not we set ${R}^{na}_{0,n}(D)=+\infty$.\\
\item[(2)] The information nonanticipative RDF rate is defined by 
\begin{align}
{R}^{na}(D)=\lim_{n\longrightarrow\infty}\frac{1}{n+1}{R}^{na}_{0,n}(D)\label{equation22}
\end{align}
provided the limit on the RHS of (\ref{equation22}) exists; if the infimum over $\overrightarrow{\cal Q}_{0,n}(D)$ in (\ref{ex12}) does not exist then we set ${R}^{na}(D)=+\infty$.
\end{itemize}
In addition, define
\begin{align}
\overrightharpoonup{R}^{na}(D)\tri\inf_{{\overrightarrow{Q}_{Y^\infty|X^\infty}(\cdot|x^{\infty})\in\overrightarrow{\cal Q}_{0,\infty}(D)}}\lim_{n\longrightarrow\infty}\frac{1}{n+1}\mathbb{I}_{X^n\rightarrow{Y^n}}(P_{X^n},{\overrightarrow Q}_{Y^n|X^n})\geq{R}^{na}(D).\label{equation64}
\end{align}
\end{definition}
\noi Since, in general, $\overrightharpoonup{R}^{na}(D)\geq{R}^{na}(D)$, then, $R^{na}(D)$ is more natural than $\overrightharpoonup{R}^{na}(D)$.  By analogy with the definition of classical RDF, one may assume that $\{(X_i,Y_i):~i=0,1,\ldots\}$ is jointly stationary and ergodic process or $\frac{1}{n+1}\log\Big(\frac{\overrightarrow{Q}_{Y^n|X^n}(\cdot|x^n)}{P_{Y^n}(\cdot)}\Big)(y^n)$ is information stable. However, we do not know \'a priori whether the joint process $\{(X_i,Y_i):~i=0,1,\ldots\}$ is stationary.


\section{Equivalence of Information Nonanticipative RDF and Nonanticipatory $\epsilon$-Entropy}\label{relations_epsilon_entropy_limit}

\par In this section, we first show equivalence of the information nonanticipative RDF (see Definition~\ref{nonanticipative_rdf1}) to Gorbunov and Pinsker's \cite{gorbunov-pinsker} nonanticipatory $\epsilon-$entropy defined by (\ref{equation1kk}), (\ref{equation1k}), respectively. Then, we consider consistent stationary sources as defined by Gorbunov and Pinsker \cite{gorbunov-pinsker}, and we use certain results from \cite{charalambous-stavrou-ahmed2014a}, to establish equality of the limiting expressions in (\ref{equation705}), finiteness of $R^{na}(D)$, and that for stationary sources the infimum over $\overrightarrow{Q}_{Y^n|X^n}(\cdot|x^n)\in\overrightarrow{\cal Q}_{0,n}(D)$ is achieved, and it is realizable by stationary source-reproduction pairs $\{(X_n,Y_n):n\in\mathbb{N}\}$.

\subsection{Equivalence of Nonanticipative RDF and Nonanticipatory $\epsilon-$Entropy}\label{equivalence_nrdf_epsilon_entropy}

\par For a given a source $P_{X^n}(\cdot)\in{\cal M}({\cal X}_{0,n})$ and a reproduction $P_{Y^n|X^n}(\cdot|\cdot)\in{\cal Q}_{0,n}(D)\subset{\cal Q}({\cal Y}_{0,n}|{\cal X}_{0,n})$, Gorbunov and Pinsker \cite{gorbunov-pinsker} restricted the fidelity set of classical RDF, $R_{0,n}(D)$, ${\cal Q}_{0,n}(D)$ defined by (\ref{equation1a}), to those reproduction distributions which satisfy the following MC.
\begin{align}
X_{n+1}^\infty\leftrightarrow{X^n}\leftrightarrow{Y^n}\Longleftrightarrow{P}_{Y^n|X^{\infty}}(dy^n|x^{\infty})={P}_{Y^n|X^n}(dy^n|x^n),~n=0,1,\ldots.\label{equation70}
\end{align}
\noi Then they introduced the nonanticipatory $\epsilon$-entropy and nonanticipatory message generation rate defined by (\ref{equation1kk}), (\ref{equation1k}), respectively. In addition, they defined
\begin{align}
\overrightharpoonup{R}^{\varepsilon}(D)\tri\inf_{\substack{P_{Y^\infty|X^\infty}(\cdot|x^{\infty})\in{\cal Q}_{0,\infty}(D):\\X^\infty_{i+1}\leftrightarrow{X^i}\leftrightarrow{Y^i},~i=0,1,\ldots}}\lim_{n\longrightarrow\infty}\frac{1}{n+1}I(X^n;Y^n)\geq{R}^{\varepsilon}(D).\label{equation23}
\end{align}
\noi The MC constraint (\ref{equation70}) is a probabilistic version of a deterministic causal reproduction coder (cascade of an encoder-decoder (ED)) \cite{neuhoff1982}, defined as follows.
\begin{definition}\cite{neuhoff1982}$(${Causal reproduction coder}$)${\ \\}\label{causal_reproduction_coder}
A reproduction coder $f_i:{\cal X}_{0,n}\longmapsto{\cal Y}_i$,~$\forall{i}=0,1,\ldots,n$, is called causal if the map $x^n\longmapsto{f}_i(x^n)$ is measurable $\forall{i}\in\mathbb{N}^n$ and
\begin{align*}
f_i(x^n)=f_i(\hat{x}^n)~~\mbox{whenever}~~x^i=\hat{x}^i,~\forall{n}\geq{i},~n\in\mathbb{N}.
\end{align*}
\end{definition}
\noi Since the class of randomized reproduction coders embeds deterministic causal reproduction coders, then probabilistically, a reproduction coder is causal if and only if the following MC holds $X_{i+1}^\infty\leftrightarrow{X}^i\leftrightarrow{Y}_i$, $\forall{i}\in\mathbb{N}$. By (\ref{equation1kk}), nonanticipatory $\epsilon$-entropy, $R^{\varepsilon}_{0,n}(D)$, imposes a probabilistic causality (nonanticipative) constraint on the optimal reproduction distribution.
\vspace*{0.2cm}\\
Next, we show that $R_{0,n}^{\varepsilon}(D)=R_{0,n}^{na}(D)$ and $R^\varepsilon(D)=R^{na}(D)$, by invoking the following equivalent statements of MCs.
\begin{lemma}$(${Equivalent nonanticipative statements}$)${\ \\}\label{equivalent_statements}
The following statements are equivalent.
\begin{description}
\item[{\bf MC1}:] $P_{Y^n|X^n}(dy^n|x^n)={\overrightarrow Q}_{Y^n|X^n}(dy^n|x^n)=\otimes_{i=0}^n{P}_{Y_i|Y^{i-1},X^i}(dy_i|y^{i-1},x^i)$,~~$\forall{n}\in\mathbb{N}$;

\item[{\bf MC2}:]  $Y_i \leftrightarrow (X^i, Y^{i-1}) \leftrightarrow (X_{i+1}, X_{i+2}, \ldots, X_n)$ forms a MC,~for each $i=0,1,\ldots, n-1$,~$\forall{n}\in\mathbb{N}$;

\item[{\bf MC3}:] $Y^i \leftrightarrow X^i \leftrightarrow X_{i+1}$ forms a MC,~for each  $i=0,1,\ldots, n-1$,~$\forall{n}\in\mathbb{N}$;

\item[{\bf MC4}:] $X_{i+1}^n\leftrightarrow{X^i}\leftrightarrow{Y^i}$ forms a MC,~for each $i=0,1,\ldots, n-1$,~$\forall{n}\in\mathbb{N}$. 
\end{description}
\end{lemma}
\begin{proof}
See Appendix~\ref{appendix_equivalent_statements}.
\end{proof}
\vspace*{0.2cm}
\noi We note that {\bf MC3} of Lemma~\ref{equivalent_statements} is precisely Granger's definition of temporal causality \cite{solo2008}, which is used in econometrics to unravel complex relations between macroeconomic variables from time series observations. It is also applied in bioengineering \cite{solo2008,lin-hara-solo-vangel-belliveau-stufflebeam-hamalainen2009}, and more recently in neuroimaging to infer that $\{Y_n:n\in\mathbb{N}\}$ does not cause $\{X_n:n\in\mathbb{N}\}$. Note that \cite{pearl2009} refers to {\bf MC4} as the ``weak union" property of conditional independence.\\
\noi In the next theorem, we utilize Lemma~\ref{equivalent_statements},  and more specifically, the fact that {\bf MC4} is equivalent to {\bf MC2} and {\bf MC1}, to show that the extremum of the nonanticipatory $\epsilon$-entropy, $R^{\varepsilon}_{0,n}(D)$, defined by (\ref{equation1k}), is equivalent to the extremum of nonanticipative RDF, $R^{na}_{0,n}(D)$, defined by (\ref{ex12}).
\begin{theorem}$($Equivalence of $R^{na}_{0,n}(D)$ and {$R^{\varepsilon}_{0,n}(D)$}$)$\label{equivalent_rdf}{\ \\}
The nonanticipative RDF and nonanticipatory $\epsilon$-entropy are equivalent notions, i.e., $R^{na}_{0,n}(D)=R^{\varepsilon}_{0,n}(D),~~\forall{n}\in\mathbb{N}^n$.
\end{theorem}
\begin{proof}  The derivation follows directly from  Lemma~\ref{equivalent_statements}.  
\end{proof}


\subsection{Infinite Horizon}\label{infinite_horizon}
\par In this section, we recall the finite horizon existence result derived in\cite[Theorem III.4]{charalambous-stavrou-ahmed2014a} to investigate existence of the information nonanticipative RDF rate, and the validity of interchanging the limit and infimum operations in (\ref{equation705}). One may also consider the two-sided definition of nonanticipative RDF by replacing $R_{0,n}^{na}(D)$ by $R^{na}_{n_1,n_2}(D)$, $n_2>n_1$, in which case, the rate is defined by $\lim_{n_2-n_{1}\longrightarrow\infty}\frac{1}{n_2-n_1+1}R^{na}_{n_1,n_2}(D)$, provided the limit exists. However, the rate is defined if and only if the following limit is defined for some $n_1$: $\lim_{n_2\longrightarrow\infty}\frac{1}{n_2-n_{1}}R^{na}_{n_1,n_2}(D)$. Hence, without loss of generality, we let $n_1=0$.\\
First, we prove the following inequality.
\begin{lemma}{\ \\}\label{useful_inequality}
The following holds.
\begin{align}
R^{na}(D)\leq \overrightharpoonup{R}^{na}(D)\tri\inf_{\overrightarrow{Q}_{Y^\infty|X^\infty}(\cdot|x^\infty)\in\overrightarrow{\cal Q}_{0,\infty}({D})}\lim_{n\longrightarrow\infty}\frac{1}{n+1}\mathbb{I}_{X^n\rightarrow{Y^n}}(P_{X^n},\overrightarrow{Q}_{Y^n|X^n}).\label{equation62ii}
\end{align}
\end{lemma}
\begin{proof}
If the infimum in (\ref{ex12}) does not exist, there is nothing to prove. Hence, suppose this infimum exists. By definition we have 
\begin{align*}
R^{na}_{0,n}(D)\leq\mathbb{I}_{X^n\rightarrow{Y^n}}(P_{X^n},\overrightarrow{Q}_{Y^n|X^n}),~\forall\overrightarrow{Q}_{Y^n|X^n}(\cdot|x^n)\in\overrightarrow{\cal Q}_{0,n}(D)
\end{align*}
and hence, by taking the limit on both sides we obtain
\begin{align*}
\lim_{n\longrightarrow\infty}\frac{1}{n+1}R^{na}_{0,n}(D)\leq\lim_{n\longrightarrow\infty}\frac{1}{n+1}\mathbb{I}_{X^n\rightarrow{Y^n}}(P_{X^n},\overrightarrow{Q}_{Y^n|X^n}),~\forall\overrightarrow{Q}_{Y^\infty|X^\infty}(\cdot|x^\infty)\in\overrightarrow{\cal Q}_{0,\infty}(D).
\end{align*}
Taking the infimum over $\overrightarrow{Q}_{Y^\infty|X^\infty}(\cdot|x^\infty)\in\overrightarrow{\cal Q}_{0,\infty}(D)$ we obtain (\ref{equation62ii}). This shows that $R^{na}(D)\leq\overrightharpoonup{R}^{na}(D)$.
\end{proof}

\noi Since by Theorem~\ref{equivalent_rdf}, $R^{na}_{0,n}(D)=R^{\varepsilon}_{0,n}(D)$,~$\forall{n}\in\mathbb{N}^n$, and since  existence (under certain assumptions) of solution to the information nonanticipative RDF (\ref{ex12}) is shown in \cite[Theorem III.4]{charalambous-stavrou-ahmed2014a}, all technical results derived in \cite[Theorems 1-4]{gorbunov-pinsker} are directly applicable to $R_{0,n}^{na}(D)$ and its rate, $R^{na}(D)$, without assuming finiteness of $R^{na}_{0,n}(D)$ for some $n$, as in \cite{gorbunov-pinsker}. Next, we summarize these results in order to conclude that the limit $\lim_{n\longrightarrow\infty}\frac{1}{n+1}{R}^{na}(D)$ is finite.
\vspace*{0.2cm}\\
\noi The following theorem is a direct consequence of Lemma~\ref{useful_inequality}, \cite[Theorem III.4]{charalambous-stavrou-ahmed2014a}, and Gorbunov-Pinsker \cite[Theorem 2]{gorbunov-pinsker}.
\begin{theorem}$(${Limits}$)$\label{fundamental_limits}{\ \\}
Suppose the conditions of \cite[Theorem III.4]{charalambous-stavrou-ahmed2014a} holds and the source is stationary as defined in \cite[Theorem 2]{gorbunov-pinsker}.\\ 
Then the following hold.
\begin{align}
{R}^{na}(D)=\lim_{n\longrightarrow\infty}\inf_{\overrightarrow{Q}_{Y^n|X^n}(\cdot|x^n)\in\overrightarrow{\cal Q}_{0,n}(D)}\frac{1}{n+1}\mathbb{I}_{X^n\rightarrow{Y^n}}(P_{X^n},\overrightarrow{Q}_{Y^n|X^n})<\infty\label{equation710}
\end{align}
e.g., the limit exists and it is finite. Moreover,
\begin{align}
{R}^{na}(D)=\overrightharpoonup{R}^{na}(D)\equiv\inf_{\overrightarrow{Q}_{Y^\infty|X^\infty}(\cdot|x^\infty)\in\overrightarrow{\cal Q}_{0,\infty}({D})}\lim_{n\longrightarrow\infty}\frac{1}{n+1}\mathbb{I}_{X^n\rightarrow{Y^n}}(P_{X^n},\overrightarrow{Q}_{Y^n|X^n}).\label{equation133i}
\end{align}
\end{theorem}
\begin{proof}
The derivation utilizes the existence result derived in \cite[Theorem III.4]{charalambous-stavrou-ahmed2014a}, and the subadditivity of $R^{na}_{0,n}(D)$, that is, $R^{na}_{0,n}(D)\leq{R}_{0,k}^{na}(D)+R^{na}_{k+1,n}(D),~0<k<n$. Since by Theorem~\ref{equivalent_rdf} we have $R^{na}_{0,n}(D)=R^{\varepsilon}_{0,n}(D)$, and by \cite[Lemma 1]{gorbunov-pinsker},  $R^{\varepsilon}_{0,n}(D)$ is subadditive, then $R^{na}_{0,n}(D)$ is also subadditive. Moreover, under the conditions of \cite[Theorem III.4]{charalambous-stavrou-ahmed2014a} we know that $R^{na}_{0,n}(D)$ is finite for any finite $n\in\mathbb{N}$. Utilizing this and the subadditivity of $R^{na}_{0,n}(D)$ we deduce that the limit in (\ref{equation710}) exist and it is finite. (\ref{equation133i}) follows from \cite[Theorem 2]{gorbunov-pinsker}, the stationarity of the source, and Theorem~\ref{equivalent_rdf}, which states that $R^{na}_{0,n}(D)=R^{\varepsilon}_{0,n}(D)$.
\end{proof}
\noi Often, in the derivation of classical RDF rate it is assumed that the process $\{(X_n,Y_n):~n\in\mathbb{N}\}$ is jointly stationary. This is not very natural because one does not know \'a priori whether the reproduction process is stationary. The next theorem utilizes \cite[Theorem III.4]{charalambous-stavrou-ahmed2014a} and the main theorems in \cite{gorbunov-pinsker} to establish that the infimum of the information nonanticipative RDF rate is achieved by a stationary reproduction distribution, which implies the process $\{(X_n,Y_n):~n\in\mathbb{N}\}$ is jointly stationary.
\begin{theorem}$(${Stationarity of the optimal reproduction distribution}$)$\label{joint_stationarity}{\ \\}
Suppose the conditions of \cite[Theorem III.4]{charalambous-stavrou-ahmed2014a} hold, the source $\{X_n:n\in\mathbb{N}\}$ is stationary and consistent, and the fidelity set is shift invariant as defined in \cite{gorbunov-pinsker}.\\
Then, the infimum in (\ref{equation133i}) is achieved by some $\overrightarrow{Q}^*_{Y^n|X^n}(\cdot|x^n)\in\overrightarrow{\cal Q}_{0,n}(D)$ such that the source-reproduction pair $\{(X_n,Y_n):~n\in\mathbb{N}\}$ is jointly stationary.
\end{theorem}
\begin{proof}
By \cite[Theorem III.4]{charalambous-stavrou-ahmed2014a} the infimum is achieved, and by Theorem~\ref{equivalent_rdf}, we have $R^{na}_{0,n}(D)=R^{\varepsilon}_{0,n}(D)$. By invoking \cite[Theorem 4]{gorbunov-pinsker} we establish the claim of stationarity of the joint process.
\end{proof}
\noi Therefore, by utilizing \cite[Theorem III.4]{charalambous-stavrou-ahmed2014a} we have strengthened the results described in \cite[Theorems 1-4]{gorbunov-pinsker}, which are based on the assumption that $R^{\varepsilon}_{0,n}(D)$ is finite for some finite $n\in\mathbb{N}$ and the infimum is achieved. 

\section{Optimal Reproduction of Nonanticipative RDF, Properties, and Examples}\label{optimal_reproduction_properties_examples}

\par In this section, we use the closed form expression of the optimal stationary reproduction corresponding to $R^{na}(D)$ \cite[Section IV]{charalambous-stavrou-ahmed2014a} to derive an alternative characterization of the solution to the nonanticipative RDF, $R^{na}(D)$, and to introduce certain of its properties. Then we apply them to compute $R^{na}(D)$ for the two running examples, the BSMS($p$) and the general multidimensional partially observed stationary Gaussian-Markov source defined by (\ref{equation1o}).\\
\noi The main assumption we impose is the following. 
\begin{assumption}\label{stationarity}{\ \\}
The $(n+1)$-fold convolution of causal conditional distribution $\overrightarrow{Q}_{Y^n|X^n}(\cdot|x^n)=\otimes^n_{i=0}Q_{Y_i|Y^{i-1},X^i}(\cdot|y^{i-1},x^i)$ which achieves the infimum of $R^{na}_{0,n}(D)$, is a convolution of stationary conditional distributions.
\end{assumption}

\subsection{Stationary optimal reproduction distribution}\label{optimal_reproduction_distribution_section}

\par Clearly, (\ref{ex12}) is a constrained problem, which is convex due to the convexity of the fidelity set, and the convexity of $\mathbb{I}_{X^n\rightarrow{Y^n}}(P_{X^n},\cdot)$, as a functional of $\overrightarrow{Q}_{Y^n|X^n}(\cdot|\cdot)\in\overrightarrow{\cal Q}({\cal Y}_{0,n}|{\cal X}_{0,n})$, (see Theorem~\ref{convexity_properties}). Therefore, we apply duality theory \cite{dluenberger69} to convert the constrained problem into an unconstrained problem using Lagrange multipliers, and then we verify the equivalence of the constrained and unconstrained problems. This procedure is done in \cite[Theorem IV.3]{charalambous-stavrou-ahmed2014a} hence, it is omitted; instead we state the main theorem.
\begin{theorem} $(${Optimal stationary reproduction distribution}$)$\label{th6}{\ \\}
Suppose Assumption~\ref{stationarity} and the conditions of \cite[Theorem IV.4]{charalambous-stavrou-ahmed2014a} hold, and $d_{0,n}=\sum_{i=0}^n\rho(T^ix^n,T^iy^n)$, $T^ix^n\subset\{x_0,\ldots,x_i\}$, $T^iy^n\subset\{y_0,\ldots,y_i\}$,~$i=0,\ldots,n$.\\
The following hold.
\begin{description}
\item[{\bf(1)}] The infimum is attained at $\overrightarrow{Q}^*_{Y^n|X^n}(\cdot|x^n)\in\overrightarrow{\cal Q}_{0,n}(D)$ given by\footnote{Due to stationarity assumption $P_{Y_i|Y^{i-1}}^*(\cdot|y^{i-1})=P^*(\cdot|y^{i-1})$ and $Q^*_{Y_i|Y^{i-1},X^i}(\cdot|y^{i-1},x^i)=Q^*(\cdot|y^{i-1},x^i)$.}
\begin{align}
&\overrightarrow{Q}^*_{Y^n|X^n}(\times_{i=0}^n{B}_i|x^n)=\int_{B_0} {Q}^*_{Y_0|X_0}(dy_0|x_0)\ldots \int_{B_n} {Q}^*_{Y_n|Y^{n-1},X^n}(dy_n|y^{n-1}, x^n)   \label{ex144}
\end{align}
where 
\begin{align}
{Q}^*_{Y_i|Y^{i-1},X^i}(dy_i|y^{i-1},x^i)=\frac{{e}^{s\rho(T^ix^n,T^iy^n)}P^*_{Y_i|Y^{i-1}}(dy_i|y^{i-1})}{\int_{{\cal Y}_i} e^{s \rho(T^ix^n,T^iy^n)} P^*_{Y_i|Y^{i-1}}(dy_i|y^{i-1})},~~i=0,1,\ldots,n, \: s\leq{0}\label{ex14}
\end{align}
and $P^*_{Y_i|Y^{i-1}}(\cdot|\cdot)\in {\cal Q}_i({\cal Y}_i|{\cal Y}_{0,{i-1}})$,~$i=0,1,\ldots,n$. 
\item[{\bf(2)}] The information nonanticipative RDF is given by
\begin{align}
{R}^{na}_{0,n}(D)&=sD(n+1)-\sum_{i=0}^n\int_{{\cal X}_{0,i}\times{\cal Y}_{0,n-1}}\log \Big( \int_{{\cal Y}_{i}} e^{s\rho(T^ix^n,T^iy^n)}P^*_{Y_i|Y^{i-1}}(dy_i|y^{i-1})\Big)\nonumber\\
&\qquad\times\overrightarrow{Q}_{Y^{i-1}|X^{i-1}}^*(dy^{i-1}|x^{i-1})\otimes{P}_{X^i}(dx^i).\label{ex15}
\end{align}
Moreover, if ${R}^{na}_{0,n}(D) > 0$ then $ s < 0$,  and
\begin{align}
\frac{1}{n+1}\sum_{i=0}^n\int_{{\cal X}_{0,i}\times{\cal Y}_{0,i}}\rho(T^ix^n,T^iy^n)\overrightarrow{Q}^*_{Y^i|X^i}(dy^i|x^i)\otimes{P}_{X^i}(dx^i)=D.\label{eq.7}
\end{align}
\end{description}
\end{theorem}
\begin{proof}
The proof is described in \cite[Theorem IV.4]{charalambous-stavrou-ahmed2014a}.
\end{proof}

\noi The RHS term of (\ref{ex14}) determines, for each $i=0,1,\ldots$, the dependence of the reproduction distribution $Q_{Y_i|Y^{i-1},X^i}^*(\cdot|y^{i-1},x^i)$ on the past reproduction $y^{i-1}$ and the past and present source symbols $x^i$. Below we list a few observations regarding the information structure of the optimal stationary reproduction conditional distribution corresponding to $R_{0,n}^{na}(D)$. We shall use some of these in specific examples.
\begin{remark}$(${Information structures of the optimal stationary reproduction distribution}$)$\label{markov_stationary}
\begin{itemize}
\item[{\bf (1)}] If $\{X_n:n\in\mathbb{N}\}$ is stationary Gaussian and $\rho(T^ix^n,T^iy^n)=||x_i-y_i||^2$, a quadratic function of $(x^n,y^n)$, then for each $(y^{i-1},x^i)\in{\cal Y}_{0,i-1}\times{\cal X}_{0,i}$, $Q_{Y_i|Y^{i-1},X^i}^*(\cdot|y^{i-1},x^i)$ is Gaussian. This follows from the fact that the exponent in the RHS of (\ref{ex14}) is quadratic in $(x_{i},y_i)\in{\cal X}_{i}\times{\cal Y}_{i}$, and thus by assuming $P_{Y_i|Y^{i-1}}^*(\cdot|y^{i-1})$ is conditionally Gaussian then the RHS of (\ref{ex14}) will be of exponential quadratic form in $(x_i,y^i)$. Hence, this RHS can be matched by a conditional Gaussian distribution for $Q_{Y_i|Y^{i-1},X^i}^*(\cdot|y^{i-1},x^i)$. The procedure is standard and involves completion of squares.
\item[{\bf (2)}] If the distortion function is $\rho(T^ix^n,T^iy^n)=\rho(x_i,y^i)$ then
\begin{align*}
{Q}_{Y_i|Y^{i-1},X^i}^*(\cdot|y^{i-1},x^i)={Q}_{Y_i|Y^{i-1},X^i}^*(\cdot|y^{i-1},x_i)-a.a.~(x^i,y^{i-1}),~i=0,1,\ldots.
\end{align*}
\noi That is, the reproduction conditional distribution is Markov in $\{X_n:n\in\mathbb{N}\}$. However, even if we further restrict the distortion function to single letter $\rho(x_i,y_i)$, we cannot deduce how far into the past $Q_{Y_i|Y^{i-1},X^i}^*(\cdot|y^{i-1},x_i)$ depends on the reproduction symbols $y^{i-1}$. If the distortion function is of the form $\rho(x_i,x_{i-1},y^i)$ then
\begin{align*}
{Q}_{Y_i|Y^{i-1},X^i}^*(\cdot;y^{i-1},x^i)={Q}_{Y_i|Y^{i-1},X_i,X_{i-1}}^*(\cdot|y^{i-1},x_i,x_{i-1})-a.a. (x^i,y^{i-1}),~i=0,1,\ldots.
\end{align*}
\end{itemize}
\end{remark}
Despite the above observations, for specific sources with memory, we need additional properties of the solution of the nonanticipative RDF to determine the dependence of the optimal reproduction distribution on past reproduction symbols. The main properties are introduced below. 

\begin{lemma}({Convexity and monotonicity of $R^{na}(D)$}){\ \\} \label{convex_nonincreasing}
${R}^{na}_{0,n}(D)$ is a convex, nonincreasing function of $D\in[0,\infty)$.
\end{lemma}
\begin{proof}
By Theorem~\ref{convexity_of_sets} the set $\overrightarrow{\cal Q}_{0,n}(D)$ is convex, and by Theorem~\ref{convexity_properties}, $\mathbb{I}_{X^n\rightarrow{Y^n}}(P_{X^n},\overrightarrow{Q}_{Y^n|X^n})$ is a convex functional of $\overrightarrow{Q}_{Y^n|X^n}(\cdot|x^n)\in\overrightarrow{\cal Q}({\cal Y}_{0,n}|{\cal X}_{0,n})$ for a fixed $P_{X^n}(dx^n)\in{\cal M}({\cal X}_{0,n})$. Hence, the result follows.
\end{proof}

\noi In the next lemma we identify minimum value of $D$ called $D_{max}$, beyond which $R^{na}_{0,n}(D)=0$.
\begin{lemma}($D_{max}$)\label{property5}{\ \\}
$R^{na}_{0,n}(D) > 0$ for all $D < D_{max}$ and $R^{na}_{0,n}(D)=0$ for all $D \geq D_{max}$, where
\begin{align*}
D_{max}\tri \min_{y^n\in {\cal Y}_{0,n}}\frac{1}{n+1}\sum_{i=0}^n \int_{{\cal X}_{0,i}}\rho(T^ix^n,T^iy^n)P_{X^i}(dx^i)
\end{align*}
provided the minimum exists.
\end{lemma}
\begin{proof}
The derivation is similar to the one for the classical RDF, hence it is omitted.
\end{proof}

\noi In the next lemma we give an alternative equivalent characterization of the optimal reproduction conditional distribution.
\begin{lemma}(Equivalent characterization of optimal stationary reproduction) \label{property3}{\ \\}
\noi The solution to the minimization problem of nonanticipative RDF defined by (\ref{ex15}) is such that
\begin{align}
\int_{{{\cal X}_{0,i}}} e^{s
\rho(T^ix^n,T^iy^n)} \lambda_{i}(x^i,y^{i-1}){P}_{X^i|Y^{i-1}}^*(dx^{i}|y^{i-1}) =1,~P_{Y_i|Y^{i-1}}^*(dy_i|y^{i-1})-a.s.,~\forall{i}\in\mathbb{N}^n\label{alternative-characterization1}
\end{align}
where 
\begin{align}
\lambda_{i}(x^i,y^{i-1})=\Big( \int_{{\cal Y}_i} e^{s\rho(T^ix^n,T^iy^n)} P_{Y_i|Y^{i-1}}^*(dy_i|y^{i-1})\Big)^{-1},~\forall{i}\in\mathbb{N}^n\label{alternative-characterization2}
\end{align}
and ${P}_{X^i|Y^{i-1}}^*(\cdot|\cdot)\in{\cal Q}({\cal X}_{0,i}|{\cal Y}_{0,i-1})$.
\end{lemma}
\begin{proof}
See Appendix~\ref{appendix_property3}.
\end{proof}
\noi Lemma~\ref{property3} characterizes the solution to the optimization problem described in Theorem~\ref{th6} on a set of $P^*_{Y_i|Y^{i-1}}$-measure 1. It can be shown, by utilizing measure theoretic arguments, that a necessary and sufficient condition for existence of a solution given in Theorem~\ref{th6} is the condition
\begin{align}
\int_{{{\cal X}_{0,i}}} e^{s
\rho(T^ix^n,T^iy^n)} \lambda_{i}(x^i,y^{i-1}){P}_{X^i|Y^{i-1}}^*(dx^{i}|y^{i-1})\leq{1},~\forall{y^i}\in{\cal Y}_{0,i}, \hso i \in {\mathbb N}^n.\label{necessary_sufficient}
\end{align}

\noi Next, we utilize Lemma~\ref{property3} to characterize the solution to the information nonanticipative RDF, as a maximization over a certain class of functions. By using this property, we can derive a lower bound on $R^{na}_{0,n}(D)$, which is analogous to Shannon's lower bound. In fact, we use this bound to determine the dependence of the optimal (stationary) reproduction distribution discussed in Remark~\ref{markov_stationary} on past reproduction symbols, and to derive the information nonanticipative RDF of the multidimensional Gaussian-Markov process given by (\ref{equation1o}).
\begin{theorem}(Alternative characterization of solution of the information nonanticipative RDF)\label{alternative_expression}{\ \\}
An alternative expression of the information nonanticipative RDF, $R^{na}_{0,n}(D)$ is
\begin{align}
R^{na}_{0,n}(D)&=\sup_{s\leq{0}}\sup_{\lambda\in\Psi_s}\bigg\{sD(n+1)+\sum_{i=0}^n\int_{{\cal X}_{0,i}\times{\cal Y}_{0,i-1}}\log\Big(\lambda_i(x^i,y^{i-1})\Big)P_{X^{i-1},Y^{i-1}}(dx^{i-1},dy^{i-1})\otimes{P}_{X_i|X^{i-1}}(dx_i|x^{i-1})\bigg\}\label{alternative-definition}
\end{align}
where
\begin{align}
\Psi_s&\tri\Big\{\lambda\tri\{\lambda_i(x^i,y^{i-1})\geq{0}:~i=0,1,\ldots,n\}:\nonumber\\
&~\int_{{\cal X}_{0,i}}e^{s\rho(T^ix^n,T^iy^n)}\lambda_i(x^i,y^{i-1}){P}_{X^i|Y^{i-1}}(dx^{i}|y^{i-1})\leq{1},~\forall{y}^i\in{\cal Y}_{0,i},~i=0,1,\ldots,n\Big\}\label{equation-ls}.
\end{align}
\noi Moreover, for each $s\leq{0}$, a necessary and sufficient condition for $\lambda_i(\cdot,\cdot)$ to achieve the supremum in (\ref{alternative-definition}) is the existence of $P_{Y_i|Y^{i-1}}^*(\cdot|\cdot)$ related to $\lambda_i(\cdot,\cdot)$ via (\ref{alternative-characterization2}) such that (\ref{alternative-characterization1}) holds with equality $a.a.~y_i\in{\cal Y}_i$, where $P_{Y_i|Y^{i-1}}(dy_i|y^{i-1})>0,~i=0,1,\ldots,n$. 
\end{theorem}
\begin{proof}
See Appendix~\ref{appendix_alternative_expression}.
\end{proof}
\noi The important point to be made regarding (\ref{alternative-definition}) is that by removing the supremum over $\lambda\in\Psi_s$ then a lower bound is obtained. For a given source and distortion function, this lower bound is then shown to be achievable by an optimal reproduction conditional distribution, which depends on finite past  reproduction symbols.

\subsection{Example 1:~Binary Symmetric Markov Source (BSMS($p$))}\label{example:bsms}

\par Consider a BSMS($p$), with stationary transition probabilities
$\big\{P_{X_i|X_{i-1}}(x_i|x_{i-1}):~(x_i,x_{i-1})\in\{0,1\}\times\{0,1\}\big\}$ given by $P_{X_i|X_{i-1}}(0|0)=P_{X_i|X_{i-1}}(1|1)=1-p$, and $P_{X_i|X_{i-1}}(1|0)=P_{X_i|X_{i-1}}(0|1)=p$, $i \in \{0,1,\ldots\}$. We consider a single letter Hamming distortion criterion, $\rho(x_i,y_i)=0$ if $x_i=y_i$ and $\rho(x_i,y_i)=1$ if $x_i \neq y_i$. The $R^{na}(D)$ of the BSMS($p$) is given in the next theorem. 
\begin{theorem}\label{marex1} The nonanticipative RDF, ${R}^{na}(D)$, for the BSMS($p$) with single letter Hamming distortion function is 
\[ { R}^{na}(D) = \left\{ \begin{array}{ll}
         H(m)-H(D) & \mbox{if $D \leq \frac{1}{2}$},~m=1-p-D+2pD\\
        0 & \mbox{otherwise}\end{array}  \right. \]
the optimal (stationary) reproduction distribution is 
\begin{align}
Q_{Y_i|Y_{i-1},X_{i}}^*(y_i|y_{i-1},x_{i})=\bbordermatrix{~ &  &  &  &  \cr
                   & \alpha & 1-\beta& \beta & 1-\alpha\vspace{0.3cm} \cr
                   & 1-\alpha & \beta& 1-\beta &  \alpha \cr}\label{marex1a}
\end{align}
where 
\begin{align}
\alpha=\frac{(1-p)(1-D)}{1-p-D+2pD},~\beta=\frac{p(1-D)}{p+D-2pD}.\label{def-alphabeta}
\end{align}
and $\{Y_i:~i=0,1\ldots\}$ is a first-order Markov with the same transition probability as that of the source.
\end{theorem}
Next, we outline the major steps of the proof; for a more detailed version see \cite[Theorem 2.34]{ck_thesis}.
\begin{proof} The steady state distribution of the source is given by $P_{X}(X_i=0)\equiv{P}_{X}(0)=P_{X}(X_i=1)\equiv{P}_{X}(1)=0.5$. By Theorem~\ref{marex1} (see Remark~\ref{markov_stationary}), the stationary reproduction distribution is Markov with respect to the source,  given by
\small
\begin{align}
Q_{Y_i|Y^{i-1},X^i}^*(y_i|y^{i-1},x^i)=Q_{Y_i|Y^{i-1},X_i}^*(y_i|y^{i-1},x_i)=
\frac{e^{s\rho(x_i,y_i)}P^*_{Y_i|Y^{i-1}}(y_i|y^{i-1})}{\sum_{y_i\in\{0,1\}}e^{s\rho(x_i,y_i)}P^*_{Y_i|Y^{i-1}}(y_i|y^{i-1})},~i=0,1,\ldots.\label{msexa}
\end{align}
\normalsize
The distribution of the reproduction symbols, $P^*_{Y_i|Y^{i-1}}(dy_i|y^{i-1})$, is calculated by reconditioning on $X_i$. Then, it is substituted into the RHS OF (\ref{msexa}) to deduce that both $P^*_{Y_i|Y^{i-1}}(dy_i|y^{i-1})$ and $Q_{Y_i|Y^{i-1},X_i}^*(y_i|y^{i-1},x_i)$ are conditional independent of $Y^{i-2}$. Solving the system of resulting equations we calculate $Q_{Y_i|Y_{i-1},X_{i}}^*(y_i|y_{i-1},x_{i})$ as a function of the Lagrange multiplier $``s"$. The Lagrange multiplier $``s"$ is found from the fidelity constraint and is substituted into the equation of the optimal reproduction distribution to obtain (\ref{marex1a}). It can be verified  that $\{Y_i:~i=0,1,\ldots\}$ is first-order Markov with the same transition probability as that of the source $\{X_i:~i=0,1,\ldots\}$. 
Finally, the information nonanticipative RDF is computed using the expression
\begin{align}
{R}^{na}(D)
&=\sum_{x_i,y_i,y_{i-1}}P^*_{X_i,Y_i,Y_{i-1}}(x_i,y_i,y_{i-1})\log\bigg(\frac{Q^*_{Y_i|Y_{i-1},X_i}(y_i|y_{i-1},x_i)}{P^*_{Y_i|Y_{i-1}}(y_i|y_{i-1})}\bigg). \label{eq.200i}
\end{align}
\end{proof}

\noi The graph of $R^{na}(D)$ is illustrated in Fig.~\ref{nardf-graph}; it shows that, as $p$ tends to $\frac{1}{2}$, and the source becomes less correlated  then $R^{na}(D)$ increases. Note that for $p=\frac{1}{2}$, then BSMS($\frac{1}{2}$) is the IID Bernoulli source. Then, $m=1-p-D+2pD=0.5$ and therefore $R^{na}(D)=1-H(D)$, $D<\frac{1}{2}$, which as expected is equal to the R(D) of an IID Bernoulli source with single letter distortion. At high resolution corresponding to $D\longrightarrow{0}$, then $\lim_{D\longrightarrow{0}}R^{na}(D)\simeq{H}(p)$ which is the entropy rate of the BSMS($p$).
\begin{figure}[ht]
\centering
\includegraphics[scale=0.70]{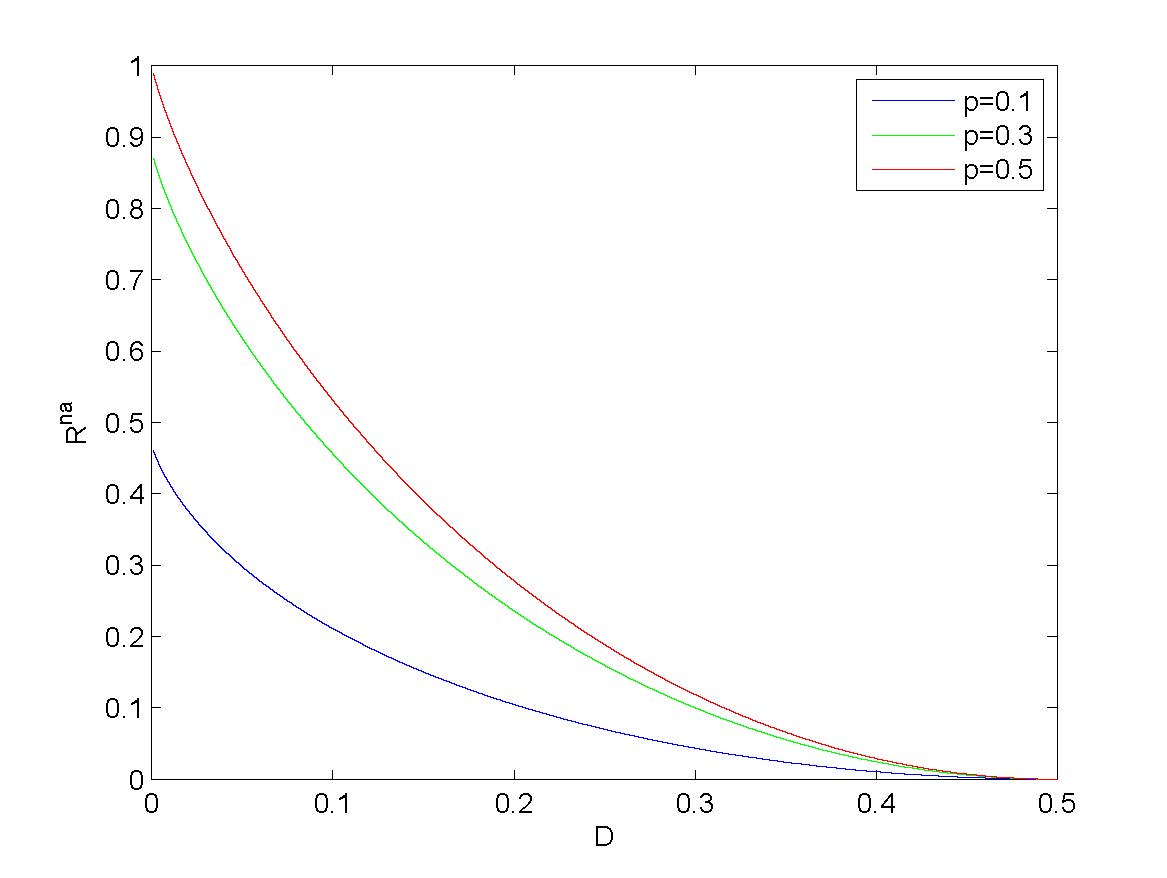}
\caption{$R^{na}(D)$ for different values of parameter $p$.}
\label{nardf-graph}
\end{figure}


\subsection{Example 2:~Multidimensional Partially Observed Gaussian Process}\label{example:gaussian}
\par Consider the discrete-time multidimensional partially observed linear Gauss-Markov source described by (\ref{equation1o}). The model in (\ref{equation1o}), is often encountered in applications where the process $\{Z_t:~t\in\mathbb{N}\}$ is not directly observed;  instead, what is directly observed is the process $\{X_t:~t\in\mathbb{N}\}$ which is a noisy version of it. This is a realistic model for any sensor which collects information on the underlying process $CZ_t$, subject to additive Gaussian noise. Hence, in this application the objective is to compress the sensor data $\{X_t:~t=0,1,\ldots,n\}$. Next, we introduce certain assumptions concerning (\ref{equation1o}) and the distortion function, which are sufficient for existence of the limit, $R^{na}(D)\tri\lim_{n\longrightarrow\infty}\frac{1}{n+1}R^{na}_{0,n}(D)$.
\begin{description}
\item[{\bf (G1)}] ($C,A$) is detectable and ($A,\sqrt{BB^{tr}}$) is stabilizable, \cite{caines1988}; 
\item[{\bf(G2)}] The state and observation noise $\{(W_t,V_t):t\in\mathbb{N}^n\}$ are Gaussian IID vectors $W_t\in\mathbb{R}^k$, $V_t\in\mathbb{R}^d$, mutually independent with parameters $N(0;I_{k\times{k}})$ and $N(0;I_{d\times{d}})$, independent of the Gaussian RV $Z_0$, with parameters $N(0;\bar{\Sigma}_0)$.
\item[{\bf (G3)}] The distortion function is single letter defined by $d_{0,n}(x^n,{y}^n)\tri\sum_{t=0}^n||x_t-{y}_t||_{\mathbb{R}^p}^2$.
\end{description}
Special cases of model (\ref{equation1o}) are discussed in  \cite{tatikonda-mitter2004,charalambous-farhadi2008}. Specifically, \cite{tatikonda-mitter2004} computes $R^{na}(D)$ for the scalar fully observed case corresponding to $X_t=Z_t\in\mathbb{R}$ and \cite{charalambous-farhadi2008} computes $R^{na}(D)$ for the scalar partially observed case $X_t\in\mathbb{R}$ via indirect methods.\\ 
Next, we invoke the characterization of the nonanticipative RDF given by Theorem~\ref{alternative_expression} to derive the exact expression of $R^{na}(D)$ for  model (\ref{equation1o}). The corresponding optimal reproduction conditional distribution is constructed as shown in Fig.~\ref{communication_system}.\\
\noi According to Theorem~\ref{th6}, the optimal stationary reproduction distribution is given by
\begin{align}
{Q}^*_{{Y}_t|Y^{t-1},X^t}(d{y}_t|y^{t-1},x^t)&=\frac{e^{s||{y}_t-x_t||_{\mathbb{R}^p}^2}P^*_{{Y}_t|{Y}^{t-1}}(d{y}_t|{y}^{t-1})}{\int_{{\cal Y}_t}e^{s||{y}_t-x_t||_{\mathbb{R}^p}^2}P^*_{{Y}_t|{Y}^{t-1}}(d{y}_t|{y}^{t-1})},~s\leq{0}\nonumber\\
&\equiv{Q}^*_{{Y}_t|Y^{t-1},X^t}(d{y}_t|y^{t-1},x_t)-a.a.~(y^{t-1},x_t).\label{eq.9}
\end{align}
\noi Hence, from (\ref{eq.9}), it follows that the optimal reproduction is Markov with respect to the process $\{X_t:~t\in\mathbb{N}\}$. Moreover, since the exponential term $||{y}_t-x_t||_{\mathbb{R}^p}^2$ in the RHS of (\ref{eq.9}) is quadratic in $(x_t,{y}_t)$, and $\{Z_t:~i\in\mathbb{N}\}$ is Gaussian then $\{(Z_t,{X}_t):~t\in\mathbb{N}\}$ are jointly Gaussian, and it follows that a Gaussian distribution $Q_{{Y}_t|{Y}^{t-1},X_t}(\cdot|{y}^{t-1},x_t)$ (for a fixed realization of $({y}^{t-1},x_t)$), and a Gaussian distribution $P_{{Y}_t|{Y}^{t-1}}(\cdot|{y}^{t-1})$ can match the left and right side of (\ref{eq.9}). Therefore, at any time $t\in\mathbb{N}$, the output ${Y}_t$ of the optimal reconstruction channel depends on $X_t$ and the previous outputs ${Y}^{t-1}$, and its conditional distribution is Gaussian. Hence, the channel connecting $\{X_t:t\in\mathbb{N}\}$ to $\{{Y}_t:t\in\mathbb{N}\}$ is realized by
\begin{eqnarray}
{Y}_t=\bar{A}X_t+\bar{B}{Y}^{t-1}+V^c_t,~t\in\mathbb{N}\label{eq.10}
\end{eqnarray}
where $\bar{A}\in\mathbb{R}^{p\times{p}}$, $\bar{B}\in\mathbb{R}^{p\times{t}p}$, and $\{V^c_t:~t\in\mathbb{N}\}$ is an independent sequence of Gaussian vectors $N(0;Q_t)$.\\
\noi Introduce the error estimate $\{K_t:~t\in\mathbb{N}\}$, and its covariance $\{\Lambda_t:~t\in\mathbb{N}\}$, defined by
\begin{eqnarray}
K_t\triangleq{X}_t-\widehat{X}_{t|t-1},~\widehat{X}_{t|t-1}\tri\mathbb{E}\Big{\{}X_t|\sigma\{{Y}^{t-1}\}\Big{\}},~\Lambda_t\triangleq\mathbb{E}\{K_tK_t^{tr}\},~t\in\mathbb{N}\label{equation52}
\end{eqnarray}
where $\sigma\{{Y}^{t-1}\}$ is the $\sigma$-algebra generated by the sequence $\{Y^{t-1}\}$. The covariance is diagonalized by introducing a unitary transformation $\{E_t:t\in\mathbb{N}\}$ such that 
\begin{eqnarray}
E_t\Lambda_t{E}_t^{tr}=diag\{\lambda_{t,1},\ldots\lambda_{t,p}\},~\Gamma_t\triangleq{E}_tK_t,~t\in\mathbb{N}.\label{equation53}
\end{eqnarray}
\noi Note that although $\{\Gamma_t:~t\in\mathbb{N}\}$ has independent Gaussian components, each one of these components are correlated. Analogously, introduce to the process $\{\tilde{K}_t:~t\in\mathbb{N}\}$ defined by
\begin{align}
\tilde{K}_t\tri{Y}_t-\widehat{X}_{t|t-1},~\tilde{\Gamma}_t=E_t\tilde{K}_t,~t\in\mathbb{N}.\label{eq.12i}
\end{align}

\noi Since $d_{0,n}(X^n,{Y}^n)=d_{0,n}(K^n,\tilde{K}^n)=\sum_{t=0}^n||\tilde{K}_t-K_t||_{\mathbb{R}^p}^2=\sum_{t=0}^n||\tilde{\Gamma}_t-\Gamma_t||_{\mathbb{R}^p}^2=d_{0,n}(\Gamma^n,\tilde{\Gamma}^n)$ the square error fidelity criterion $d_{0,n}(\cdot,\cdot)$ is not affected by the above preprocessing and post processing of  $\{(X_t,Y_t):~t\in\mathbb{N}\}$. Moreover, using basic properties of conditional entropy, if necessary, we can show the following expressions are equivalent. 
\begin{align}
R^{na}(D)&=\lim_{n\longrightarrow\infty}R_{0,n}^{na,K^n,\tilde{K}^n}(D)\tri\lim_{n\longrightarrow\infty}\inf_{\overrightarrow{P}_{\tilde{K}^n|K^n}:~\mathbb{E}\big{\{}d_{0,n}(K^n,\tilde{K}^n)\leq{D}\big{\}}}\frac{1}{n+1}\sum_{t=0}^n{I}(K_t;\tilde{K}_t|\tilde{K}^{t-1})\nonumber\\
&=\lim_{n\longrightarrow\infty}R_{0,n}^{na,\Gamma^n,\tilde{\Gamma}^n}(D)\tri\lim_{n\longrightarrow\infty}\inf_{\overrightarrow{P}_{\tilde{\Gamma}^n|\Gamma^n}:~\mathbb{E}\big{\{}d_{0,n}(\Gamma^n,\tilde{\Gamma}^n)\leq{D}\big{\}}}\frac{1}{n+1}\sum_{t=0}^n{I}(\Gamma_t;\tilde{\Gamma}_t|\tilde{\Gamma}^{t-1}).\label{equation.2}
\end{align}
\noi Next, we give the expression of $R^{na}(D)$ by using the specific realization of (\ref{eq.10}) shown in Fig.~\ref{communication_system}, where $\{{\cal A}_\infty,{\cal B}_\infty:~t=0,1,\ldots\}$ are to be determined.
 
\begin{theorem}($R^{na}(D)$ of multidimensional stationary partially observed Gaussian source)\label{solution_gaussian}{\ \\}
Under Assumptions ({\bf G1})-({\bf G3}), the information nonanticipative RDF rate for the multidimensional stationary partially observed Gaussian source (\ref{equation1o}) is given by
\begin{align}
R^{na}(D)=\frac{1}{2}\sum_{i=1}^p\log\Big{(}\frac{\lambda_{\infty,i}}{\delta_{\infty,i}}\Big{)}\label{equa.13}
\end{align}
where $diag\{\lambda_{\infty,1},\ldots,\lambda_{\infty,p}\}=\lim_{t\longrightarrow\infty}E_t\Lambda_tE^{tr}_{t}=E_\infty{\Lambda}_\infty{E}_\infty^{tr}$,
\begin{align} 
\Lambda_\infty&=\lim_{t\longrightarrow\infty}\mathbb{E}\Big\{\Big(C\big(Z_t-\mathbb{E}\big\{Z_t|\sigma\{Y^{t-1}\}\big\}\big)+NV_t\Big)\Big(C\big(Z_t-\mathbb{E}\big\{Z_t|\sigma\{Y^{t-1}\}\big\}\big)+NV_t\Big)^{tr}\Big\}\nonumber\\
&=C\lim_{t\longrightarrow\infty}\Sigma_t{C}^{tr}+NN^{tr}=C\Sigma_{\infty}C^{tr}+NN^{tr}\label{lambda_infty}
\end{align}
\begin{eqnarray}
\delta_{\infty,i} \tri\left\{ \begin{array}{ll} \xi_\infty & \mbox{if} \quad \xi_\infty\leq\lambda_{\infty,i} \\
\lambda_{\infty,i} &  \mbox{if}\quad\xi_\infty>\lambda_{\infty,i} \end{array} \right.,~i=2,\ldots,p\label{equa.13i}
\end{eqnarray}
and $\xi_\infty$ is chosen such that $\sum_{i=1}^p\delta_{\infty,i}=D$. Moreover, $\Sigma_\infty$ is the steady state covariance of the error $Z_t-\mathbb{E}\{Z_t|{Y}^{t-1}\}\sim{N}(0,\Sigma_\infty),~\widehat{Z}_{t|t-1}\tri\mathbb{E}\{Z_t|{Y}^{t-1}\}$, of the Kalman filter given by
\begin{align}
\widehat{Z}_{t+1|t}&=A\widehat{Z}_{t|t-1}\nonumber\\
&+A\Sigma_\infty(E_\infty^{tr}H_{\infty}E_{\infty}C)^{tr}M_\infty^{-1}\big({Y}_t-C\widehat{Z}_{t|t-1}\big),~\hat{Z}_{0|-1} =\mathbb{E}\{ Z_0| Y^{-1}\}, Z_0- \hat{Z}_{0|-1}\sim{N}(0,\Sigma_\infty),~t\in\mathbb{N}\label{10}\\
\Sigma_{\infty}&=A\Sigma_\infty{A}^{tr}-A\Sigma_{\infty}(E_\infty^{tr}H_\infty{E}_{\infty}C)^{tr}M_{\infty}^{-1}(E_{\infty}^{tr}H_{\infty}E_{\infty}C)\Sigma_{\infty}A^{tr}+BB_{\infty}^{tr}\label{11}
\end{align}
\begin{align}
M_\infty=E_\infty^{tr}H_\infty{E}_{\infty}C\Sigma_{\infty}(E_{\infty}^{tr}H_{\infty}E_{\infty}C)^{tr}+E_{\infty}^{tr}H_{\infty}E_{\infty}NN^{tr}(E_{\infty}^{tr}H_{\infty}E_{\infty})^{tr}+E_{\infty}^{tr}{\cal B}_{\infty}Q{\cal B}_{\infty}^{tr}E_\infty\label{11i}
\end{align}
and
\begin{align}
H_\infty&=\lim_{t\longrightarrow\infty}H_t={d}iag\{\eta_{\infty,1},\ldots,\eta_{\infty,p}\},~H_t\tri{d}iag\{\eta_{t,1},\ldots,\eta_{t,p}\},~\eta_{t,i}=1-\frac{\delta_{t,i}}{\lambda_{t,i}},~i=1,\ldots,p,~t\in\mathbb{N}\label{equation11mmm}\\
{\cal B}_{\infty}&=\lim_{t\longrightarrow\infty}{\cal B}_t=\sqrt{H_\infty\Delta_\infty{Q}^{-1}},~{\cal B}_t\triangleq\sqrt{H_t\Delta_tQ^{-1}},~t\in\mathbb{N}\label{equation11mmma}\\
\Delta_\infty&=\lim_{t\longrightarrow\infty}\Delta_t=diag\{\delta_{\infty,1},\ldots,\delta_{\infty,p}\},~\Delta_t=diag\{\delta_{t,1},\ldots,\delta_{t,p}\}, t\in\mathbb{N}\label{equation11mmmb}\\
Q&=\lim_{t\longrightarrow\infty}Q_t=diag\{q_{\infty,1},\ldots,q_{\infty,p}\}, Q_t=diag\{q_{t,1},\ldots,q_{t,p}\},~t\in\mathbb{N}\label{equation11mmmbe}.
\end{align}
Moreover, 
\begin{align}
{Y}_t=E_\infty^{tr}H_{\infty}E_{\infty}C(Z_t-\widehat{Z}_{t|t-1})+C\widehat{Z}_{t|t-1}+E_\infty^{tr}H_{\infty}E_{\infty}N V_t+E_{\infty}^{tr}{\cal B}_{\infty}V^c_t,~t\in\mathbb{N}.\label{equation_decoder1}
\end{align}
\end{theorem}
\begin{proof}
See Appendix~\ref{appendix_solution_gaussian}.
\end{proof}
\vspace*{0.2cm}
\noi For scalar stationary Gaussian sources with memory, $R^{na}(D)$ given in Theorem~\ref{solution_gaussian} simplifies considerably. We illustrate this via examples.\\

\noi{\bf Degraded cases.}
\begin{itemize}
\item[{\bf(1)}] \underline{Scalar Stationary Partially Observed Gaussian Markov Source:} This corresponds to (\ref{equation1o}) by setting $m=p=1$, $C=c$, $N=\sigma_V$, $A=\alpha$, $B=\sigma_W$, i.e., $\sigma_{W}W_t\sim{N}(0;\sigma^2_{W})$ and $\sigma_{V}V_t\sim{N}(0;\sigma^2_{V})$ giving
\begin{align}
\left\{ \begin{array}{ll} Z_{t+1}=\alpha{Z}_t+\sigma_{W}W_t,~t=0,1,\ldots,~Z_0\sim{N}(0;\sigma^2_W(1-\alpha^2)^{-1}), \hso |\alpha|<1, \\
X_t=cZ_t+\sigma_V{V}_t,~t=0,1,\ldots\end{array} \right.\label{equation27j}
\end{align}
Then $\sigma^2_{Z_t}\tri Var(Z_t)=\sigma^2_W(1-\alpha^2)^{-1}$, $\sigma^2_{X_t}=c^2\sigma^2_{Z_t}+\sigma^2_V$. In this case, by Theorem~\ref{solution_gaussian} we have 
\begin{align}
\Lambda_\infty=\lambda_{\infty,1}=c^2\Sigma_\infty+\sigma^2_V,~\Delta_\infty=\delta_{\infty,1}=D,~\mbox{where}~H_\infty=1-\frac{D}{c^2\Sigma_\infty+\sigma^2_V}~\mbox{and}~E_\infty=1.\nonumber 
\end{align}
Hence, using (\ref{11i}), we obtain
\begin{align}
&M_{\infty}=c^2\Sigma_{\infty}H^2_\infty+\sigma^2_V{H}^2_{\infty}+H_\infty{D}=H^2_\infty\big(c^2\Sigma_\infty+\sigma^2_V)+H_\infty{D}=(c^2\Sigma_\infty+\sigma^2_V){H}_\infty.\label{11jj}
\end{align}
Also, from (\ref{11}), we obtain
\begin{align}
\Sigma_\infty&=\alpha^2\Sigma_\infty-\alpha^2c^2\Sigma^2_\infty{H}^2_\infty{M}^{-1}+\sigma_W^2\nonumber\\
&\stackrel{(a)}
\Longrightarrow{c^4}\Sigma^3_\infty+\Sigma^2_\infty(2c^2\sigma^2_V-\alpha^2c^2\sigma^2_V-\alpha^2c^2D-c^4\sigma^2_W)-\Sigma_\infty(\alpha^2\sigma^4_V+2c^2\sigma^2_V\sigma^2_W)-\sigma^4_V\sigma^2_W=0\label{11kk}
\end{align}
where $(a)$ follows from (\ref{11jj}). It can be verified that the cubic equation (\ref{11kk}) admits a positive  solution.    
From (\ref{equa.13}) we obtain
\begin{align}
R^{na}(D)=\frac{1}{2}\log\frac{\lambda_{\infty,1}}{\delta_{\infty,1}},~\lambda_{\infty,1}=c^2\Sigma_\infty+\sigma^2_V \geq{D}.\label{partially_scalar}
\end{align}
\item[{\bf(2)}] \underline{Scalar Stationary Fully Observed Gaussian Markov Source:} This corresponds to (\ref{equation27j}) by setting $c=1$, $\sigma_V=0$ giving
\begin{align}
X_{t+1}=\alpha{X}_t+\sigma_W{W}_t,~W_t\sim{N}(0;1),~X_0\sim{N}(0;\sigma^2_W(1-\alpha^2)^{-1}),~|\alpha|<1.\label{equation111}
\end{align}
Then $\sigma^2_{X_t} \tri Var(X_t)=\sigma^2_W(1-\alpha^2)^{-1}$.
Since this is a special case of {\bf(1)} for $c=1$, $\sigma_V=0$, we obtain from (\ref{11jj})
\begin{align}
&M_{\infty}=\Sigma_\infty{H}_\infty.\label{11j}
\end{align}
Also, using (\ref{11kk}), we obtain
\begin{align}
\lambda_{\infty,1} = \Sigma_\infty=\alpha^2{D}+\sigma_W^2\label{11k}
\end{align}
 Finally, by substituting (\ref{11k}) in the expression of the nonanticipative RDF (\ref{equa.13}) we obtain
\begin{align}
R^{na}(D)=\frac{1}{2}\log\frac{\lambda_{\infty,1}}{\delta_{\infty,1}}=\frac{1}{2}\log\frac{\Sigma_\infty}{D}=\frac{1}{2}\log\Big(\frac{\alpha^2{D}+\sigma_W^2}{D}\Big)=\frac{1}{2}\log\Big(\alpha^2+\frac{\sigma_W^2}{D}\Big),~\Sigma_\infty\geq{D},~|\alpha|<1.\label{12j}
\end{align} 
\noi This is precisely the expression derived in \cite[Theorem 3]{derpich-ostergaard2012} using power spectral densities. 
\item[{\bf(3)}]\underline{IID Gaussian Source:} This corresponds to (\ref{equation111}) by setting  $\alpha=0$, $\sigma_X=\sigma_W$, which implies $\{X_t:~t=0,\ldots\}$ is $N(0;\sigma^2_X)$. By (\ref{12j}), with $\alpha=0$, $\sigma_X=\sigma_W$ then $R^{na}(D)=R(D)=\frac{1}{2}\log\frac{\sigma^2_{X}}{D}$, $\sigma^2_{X}\geq{D}$.\\
This is the well known RDF of the IID Gaussian source \cite{berger}. 
\item[{\bf(4)}]\underline{Vector Source with Independent Component versions of {\bf(1)}, {\bf(2)}.} The vector versions of (\ref{equation27j}) and (\ref{equation111}) with independent spacial components, $\{Z_t^1,\ldots,Z_t^m\}$, and $\{X_t^1,\ldots,X_t^p\}$ is a straight forward extension of (\ref{partially_scalar}) and (\ref{12j}).
\end{itemize}

\section{JSCC Design Based on Symbol-by-Symbol Transmission of Gaussian Sources with Memory}\label{joint_source_channel_coding}
In this section, we use the solution of the nonanticipative RDF $R^{na}(D)$ of the multidimensional stationary Gaussian source with memory in  JSCC design using symbol-by-symbol transmission  over an AWGN channel with or without feedback. We also illustrate that Theorem~\ref{solution_gaussian} and its corresponding realization scheme shown in Fig.~\ref{communication_system} gives as a special case the Schalkwijk-Kailath coding scheme.\\
\noi First, we show that $\{\tilde{K}_t:~t\in\mathbb{N}\}$ is the innovation process of $\{Y_t:~t\in\mathbb{N}\}$, and hence the two processes generate the same $\sigma$-algebras (they contain the same information).
\begin{lemma}(Equivalence of information generated by $\{Y_t:~t=0,\ldots\}$ and $\{\tilde{K}_t:~t=0,\ldots\}$)\label{sigma-algebras}{\ \\}
The following hold.
\begin{align*}
{\cal F}_{0,t}^{Y}\tri\sigma\{Y_s:~s=0,1,\ldots,t\}={\cal F}_{0,t}^{\tilde{K}}\tri\sigma\{\tilde{K}_s:~s=0,1,\ldots,t\},~\forall{t}\in\mathbb{N}.
\end{align*}
that is, ${\cal F}_{0,t}^{Y}\subseteq{\cal F}_{0,t}^{\tilde{K}}$ and ${\cal F}_{0,t}^{\tilde{K}}\subseteq{\cal F}_{0,t}^{Y}$,~$\forall{t}\in\mathbb{N}$.
\end{lemma}
\begin{proof}
Since $\tilde{K}_s=Y_s-\mathbb{E}\big{\{}X_s|Y^{s-1}\big{\}}$, $0\leq{s}\leq{t}$, then ${\cal F}_{0,t}^{\tilde{K}}\subseteq{\cal F}_{0,t}^{Y}$, $\forall{t}\in\mathbb{N}$. Hence, we need to show that ${\cal F}_{0,t}^{Y}\subseteq{\cal F}_{0,t}^{\tilde{K}}$, $\forall{t}\in\mathbb{N}$. The innovation process of $\{Y_t:~t\in\mathbb{N}\}$ is by definition (see Fig.~\ref{communication_system}, (\ref{eq.12i}), (\ref{equation95ii}))
\begin{align}
{I}_t&=Y_t-\mathbb{E}\big{\{}Y_t|Y^{t-1}\big{\}}\nonumber\\
&=E_\infty^{tr}H_{\infty}E_{\infty}\Big{(}X_t-\mathbb{E}\big{\{}X_t|Y^{t-1}\big{\}}\Big{)}+E_{\infty}^{tr}{\cal B}_{\infty}V^c_t+\mathbb{E}\big{\{}X_t|Y^{t-1}\big{\}}-\mathbb{E}\big{\{}X_t|Y^{t-1}\big{\}}\nonumber\\
&=E_{\infty}^{tr}H_{\infty}E_{\infty}\Big(X_t-\mathbb{E}\big{\{}X_t|Y^{t-1}\big{\}}\Big)+E_{\infty}^{tr}{\cal B}_{\infty}V^c_t=\tilde{K}_t.\label{equation_sigma_algebra}
\end{align}
Since the innovation process $\{I_s:~s=0,1,\ldots,{t}\}$ and the optimal reproduction process $\{Y_s:~s=0,1,\ldots,t\}$ generates the same $\sigma-$algebras, then ${\cal F}_{0,t}^{I}\subseteq{\cal F}_{0,t}^{Y}$, ${\cal F}_{0,t}^{Y}\subseteq{\cal F}_{0,t}^{I}$, i.e., ${\cal F}_{0,t}^{Y}={\cal F}_{0,t}^{I}$, and hence, by (\ref{equation_sigma_algebra}) we also obtain ${\cal F}_{0,t}^{Y}\subseteq{\cal F}_{0,t}^{\tilde{K}}$,~$\forall{t}\in\mathbb{N}$.
\end{proof}

\noi We now observe the following consequence of Lemma~\ref{sigma-algebras}.

\begin{remark}
\label{scalar}{\ \\}
By Lemma~\ref{sigma-algebras}, all conditional expectations with respect to the process $\{Y_t:  t=0,1, \ldots\}$ can be replaced by conditional expectations with respect to the independent process $\{\tilde{K}_t: t=0, 1, \ldots\}$. Hence, the process $\{K_t: t=0, 1, \ldots\}$ can be written as $K_t=X_t-\mathbb{E}\big\{X_t|\sigma\{Y^{t-1}\}\big\}=X_t-\mathbb{E}\big\{X_t|\sigma\{\tilde{K}^{t-1}\}\big\}$, while its reconstruction is given by 
\begin{align}
\tilde{K}_t=E_{\infty}^{tr}H_{\infty}E_{\infty}\Big(X_t-\mathbb{E}\big{\{}X_t|\tilde{K}^{t-1}\big{\}}\Big)+E_{\infty}^{tr}{\cal B}_{\infty}V^c_t=E_{\infty}^{tr}H_{\infty}E_{\infty}K_t+E_{\infty}^{tr}{\cal B}_{\infty}V^c_t,~t=0, 1, \ldots. \label{scalar1}
\end{align}
Furthermore, by Lemma~\ref{sigma-algebras}, $K_t$ and $\tilde{K}_t$ are independent of $Y_0,\ldots,Y_{t-1}$, and $\tilde{K}_0,\ldots,\tilde{K}_{t-1}$, $t=0,1,\ldots$. This property is analogous to the JSCC of a scalar RV over a scalar additive Gaussian noise channel with feedback \cite[Theorem 5.6.1]{ihara1993}. 
\end{remark}

\noi The realization of Fig.~\ref{communication_system} illustrates a ``Duality of a Source and a Channel" \cite{charalambous-kourtellaris-stavrou2015}, that of the multidimensional stationary Gaussian source process and the multidimensional memoryless AWGN channel. To give an operational meaning to this duality, based on the realization of $R^{na}(D)$ of Fig.~\ref{communication_system}, we need to ensure that end-to-end average distortion is achieved.  We do this by using existing results on capacity of memoryless  AWGN channels. 

%

\noi Consider the memoryless AWGN channels which appear in Fig.~\ref{communication_system}, defined by $B_{t}=A_{t}+V^c_{t},~t=0,\ldots,n$, with  $\{V^c_t \tri Vector\{V_{t,1}^c, V_{t,2}^c, \ldots, V_{t, p}^c\}: t=0,\ldots, n\}$,~$N(0;Q_t)$ (Gaussian) with  $Q_t\tri{C}ov(V^c_t)=diag\{q_{t,1},q_{t,2},\ldots,q_{t,p}\}$,  $\{A_t \tri Vector\{A_{t,1},A_{t,2}, \ldots, A_{t, p}\}: t=0,\ldots, n\}$, $P_t=Cov(A_t), t=0, \ldots, n$,  and power constraint $\frac{1}{n+1}\sum_{t=0}^n {\mathbb E} ||A^2_t||_{\mathbb{R}^p} =\frac{1}{n+1} \sum_{t=0}^n Trace(P_t)  \leq{P}$. It is known that the capacity of such a channel with or without feedback  subject to a power constraint is the same, and it is given by  $ C(P)\tri \lim_{n\longrightarrow \infty}\frac{1}{n+1} C_{0,n}(P)$, where 
\begin{align}
C_{0,n}(P)&=\sup_{P_{A^n}:~\frac{1}{n+1}\mathbb{E}\{\sum_{t=0}^n||A^2_t||_{\mathbb{R}^p}\}\leq{P}}I(A^n;B^n)\label{nonstationary:capacity:waterfilling}
\end{align}
\noi Further, it is known that the channel input distribution corresponding to the maximization of  the right hand side of  (\ref{nonstationary:capacity:waterfilling}) satisfies $P_{A_t|A^{t-1}}=P_{A_t}, t=0, \ldots, n$, and $\{A_t: t=0,1, \ldots, n\}$   is Gaussian $N(0;P_t), P_t \tri E \{A_tA_t^{tr}\}=Cov(A_t)$. Denote the eigenvalues of $P_{t}$ by  $P_{t,1}, P_{t,2}, \ldots, P_{t, p}$, for $t=0,1, \ldots, n$. Then (\ref{nonstationary:capacity:waterfilling}) becomes
 \begin{align}
 C_{0,n}(P)\equiv{C}_{0,n}(P^*_{t,1},\ldots,P^*_{t,p}:~t=0,\ldots,n)= \max_{\frac{1}{n+1}\sum_{t=0}^n \sum_{i=1}^p P_{t,i} \leq P} \frac{1}{2} \frac{1}{n+1} \sum_{t=0}^n \sum_{i=0}^p\log(1+\frac{P_{t,i}}{q_{t,i}}). \label{WFC}
 \end{align}
 where $\{P^*_{t,1},\ldots,P^*_{t,p}\}$ is the optimal allocation of power. The capacity is given by
\begin{align}
&C(P)\equiv{C}(P^*_{\infty,1},\ldots,P^*_{\infty,p})=\lim_{n\longrightarrow\infty}C_{0,n}(P)\nonumber\\
&=\lim_{n\longrightarrow\infty}\max_{\frac{1}{n+1}\sum_{t=0}^n \sum_{i=1}^p P_{t,i} \leq P} \frac{1}{2} \frac{1}{n+1} \sum_{t=0}^n \sum_{i=0}^p \log(1+ \frac{P_{t,i}}{q_{t,i}})=\frac{1}{2}\sum_{i=0}^p\log(1+ \frac{P^*_{\infty,i}}{q_{\infty,i}}),~\sum_{i=1}^p{P}^*_{\infty,i}=P.\label{WFC_limit}
\end{align}
The solution can be found using standard techniques and corresponds to water-filling of parallel memoryless Gaussian channels.\\
 \noi Next, we ensure $R(D)=C(P)$, end-to-end average distortion is satisfied, and the encoder operates at channel capacity. For a given $D\in[D_{min},D_{max}]>0$ there exists a power $P \in [P_{min}, P_{max}]$ such that 
\begin{align}
R^{na}(D)&=\lim_{n\longrightarrow\infty}{R}_{0,n}^{na}(D)=\lim_{n\longrightarrow\infty}\frac{1}{2(n+1)} \sum_{t=0}^n\log\frac{|\Lambda_{t}|}{|\Delta_{t}|}\nonumber\\
&=\frac{1}{2}\sum_{i=1}^p\log\Big(\frac{\lambda_{\infty,i}}{\delta_{\infty,i}}\Big)=\frac{1}{2}\sum_{i=0}^p\log(1+ \frac{P^*_{\infty,i}}{q_{\infty,i}})=C(P)\Big{|}_{\frac{P^*_{\infty,i}}{q_{\infty,i}}=\frac{\lambda_{\infty,i}}{\delta_{\infty,i}}-1,~i=1,\ldots,p}.\label{equation59}
\end{align}
Below we give the JSCC design based on symbol-by-symbol transmission with and without feedback encoding.

\noi{\bf JSCC Design with Feedback.} The AWGN channel with feedback has the following implementation.
\begin{align}
B_t=A^{fb}_\infty(X_t,B^{t-1})+V_t^{c}&={\cal A}_{\infty}E_{\infty}\Big(X_t-\mathbb{E}\{X_t|B^{t-1}\}\Big)+V_t^c,~{\cal A}_\infty\tri\sqrt{Q\Delta_{\infty}^{-1}H_{\infty}},~t\in\mathbb{N},~i=1,\ldots,p.\label{additive_channel_limit}
\end{align}
where ${\cal A}_\infty$ is chosen to ensure the power allocation is satisfied and to guarantee the encoder operates at $C(P)$.
\noi This shows that, for a given distortion level $D\in[D_{min},D_{max}]>0$, the realization shown in Fig.~\ref{communication_system} is optimal in the sense that the end-to-end nonanticipative RDF, $R^{na}(D)$ is achieved (with the prescribed average distortion), the encoder (\ref{additive_channel_limit}) achieves the capacity of the channel, and (\ref{equation59}) is satisfied. Thus, $R^{na}(D)$ is achievable over the $\{$encoder, channel, decoder$\}$ design  shown in Fig.~\ref{communication_system}.\\

\noi{\bf JSCC Design without Feedback.} The AWGN channel without feedback follows from (\ref{additive_channel_limit}) with $\mathbb{E}\{X_t|B^{t-1}\}$ replaced by $\mathbb{E}\{X_t|\sigma(\mathbb{R}^p)\}=EX_t=0$, that is, $B_t=A^{nfb}_\infty(X_t)+V_t^{c}={\cal A}_{\infty}E_{\infty}X_t+V_t^c,~{\cal A}_\infty\tri\sqrt{Q\Delta_{\infty}^{-1}H_{\infty}},~t\in\mathbb{N},~i=1,\ldots,p$.\\

\noi Note that the JSCC design (shown in Fig.~\ref{communication_system}) is not the only choice. One may consider JSCC design of the multidimensional Gaussian source with memory over other types of Gaussian channels.
\vspace*{0,2cm}

\noi Next, we show that {\it JSCC design of Multidimensional stationary Gaussian Sources with Memory} presented in Fig.~\ref{communication_system} includes as degraded special cases several scenarios, including previously known results.\\

\noi{\bf Degraded cases.}

\begin{itemize}
\item[{\bf(FB1)}]\underline{\it JSCC design of transmitting a scalar Gaussian Markov Source over a memoryless AWGN channel with feedback.} Let $X_t$ be a scalar Gaussian Markov source defined by (\ref{equation111}), where it is shown that $\lambda_{\infty,1}=\alpha^2{D}+\sigma^2_W$, $R^{na}(D)=\frac{1}{2}\log(\alpha^2+\frac{\sigma^2_W}{D})$ (see (\ref{11k}), (\ref{12j})). The capacity of a memoryless AWGN channel with or without feedback is obtained from (\ref{WFC_limit}) by setting $p=1$ and $q_{\infty,1}=\mathbb{E}\{|V^c_t|^2\}=Q=\sigma^2_{V^c}$ giving
\begin{align}
C(P)=\frac{1}{2}\log(1+\frac{P}{\sigma^2_{V^c}}).\label{gaussian_capacity_scalar}
\end{align}
Suppose that the channel is used once per source symbol. For this realization, the smallest achievable distortion is
\begin{align}
{D}_{min}=\frac{\sigma^2_{W}\sigma^2_{V^c}}{(1-\alpha^2)\sigma^2_{V^c}+P}.\label{minimum_distortion_2}
\end{align}
\noi Using Theorem~\ref{solution_gaussian} and (\ref{additive_channel_limit}) we obtain  
\begin{align}
B_t= \sqrt{\frac{P}{\lambda_{\infty,1}}}K_t+V_t^c=\sqrt{\frac{P\big((1-\alpha^2)\sigma^2_{V^c}+P\big)}{\sigma^2_W(\sigma^2_{V^c}+P)}}K_t+V_t^c,~K_t=X_t-\mathbb{E}\{X_t|B^{t-1}\}.\label{additive_channel_feedback1b} 
\end{align}
From Remark~\ref{scalar} (by setting $p=1$), the decoder expression (\ref{scalar1}) becomes
\begin{align}
\tilde{K}_t=H_\infty{K}_t+{\cal B}_\infty{V}_t^c={\cal B}_\infty\big({\cal A}_\infty{K}_t+V_t^c\big)\stackrel{(a)}={\cal B}_\infty{B}_t
\label{decoder1}
\end{align}
where $(a)$ follows from (\ref{additive_channel_limit}). By using the fact that $q_{\infty,1}=\sigma^2_{V^c},~\delta_{\infty,1}=D$ the scaling factor ${\cal B}_\infty$ (\ref{equation11mmma}) (for the scalar case $p=1$), which guarantees the minimum end-to-end distortion error is 
\begin{align}
{\cal B}_\infty=\sqrt{\frac{(1-\frac{\delta_{\infty,1}}{\lambda_{\infty,1}})\delta_{\infty,1}}{q_{\infty,1}}}=\sqrt{\frac{\sigma^2_W{P}}{\big((1-\alpha^2)\sigma^2_{V^c}+P\big)\big(\sigma^2_{V^c}+P\big)}}.\label{scaling_factor_feedback_2}
\end{align}
By substituting (\ref{scaling_factor_feedback_2}) into (\ref{decoder1}) we get
\begin{align}
\tilde{K}_t=\sqrt{\frac{\sigma^2_W{P}}{\big((1-\alpha^2)\sigma^2_{V^c}+P\big)\big(\sigma^2_{V^c}+P\big)}}{B}_t.
\label{decoder2b}
\end{align}
Finally, the average end-to-end distortion is computed by evaluating the expectation
\begin{align*}
D=\mathbb{E}\{|X_t-{Y}_t|^2\}=\mathbb{E}\{|K_t-\tilde{K}_t|^2\}=\frac{\sigma^2_W\sigma^4_{V^c}+\sigma^2_W\sigma^2_{V^c}P}{\big((1-\alpha^2)\sigma^2_{V^c}+P\big)(P+\sigma^2_{V^c})}=\frac{\sigma^2_W\sigma^2_{V^c}}{(1-\alpha^2)\sigma^2_{V^c}+P}=D_{min}.
\end{align*} 
\noi This JSCC  design is the one  illustrated in Fig.~\ref{jscc_fully_feedback}.

%

\end{itemize}

\vspace*{0.2cm}
\noi\underline{Special case: Realization Without Feedback.} When there is no feedback, all statements presented in {\bf(FB1)} holds, with ${\mathbb E}\{X_t|B^{t-1}\}$ replaced by ${\mathbb E}\{X_t|\sigma\{ {\mathbb R}\}\}=  {\mathbb E}\{X_t\}=0$ (i.e., only \'a priori information  is used), and the encoder has the following structure.  
\begin{align}
B_t= \sqrt{\frac{P}{\lambda_{\infty,1}}}X_t+V_t^c,~\lambda_{\infty,1}=Var(X_t)=\sigma_W^2(1-|\alpha|^2)^{-1},~|\alpha|<1. \label{additive_channel_nofeedback} 
\end{align}
Toward this, we inspect {\bf(FB1)} without feedback.

\begin{itemize}
\item[{\bf (NFB1)}] \underline{\it JSCC design of transmitting a scalar Gaussian Markov source over a memoryless AWGN channel without feedback.} Let $\{X_t: t=0,1, \ldots \}$ be the scalar stationary Gaussian Markov source defined by (\ref{equation111}). In contrary to the case {\bf (FB1)}, here the encoder has the form of (\ref{additive_channel_nofeedback}), where $\lambda_{\infty,1}=\sigma^2_{X_t}=\frac{\sigma^2_W}{1-\alpha^2}$, and $R^{na}(D)=\frac{1}{2}\log\frac{\sigma^2_W}{(1-\alpha^2)D},~\frac{\sigma^2_W}{1-\alpha^2}\geq{D},~|\alpha|<1$. Consider the realization of $R^{na}(D)$ over a memoryless AWGN channel without feedback whose channel capacity is defined by (\ref{gaussian_capacity_scalar}). Suppose that the channel is used once per source symbol. For this realization, the smallest achievable distortion is
\begin{align}
{D}_{min}=\frac{\sigma^2_{W}\sigma^2_{V^c}}{(1-\alpha^2)(P+\sigma^2_{V^c})}.\label{minimum_distortion_2_nofeedback}
\end{align}
\noi Then (\ref{additive_channel_nofeedback}) is expressed as 
\begin{align}
B_t= \sqrt{\frac{P}{\lambda_{\infty,1}}}K_t+V_t^c=\sqrt{\frac{(1-\alpha^2)P}{\sigma^2_W}}X_t+V_t^c.\label{additive_channel_nofeedback1b} 
\end{align}
\noi By using the fact that $q_{\infty,1}=\sigma^2_{V^c},~\delta_{\infty,1}=D$ the scaling factor ${\cal B}_\infty$ (\ref{equation11mmma}) for the scalar case $p=1$, which guarantees the minimum end-to-end distortion error is 
\begin{align}
{\cal B}_\infty=\sqrt{\frac{(1-\frac{\delta_{\infty,1}}{\lambda_{\infty,1}})\delta_{\infty,1}}{q_{\infty,1}}}=\sqrt{\frac{\lambda_{\infty,1}}{P}}\frac{P}{P+\sigma^2_{V^c}}=\sqrt{\frac{\sigma^2_W}{(1-\alpha^2)P}}\frac{P}{P+\sigma^2_{V^c}}.\label{scaling_factor_nofeedback_2}
\end{align}
By substituting (\ref{scaling_factor_nofeedback_2}) into (\ref{decoder1}) we get
\begin{align}
\tilde{K}_t=\sqrt{\frac{\sigma^2_W}{(1-\alpha^2)P}}\frac{P}{P+\sigma^2_{V^c}}{B}_t.
\label{decoder2b_nofeedback}
\end{align}
Finally, the end-to-end distortion is computed by evaluating the expectation
\begin{align*}
D=\mathbb{E}\{|X_t-{Y}_t|^2\}=\mathbb{E}\{|K_t-\tilde{K}_t|^2\}=\frac{\sigma^2_W\sigma^4_{V^c}}{(1-\alpha^2)(P+\sigma^2_{V^c})^2}+\frac{\sigma^2_W\sigma^2_{V^c}P}{(1-\alpha^2)(P+\sigma^2_{V^c})^2}=\frac{\sigma^2_W\sigma^2_{V^c}}{(1-\alpha^2)(P+\sigma^2_{V^c})}=D_{min}.
\end{align*}

\noi This JSCC design is the one   illustrated in Fig.~\ref{jscc_fully_nofeedback}.

\item[{\bf (NFB2)}] \underline{\it JSCC design of transmitting a scalar IID Gaussian source over a memoryless AWGN channel without feedback.} By setting $\alpha=0$ and $\sigma_X=\sigma_W$ in case {\bf(NFB1)}, the encoder has the form of (\ref{additive_channel_nofeedback}) where $\lambda_{\infty,1}=\sigma^2_X$.\\
\noi This is the case discussed in \cite{goblick1965} (see also \cite[Example 2.2]{gastpar2002}).\\
\end{itemize}

\vspace*{0.2cm}
\noi Next, we show how to recover from   the realization scheme depicted in Fig.~\ref{communication_system}, the {\it Schalkwijk-Kailath coding scheme which achieves the feedback capacity of memoryless Gaussian channels} \cite{schalkwijk-kailath1966}.\\
\underline{Feedback Realization: The Schalkwijk-Kailath coding scheme.} Consider a scalar Gaussian RV $X\sim{N}(0;\sigma^2_X)$. By letting $p=1$ in (\ref{equation59})-(\ref{additive_channel_limit}), we have ${\cal A}_\infty\equiv\sqrt{Q\Delta^{-1}_\infty{H}_\infty}=\sqrt{q_{\infty,1} \frac{1}{\delta_{\infty,1}}\big(1- \frac{\delta_{\infty,1}}{\lambda_{\infty,1}}\big)}, \frac{P^*_{\infty,1}}{q_{\infty,1}}=\frac{\lambda_{\infty,1}}{\delta_{\infty,1}}-1$ which implies ${\cal A}_\infty=\sqrt{\frac{P^*_{\infty,1}}{\lambda_{\infty,1}}}$, $\lambda_{\infty,1}= Var(X - {\mathbb E}\{X| B^{t-1}\})=Var(K_t)$, $K_t=X-{\mathbb E}\{X| B^{t-1}\}$ and $B_t= \sqrt{\frac{P^*_{\infty,1}}{\lambda_{\infty,1}}}K_t+V_t^c, t\in\mathbb{N}$. \\ 
Substituting into the encoder (\ref{additive_channel_limit}) the limiting values, $P^*_{\infty, 1}=P$  then 
\begin{align}
B_t= \sqrt{\frac{P}{\lambda_{\infty,1}}}\big(X-{\mathbb E}\{X| B^{t-1}\}\big)+V_t^c,~t=0,1,\ldots.\label{additive_channel_feedback} 
\end{align}
This is the Schalkwijk-Kailath coding scheme \cite{schalkwijk-kailath1966} of a scalar Gaussian RV $X$ in which an encoder is designed to achieve the capacity of memoryless AWGN channel with feedback. \\
\begin{remark}{\ \\}
The above JSCC designs  based on symbol-by-symbol transmission of scalar sources over  AWGN channels,  can be generalized to vector sources over vector AWGN channels.
\end{remark}
\section{Bounds on OPTA by Noncausal and Causal Codes}\label{bounds_opta} 

\par In this section, we show that the nonanticipative RDF is a lower bound on the OPTA by causal codes developed in \cite{neuhoff1982}.\\
\noi Consider a causal source code \cite{neuhoff1982} and define the average fidelity by
\begin{align*}
{d}^{+}(x,y)\tri\limsup_{k\longrightarrow\infty}\frac{1}{n+1}\mathbb{E}\Big\{d_{0,n}(x^n,y^n)\Big\},~d_{0,n}(x^n,y^n)\tri\sum_{i=0}^n\rho(x_i,y_i).
\end{align*}
Let $l_n(x^\infty)$ denote the total number of bits received at the decoder at the time it reproduces the output sequence $\{Y_n:~n\in\mathbb{N}\}$, when the source is $\{X_n:~n\in\mathbb{N}\}$. In \cite{neuhoff1982} the average rate of the encoder-decoder pairs using causal reproduction coders is measured by $\limsup_{n\longrightarrow\infty}\frac{1}{n+1}\mathbb{E}\Big\{l_n(X^\infty)\Big\}$.\\
\noi Moreover, given a source $\{X_n:~n=0,1,\ldots\}$ the OPTA by causal codes subject to fidelity is given by \cite{neuhoff1982} 
\begin{align}
r^{c,+}(D)\tri\inf_{\substack{y_i:~y_i=f_i(x^i),\forall{i}\in\mathbb{N}\\~f_i~is~causal,~\forall{i}\in\mathbb{N},~{d}^{+}(x,y)\leq{D}}}\limsup_{n\longrightarrow\infty}\frac{1}{n+1}\mathbb{E}\Big\{l_n(X^\infty)\Big\}\label{equation702}.
\end{align}
\noi Causal codes based on \cite{neuhoff1982} are analyzed and further generalized in \cite{weissman-merhav2005} for stationary ergodic sources, under a variety of side information available at the encoder and decoder. Although  expression (\ref{equation702}) is very attractive, its computation for general sources is very difficult.\\ 
\noi Next, we show that the OPTA by causal codes is bounded below by the expression of information nonanticipative RDF rate. Consider the joint distribution defined by $P_{X^n}(dx^n)$, and a reproduction distribution $\overrightarrow{Q}_{Y^n|X^n}(\cdot|x^n)$, so that the randomized coders are consistent with Definition~\ref{causal_reproduction_coder}.
Then, by data processing inequality we have the following bounds.
\begin{align}
\mathbb{E}\Big\{l_n(X^\infty)\Big\}\geq{H}(Y^n)\geq\sum_{i=0}^n\Big\{H(Y_i|Y^{i-1})-H(Y_i|Y^{i-1},X^i)\Big\}\stackrel{(b)}=\mathbb{I}_{X^n\rightarrow{Y^n}}(P_{X^n},\overrightarrow{Q}_{Y^n|X^n})\label{equation750}
\end{align}
where $(b)$ follows from the fact that the joint distribution is defined by $P_{X^n}(dx^n)$ and the conditional reproduction distribution $\overrightarrow{Q}_{Y^n|X^n}(\cdot|x^n)$. Therefore, by taking the infimum of the RHS of (\ref{equation750}) over $\overrightarrow{Q}_{Y^n|X^n}(\cdot|x^n)\in\overrightarrow{\cal Q}_{0,n}(D)$ and its left side of reproduction codes as in (\ref{equation702}) we obtain
\begin{align}
r^{c}_{0,n}(D)&\tri\inf_{\substack{\{y_i:~y_i=f_i(x^i),~f_i~is~causal,~i=0,1,\ldots,n\}\\\frac{1}{n+1}\mathbb{E}\{d_{0,n}(x^n,y^n)\}\leq{D}}}\frac{1}{n+1} \mathbb{E}\Big\{l_n(X^\infty)\Big\}\geq\frac{1}{n+1}{R}^{na}_{0,n}(D)\stackrel{(c)}\geq\frac{1}{n+1}{R}_{0,n}(D)\nonumber
\end{align}
where $(c)$ follows from the fact that $R^{na}_{0,n}(D)\geq{R}_{0,n}(D)$. \\
\noi In the previous bounds we can first take $\limsup_{n\longrightarrow\infty}$ and then the infimum giving
\begin{align*}
r^{c,+}(D)\geq{R}^{na,+}(D)\tri&\inf_{\substack{\overrightarrow{Q}_{Y^\infty|X^\infty}(\cdot|x^{\infty})\\ \in\overrightarrow{\cal Q}_{0,\infty}(D)}}\limsup_{n\longrightarrow\infty}\frac{1}{n+1}R^{na}_{0,n}(D)\nonumber\\
&\geq{R}^{+}(D)\tri\inf_{\substack{\overrightarrow{Q}_{Y^\infty|X^\infty}(\cdot|x^{\infty})\\ \in\overrightarrow{\cal Q}_{0,\infty}(D)}}\limsup_{n\longrightarrow\infty}\frac{1}{n+1}R(D).
\end{align*}
Therefore, the information nonanticipative RDF, $R^{na}(D)$, and rate $R^{na,+}(D)$, are lower bounds on $r^{c,+}(D)$, the OPTA by causal codes, and upper bounds to the classical RDF and rate $R^{+}(D)$.
\vspace*{0.2cm}

\noi{\it Bounds on the OPTA by causal codes for BSMS($p$).} 
\par The classical RDF for the BSMS($p$) is only known for the specific distortion region $0\leq{D}\leq{D}_c$ \cite{gray1971},
 and is given by
\begin{align}
R(D)=H(p)-H(D) \ \text{if} ~D\leq{D}_c=\frac{1}{2}\Big(1-\sqrt{1-\big(\frac{p}{q}\big)^2}\Big), \ p\leq 0.5.\label{graybound}
\end{align}

\begin{figure}
        \centering
        \begin{subfigure}[b]{0.5\textwidth}
                \includegraphics[width=\textwidth]{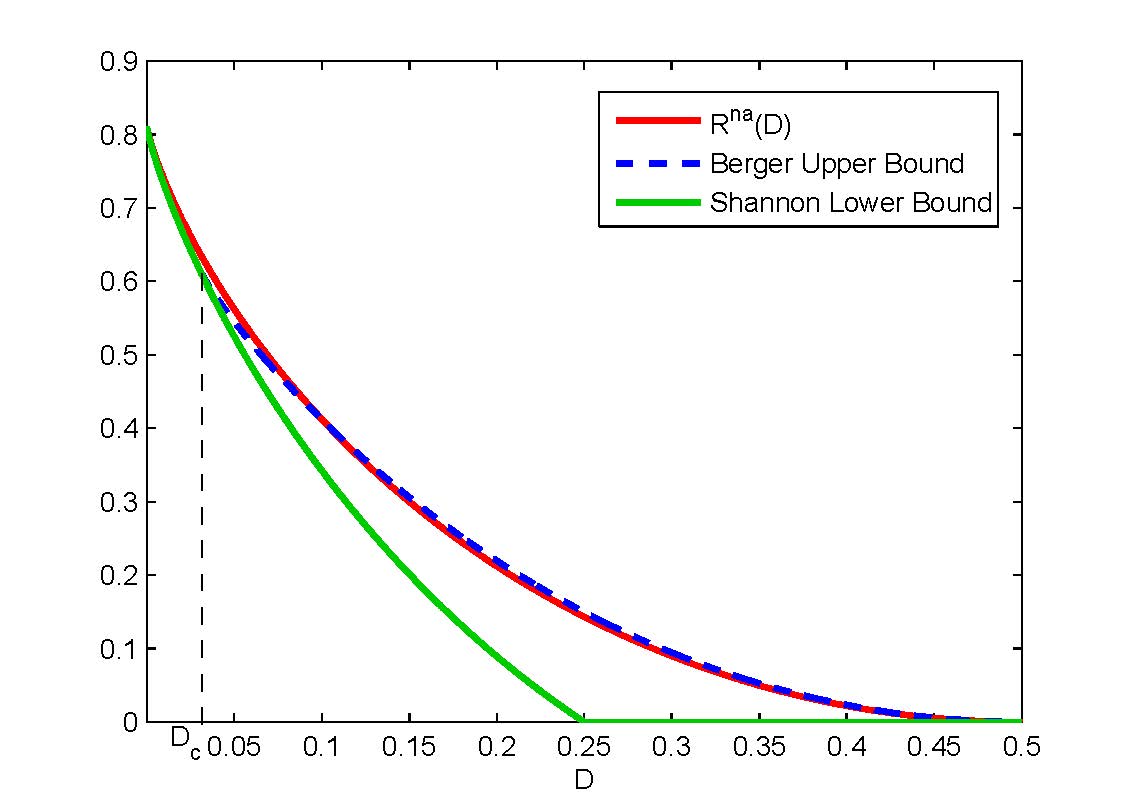}
                \caption{p=0.25 ($D_C=0.0286$).}
                \label{fig:p025}
        \end{subfigure}%
        \begin{subfigure}[b]{0.5\textwidth}
                \includegraphics[width=\textwidth]{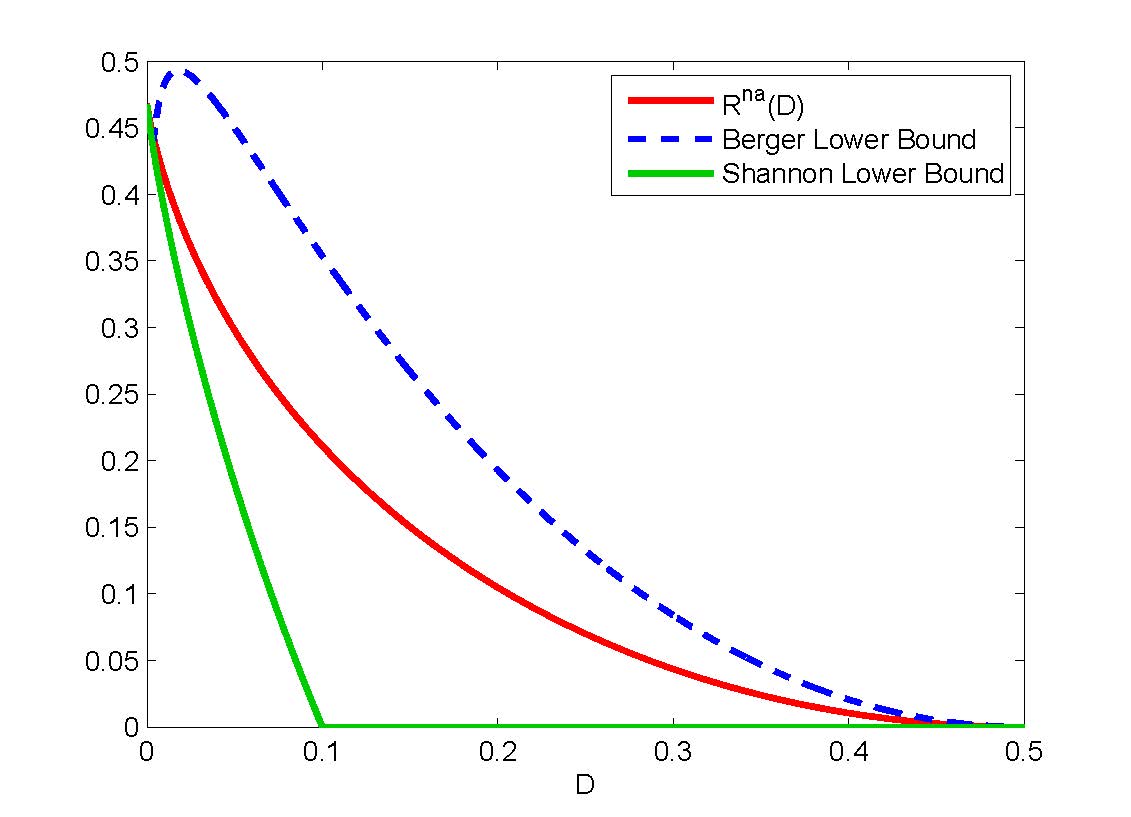}
                \caption{p=0.1 ($D_C=0.0031$).}
                \label{fig:p01}
        \end{subfigure}
        \caption{Bounds for BSMS($p$). Shannon's lower bound is tight for $D<D_C$ \cite{gray1971}.}\label{figmakrov}
\end{figure}

\noi For the remainder of the distortion region $D>{D}_c$ only upper and lower bounds on $R(D)$ are known \cite{berger1977}. In fact, it is  shown by Gray \cite{gray1971} that (\ref{graybound}) is also a lower bound for $R(D)$ in the region  $D_c \leq D \leq \frac{1}{2}$ and is equivalent to the Shannon lower bound. Our expression of the nonanticipative RDF provides an upper bound on the classical RDF for all possible values of $D$, $0\leq D\leq 0.5$. Figure~\ref{figmakrov} shows the upper bound derived by Berger \cite[equations 46, 47]{berger1977}, which hold for $D_c\leq{D}\leq\frac{1}{2}$, Shannon's lower bound, and the upper bound based on $R^{na}(D)$. Shannon's lower bound is tight for $D \leq D_c$ \cite{gray1971}. Moreover, since $R^{na}(D)$ is nonincreasing and convex as a function of $D$, and nonincreasing for all values of $p\in[0,0.5]$ (these are easily shown), then the upper bound based on $R^{na}(D)$ is convex, when compared to Berger's upper bound which is not necessarily convex and nonincreasing (as illustrated by the blue graph in Figure~\ref{fig:p025}). Finally, we use the bound $R(D)\leq{R}^{na}(D)\leq{r}^{c,+}(D)$ to deduce that the RL of causal codes with respect to the OPTA by noncausal codes of the BSMS($p$) cannot exceed

\[ RL=R^{na}(D)-R(D) \leq \left\{
  \begin{array}{l l}
    H(m)-H(p) & \quad \text{if} \ 0\leq{D}\leq p \label{ratelosub6} \\
    H(m)-H(D) & \quad \text{if} \  p<{D}\leq 0.5.
  \end{array} \right. \]
This bound on the RL is illustrated in Figure~\ref{maximum_rate_loss}, and demonstrates the fluctuation of the RL for $p\in[0,0.5]$. It is interesting to see that the maximum value of the RL is $0.2144$ and corresponds to $ (p=0.1012,D=0.1012)$. For high resolution ($D\longrightarrow  0$), the classical RDF and the nonanticipative RDF are equivalent and equal to $H(p)$.

\begin{figure}
\centering
\includegraphics[scale=0.7]{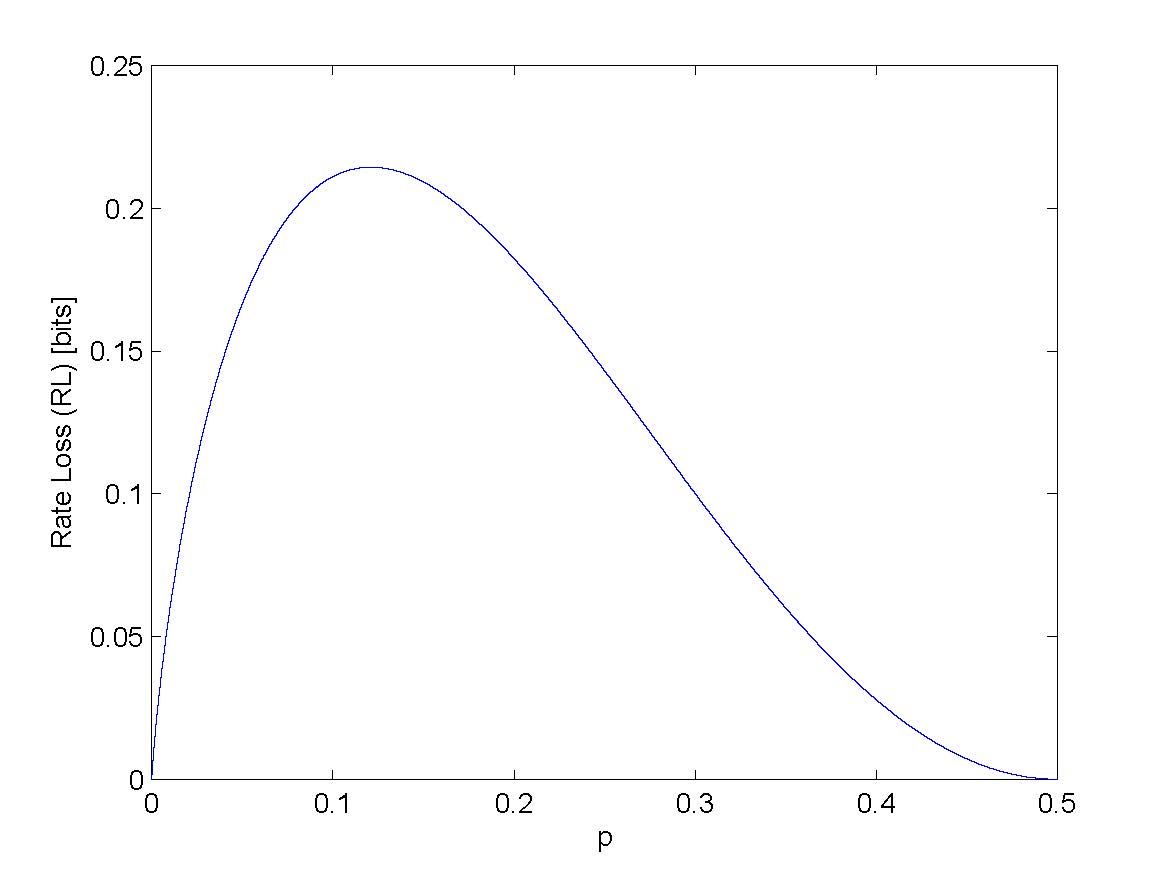}
\caption{Upper bound on the Rate Loss (RL) for the for BSMS, $p\in[0,0.5]$.}
\label{maximum_rate_loss}
\end{figure}
 
\vspace*{0.2cm}

\noi{\it Bounds on multidimensional partially observed Gaussian source.}

\par For the multidimensional partially observable stationary Gaussian-Markov source given in (\ref{equation1o}), the RL of causal codes with respect to $R(D)$ is at most $R^{na}(D)-R(D)$ bits/sample, where $R^{na}(D)$ is given in Theorem~\ref{solution_gaussian}, while the expression for $R(D)$ is found in \cite{berger}. On the other hand, $R^{na}(D)-R(D)$ is the RL of zero-delay codes with respect to the values of $R(D)$. To facilitate the computation of RL of zero delay codes with respect to those  of $R(D)$, one can work in frequency domain, by deriving the equivalent expression of $R^{na}(D)$ using the solution given in Theorem~\ref{solution_gaussian} and Szeg\"o formulae. For scalar Gaussian stationary processes such an expression is given in \cite{pinsker-gorbunov1987,gorbunov-pinsker1991}.\\
Next, we evaluate the RL of causal codes with respect $R(D)$ by considering the first-order (scalar) Gaussian-Markov autoregressive source given by (\ref{equation111}). For this model we take $\alpha=1$ which is the model with a parametric expression for specific distortion region, discussed in \cite[Example 6.3.2.1]{berger}. Specifically, for this model, the parametric expression of $R(D)$ is
$R(D)=\frac{1}{2}\log\frac{\sigma^2_W}{D},~0\leq{D}\leq{D}_c=\frac{\sigma^2_W}{4}$. For the nonanticipative RDF, this model can be explicitly computed by (\ref{12j}) by setting $\alpha=1$ as follows.
\begin{align*}
R^{na}(D)=\frac{1}{2}\log\Big(1+\frac{\sigma_W^2}{D}\Big),~0\leq{D}\leq\infty.
\end{align*} 
Hence, the RL due to causal codes for the first order Gaussian-Markov autoregressive source given by (\ref{equation111}) for $\alpha=1$ cannot exceed
\begin{align*}
RL=R^{na}(D)-R(D)=\frac{1}{2}\log\big(1+\frac{D}{\sigma^2_W}\big),~0\leq{D}\leq{D}_c=\frac{\sigma^2_W}{4}.
\end{align*}
Note for $D>\frac{\sigma^2_W}{4}$, no explicit closed form expression of $R(D)$ for the first order Gaussian autoregressive source is found in the literature. Instead, only approximation solutions are given \cite[Chapter 6]{berger}.

\section{Coding Theorem for JSCC Design Using Nonanticipative Codes}\label{coding_theorems}

\par In this section, we describe a general constructive procedure for JSCC design based on nonanticipative transmission, for sources with memory and channels with memory with and without feedback, with respect to average end-to-end distortion or excess distortion probability, operating optimally, that is, 
\begin{itemize}
\item[{\bf(1)}] the end-to-end average or excess distortion probability is achieved;
\item[{\bf(2)}] the encoder achieves the channel capacity;
\item[{\bf(3)}] $R^{na}(D)=C(P)$, where $C(P)$ is the capacity of the channel with power $P$.
\end{itemize}
The constructive procedure is a generalization of the JSCC design of  multidimensional stationary Gaussian-Markov source transmitted over the vector memoryless AWGN channel presented in Section~\ref{example:gaussian}. For memoryless sources and memoryless channels a similar method is described in \cite{kostina-verdu2012itw}, and evaluated for {\it Example-IID-BSS} and {\it Example-IID-GS}. 
 

\subsection{Nonanticipative  JSCC Design}\label{noisy_coding_theorem_unmatched}

\par The elements of the JSCC design are illustrated in Fig.~\ref{realization_nonanticipative_RDF_sbs}. We focus on JSCC design systems which are nonanticipative, that is, the encoder, channel, and decoder at each time instant $i$ process samples causally, with memory on past symbols, and without anticipation with respect to symbols occurring at times $j>i$. 
\begin{figure}[t]
\centering
\includegraphics[scale=0.70]{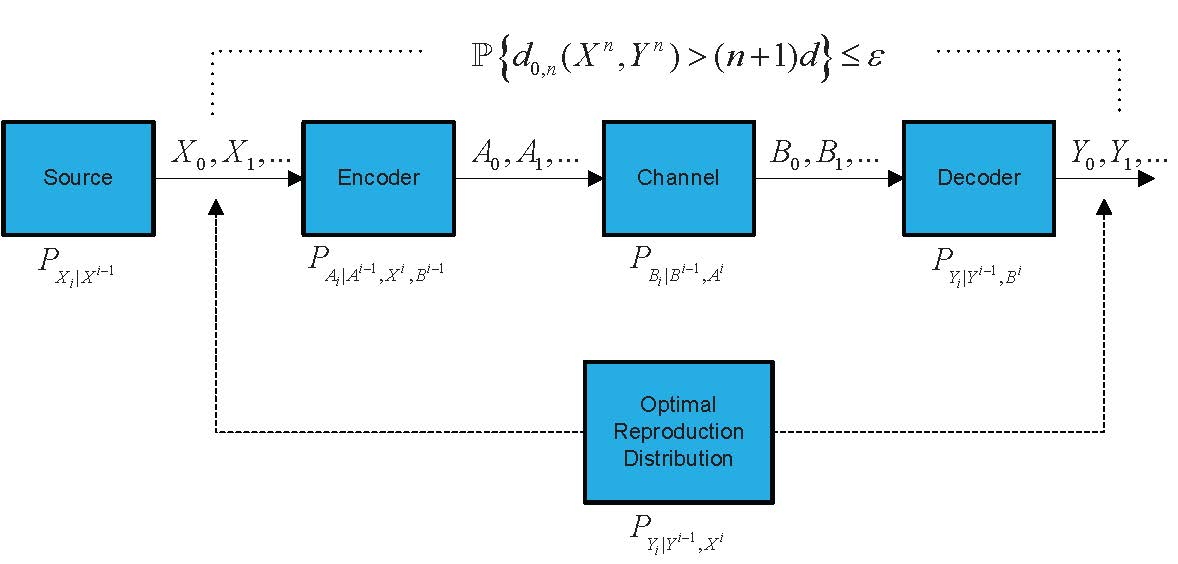}
\caption{JSCC design based on nonanticipative transmission for sources with memory.}
\label{realization_nonanticipative_RDF_sbs}
\end{figure}
Since by data processing inequality,  information nonanticipative RDF is lower than the capacity of the channel, we impose a  cost for transmitting symbols over the channel, which  is a measurable function
\begin{align}
&c_{0,n}\hspace{-0.1cm}:\hspace{-0.1cm}{\cal A}_{0,n}\hspace{-0.1cm}\times\hspace{-0.1cm}{\cal B}_{0,n-1}\hspace{-0.1cm}\longmapsto \hspace{-0.1cm}[0,\infty), \hso c_{0,n}(a^n,b^{n-1})\hspace{-0.1cm}\tri\hspace{-0.1cm}\sum_{i=0}^{n}
{\gamma}(T^ia^n,T^ib^{n-1})\label{chsteq1}
\end{align}
where $T^ia^n\subset\{a_0,\ldots,a^i\}$ and $T^ib^{n-1}\subset\{b_0,\ldots,b^{i-1}\}$.\\
\noi We use the following definition of  a nonanticipative  code.
\begin{definition}
\label{symbol_by_symbol_code_with_memory_without_anticipation}
(Nonanticipative  code){\ \\}
 An $(n,d,\epsilon,P)$ nonanticipative 
code is a tuple
\begin{align}
\Big({\cal X}_{0,n}, {\cal A}_{0,n}, {\cal B}_{0,n}, {\cal Y}_{0,n}, P_{X^n}, \overrightarrow{P}_{A^n|B^{n-1},X^n}, {\overrightarrow P}_{B^n|A^n}, \overrightarrow{P}_{Y^n|B^n}, d_{0,n}, c_{0,n}\Big)\nonumber
\end{align}
\parbox{15.5cm}{
 where
$\overrightarrow{P}_{A^n|B^{n-1},X^n}\sim\{P_{A_i|A^{i-1}, B^{i-1},X^i}(\cdot|\cdot,\cdot,\cdot): i \in\mathbb{N}^n\}$, $\overrightarrow{P}_{Y^n|B^n}\sim\{P_{Y_i|Y^{i-1},B^i}$ $(\cdot|\cdot,\cdot): i \in\mathbb{N}^n\}$, is the code (i.e., \{encoder, decoder\}), ${\overrightarrow P}_{B^n|A^n}\sim\{P_{B_i|B^{i-1}, A^{i}}(\cdot|\cdot,\cdot): i \in\mathbb{N}^n\}$ is the channel, with excess distortion probability}
\begin{align}
{\mathbb P}\Big\{d_{0,n}(X^n,Y^n)>(n+1)d\Big\}\leq\epsilon, \ \epsilon\in(0,1), \ d>0 \nonumber
\end{align}
and transmission cost 
\begin{align}
\frac{1}{n+1} {\mathbb E}\Big\{c_{0,n}(A^n,B^{n-1})\Big\}\leq P, \  P>0\nonumber
\end{align}
where ${\mathbb P}\{\cdot\}$ and $\mathbb{E}\{\cdot\}$ are taken with respect to the joint distribution $P_{X^n,A^n,B^n,Y^n}(d{x}^n,d{a}^n,d{b}^n,d{y}^n)$ induced by $\{$source, encoder, channel, decoder$\}$ .\\
\noi An uncoded nonanticipative code, denoted by $(n,d,\epsilon)$, is a subset of an $(n,d,\epsilon,P)$ nonanticipative code in which an encoder and decoder are identity maps,
$P_{A_i|A^{i-1},B^{i-1},X^i}$ $(da_i|a^{i-1},b^{i-1},x^i)=\delta_{X_i}(d{a}_i)$, $P_{Y_i|Y^{i-1},B^{i}}(dy_i|y^{i-1},b^{i})=\delta_{B_i}(d{y}_i)$, that is, $A_i=X_i$, $Y_i=B_i$, $i=0,1,\ldots,n$.
\end{definition}

\noi The well-known {\it Example-IID-BSS} of JSCC design utilizes an uncoded nonanticipative code, while {\it Example-IID-GS} utilizes an $\{$encoder, decoder$\}$ pair, which scale their input (see \cite{kostina-verdu2012itw}).\\
\noi Next, we define the minimum excess distortion as follows.
\begin{definition}(Minimum excess distortion)
\label{minimu_excess_distortion}{\ \\}
The minimum excess distortion achievable by a nonanticipative
code $(n,d,\epsilon,P)$ is defined by
\begin{align}
D^o(n,\epsilon, P)\tri\inf\Big\{d:  \exists(n,d,\epsilon,P)\  \ \mbox{nonanticipative code}\Big\}. \label{eq.10101}
\end{align}
\noi For uncoded nonanticipative  code (\ref{eq.10101}) is replaced by 
\begin{align}
\bar{D}^o(n,\epsilon)\tri\inf\big\{d: \exists(n,d,\epsilon) \ \ \mbox{nonanticipative code}\big\}.\label{anneanne4}
\end{align}
\end{definition}
\noi Note that in our definition of nonanticipative code $(n,d,\epsilon,P)$ we have assumed indirectly that the channel capacity is defined by
\begin{align}
C(P)\tri\lim_{n\longrightarrow \infty}\frac{1}{n+1}C_{0,n}(P)\label{eq.nafr1111}
\end{align}
where
\begin{align}
C_{0,n}(P)\tri\sup_{\{P_{A_i|A^{i-1},B^{i-1}}(da_i|a^{i-1},b^{i-1}):~i=0,1,\ldots,n\}\in{\cal P}_{0,n}(P)}I(A^n\rightarrow B^n)\label{eq.1111}
\end{align}
and the average power constraint is
\begin{align}
{\cal P}_{0,n}(P)\tri\Big\{\{P_{A_i|A^{i-1},B^{i-1}}(da_i|a^{i-1},b^{i-1}):~i=0,1,\ldots,n\}:\frac{1}{n+1}{\mathbb E}\{c_{0,n}(a^n,b^{n-1})\}\leq P\Big\}. \nonumber
\end{align}
Here $I(A^n\rightarrow B^n)$ is the directed information from $A^n$ to $B^n$ defined by 
\begin{align}
I(A^n\rightarrow B^n)\tri\sum_{i=0}^{n}I(A^i;B_i|B^{i-1}). \label{annofodiin8}
\end{align}

\noi Since we consider nonanticipative transmission based on Definition~\ref{symbol_by_symbol_code_with_memory_without_anticipation}, we also define the notion of probabilistic realization of the optimal nonanticipative reproduction distribution corresponding to $R^{na}(D)$ based on nonanticipative processing of information by the  encoder, channel, decoder, that is,  processing symbols causally, as follows (see Fig.~\ref{realization_nonanticipative_RDF_sbs}).
\begin{definition}\label{realdef}
(Probabilistic Realization){\ \\}
Given a source $\{P_{X_i|X^{i-1}}$ $(d{x}_i|x^{i-1}):  i \in {\mathbb N}^n\}$, then a channel
$\{P_{B_i|B^{i-1},A^i}$ $(d{b}_i|b^{i-1},a^i):  i \in {\mathbb N}^n\}$  is a realization of the optimal reproduction distribution $\{Q_{Y_i|Y^{i-1},X^i}^*(d{y}_i|y^{i-1},x^i):  i \in {\mathbb N}^n\}$ corresponding to $R^{na}_{0,n}(D)$, if there exists a pre-channel encoder $\{P_{A_i|A^{i-1},B^{i-1},X^i}$ $(d{a}_i|a^{i-1},b^{i-1},x^i):  i \in {\mathbb N}^n\}$ and a post-channel
decoder $\{P_{Y_i|Y^{i-1},B^{i}}$ $(d{y}_i|y^{i-1},b^{i}):  i \in {\mathbb N}^n\}$ such that
\begin{align}
{\overrightarrow Q}^*_{Y^n|X^n}(d{y}^n|x^n)=\otimes_{i=0}^n{Q}^*_{Y_i|Y^{i-1},X^i}(d{y}_i|y^{i-1},x^i)=\otimes_{i=0}^n{Q}_{Y_i|Y^{i-1},X^i}
(d{y}_i|y^{i-1},x^i)\label{nafrscmrd}
\end{align}
where the joint distribution from which the RHS of (\ref{nafrscmrd}) is obtained  is  precisely
\begin{align}
P_{X^n,A^n,B^n,Y^n}(d{x}^n,d{a}^n,d{b}^n,d{y}^n)&=\otimes_{i=0}^{n} P_{Y_i|Y^{i-1},B^i}(d{y}_i|y^{i-1},b^i)\otimes P_{B_i|B^{i-1},A^{i}}(d{b}_i|b^{i-1},a^{i})\nonumber\\
&\otimes{P}_{A_i|A^{i-1},B^{i-1},X^i}(d{a}_i|a^{i-1},b^{i-1},x^i)\otimes P_{X_i|X^{i-1}}(d{x}_i|x^{i-1}). \nonumber  
\end{align}
\noi Moreover, ${R}^{na}(D)$ is realizable if in addition the realization operates with average distortion $D$ and $\lim_{n\longrightarrow\infty}\frac{1}{n+1}{\mathbb I}_{X^n\rightarrow{Y^n}}(P_{X^n},$ $\overrightarrow{Q}^*_{Y^n|X^n})={R}^{na}(D)
\tri\lim_{n\longrightarrow \infty}\frac{1}{n+1}{R}^{na}_{0,n}(D)$.
\end{definition}

\noi The above definition of probabilistic realization is precisely the one utilized in Section~\ref{joint_source_channel_coding} for JSCC design with respect to achieving  average end-to-end distortion; it is also the one utilized to obtain the solution of the nonanticipative RDF for the multidimensional stationary Gaussian source, depicted in Fig.~\ref{communication_system}.

\noi Using the above definition of probabilistic realization we now prove achievability of the nonanticipative code with respect to excess distortion probability for sources with memory.
\begin{theorem}\label{achievability_noisy_coding_theorem}
(Achievability of nonanticipative code){\ \\}
{\bf Part A.}~(Coded transmission){\ \\}
Suppose the following conditions hold.
\begin{enumerate}
\item[(1)] ${ R}^{na}_{0,n}(D)$ has a solution and the optimal reproduction distribution is stationary.
\item[(2)] $C_{0,n}(P)$ has a solution and the maximizing distribution is stationary.
\item[(3)] The optimal stationary reproduction distribution $\overrightarrow{Q}^*_{Y^n|X^n}(dy^n|x^n)$ given by Theorem~\ref{th6} is realizable, and $R^{na}(D) = \lim_{n \longrightarrow \infty}\frac{1}{n+1}  { R}^{na}_{0,n}(D)$ is also realizable.
\item[(4)] For a given $D\in[D_{min},D_{max}]$ there exists a $P$ such that ${R}^{na}(D)=C(P)$.
\end{enumerate}
If
\begin{align}
\mathbb{P}\Big\{d_{0,n}(X^n,Y^n)>(n+1)d\Big\}\leq\epsilon, \ \ d > D \label{edp}
\end{align}
where ${\mathbb P}\{\cdot\}$ is taken with respect to $P_{Y^n,X^n}(d{y}^n,d{x}^n)= P_{X^n}(d{x}^n)\otimes{\overrightarrow Q}^*_{Y^n|X^n}(d{y}^n|x^n)$,
then there exists an $(n,d,\epsilon,P)$ nonanticipative code.\\
\noi{\bf Part B.}~(Uncoded transmission){\ \\}
Suppose the following conditions hold.

\begin{enumerate}
\item[(1)] Condition {\bf Part A.} (1) holds. 

\item[(2)] The encoder and the decoder are identity maps
(uncoded), and the channel $P_{B_i|B^{i-1},A^i}$
corresponds to $Q_{Y_i|Y^{i-1},X^i}$ (i.e., $A_i=X_i$, $Y_i=B_i$), $i=0,1,\ldots,n$.

\item[(3)] For a given $D\in[D_{min},D_{max}]$,  the expression  $\lim_{n\rightarrow\infty}\frac{1}{n+1}I(A^n\rightarrow B^n)$ corresponding to the optimal reproduction distribution of ${ R}^{na}(D)$  is finite.
\end{enumerate}
If
\begin{align}
\mathbb{P}\Big\{d_{0,n}({X^n},{Y^n})>(n+1)d\Big\}\leq\epsilon, \ \ d >D \label{nafredp1}
\end{align}
where ${\mathbb P}\{\cdot\}$ is taken with respect to $P_{Y^n,X^n}(d{y}^n,d{x}^n)=P_{X^n}(d{x}^n)\otimes{\overrightarrow Q}^*_{Y^n|X^n}(d{y}^n|x^n)$, then there exists an uncoded $(n,d,\epsilon)$ nonanticipative code.
\end{theorem}
\begin{proof}
{\bf Part A.} If conditions (1)-(3) hold then the optimal reproduction distribution is stationary, it is realizable, and this realization achieves ${ R}^{na}(D)$. By (4) ${ R}^{na}(D)=C(P)$. If (\ref{edp}) is satisfied then a nonanticipative code exists.  {\bf Part B.} This is a special case of {\bf Part A.}; by the data processing and condition {\bf Part B.}, (3),  we know that $R^{na}(D)\leq \lim_{n\rightarrow\infty}\frac{1}{n+1} I(A^n\rightarrow B^n)<\infty$.
Hence, if (\ref{nafredp1}) holds, there exists an uncoded $(n,d,\epsilon)$ nonanticipative code.
\end{proof}
\noi The method described in Theorem~\ref{achievability_noisy_coding_theorem}, {\bf Part A.}, ensures JSCC so that the channel operates at $C(P)$, and hence $R^{na}(D)$ is the minimum rate of reproducing source messages at the decoder, i.e., $R^{na}(D)=C(P)$. This noisy coding theorem is the one applied to the multidimensional Gaussian example, with respect to the average distortion instead of the excess distortion probability. The method described in Theorem~\ref{achievability_noisy_coding_theorem}, {\bf Part B.}, is simpler; find the optimal reproduction distribution of $R^{na}(D)$, then use this distribution as the channel and ensure that (\ref{nafredp1}) holds, which implies  achievability of the uncoded nonanticipative code. The only disadvantage is the loss of resources, because in general, the channel corresponding to the optimal reproduction distribution of $R^{na}(D)$ will have higher capacity than the value of $R^{na}(D)$. With respect to the terminology in \cite{gastpar2002}, this means that the source is not probabilistically matched to the channel.

\noi Below, we apply Theorem~\ref{achievability_noisy_coding_theorem}, {\bf Part A.} to multidimensional Gaussian process (\ref{equation1o}).  
\vspace*{0.2cm}\\
\noi{\it Example 1: Excess Distortion Probability of Multidimensional Gaussian Process.} Consider the multidimensional stationary Gaussian source and its nonanticipative rate $R^{na}(D)$ computed in Theorem~\ref{solution_gaussian}. For a given $D>0$, using coded transmission as in Theorem~\ref{achievability_noisy_coding_theorem}, {\bf Part A.}  the  calculation of the excess distortion probability
 ${\mathbb P}\Big\{d_{0,n-1}(X^{n-1},Y^{n-1})>n d\Big\}\leq\epsilon, \ \epsilon\in(0,1), \ d>D$, can be done using Cramer's theorem \cite{deuschel-stroock1989}  as follows. First note that all conditions of Theorem~\ref{achievability_noisy_coding_theorem} are satisfied. It remains to state how to compute the excess distortion probability. By Chernoff bound we have the following.
\begin{align*}
\mathbb{P}\Big\{\sum_{i=0}^{n-1}\rho(X_i,Y_i)>nd\Big\}=\mathbb{P}\Big\{\sum_{i=0}^{n-1}\rho(K_i,\tilde{K}_i)>nd\Big\}\leq{e}^{-\lambda{n}d}\mathbb{E}\Big\{e^{\lambda\sum_{i=0}^{n-1}\rho(K_i,\tilde{K}_i)}\Big\},~\lambda>0, d>D.
\end{align*}
Optimizing over all $\lambda>0$ then
\begin{align*}
\frac{1}{n}\log\mathbb{P}\Big\{\sum_{i=0}^{n-1}d_{0,n-1}(X^{n-1},Y^{n-1})>n d\Big\}\leq-\sup_{\lambda>0}\Big\{\lambda{d}-\frac{1}{n}\mathbb{E}\big\{e^{\lambda\sum_{i=0}^{n-1}\rho(K_i,\tilde{K}_i)}\big\}\Big\},~d>D.
\end{align*}
Clearly, the rate function is defined by
\begin{align*}
I_{0,n-1}(d)=\sup_{\lambda>0}\Big\{\lambda{d}-\frac{1}{n}\mathbb{E}\big\{e^{\lambda\sum_{i=0}^{n-1}\rho(K_i,\tilde{K}_i)}\big\}\Big\}\stackrel{(a)}\geq {0}, d>D
\end{align*}
where $(a)$ follows from the properties of the rate function (see \cite[Lemma 1.2.3, p. 3]{deuschel-stroock1989}). That is, the bound for $d>D$ is nontrivial. Indeed, for fixed $n-1$, it can be shown that  $I_{0,n-1}(d)$ is convex, non-decreasing function of $d \in [D, \infty]$.  Next, we show how to compute $\mathbb{E}\big\{e^{\lambda\sum_{i=1}^{n-1}\rho(K_i,\tilde{K}_i)}\big\}$, when $\rho(K_i,\tilde{K}_i)=||K_i-\tilde{K}_i||_{\mathbb{R}^p}$.
By (\ref{scalar1}), then
\begin{align}
e_i&=\tilde{K}_i-K_i=E^{tr}_\infty{H}_\infty{E}_{\infty}K_i+E^{tr}_\infty{\cal B}_\infty{V}^c_i-K_i=\big(E^{tr}_\infty{H}_\infty{E}_{\infty}-I\big)K_i+E^{tr}_\infty{\cal B}_\infty{V}^c_i,~i=0,1,\ldots.\label{error_exponent1}
\end{align}
By (\ref{equation52}), then 
\begin{align}
K_{i}=X_{i}-\hat{X}_{i|i-1}=CZ_{i}+NV_{i}-C\hat{Z}_{i|i-1}=C\big(Z_{i}-\hat{Z}_{i|i-1}\big)+NV_{i}.\label{error_exponent1a}
\end{align}
Define the error
\begin{align}
\bar{e}_{i+1}&\tri{Z}_{i+1}-\hat{Z}_{i+1|i}=AZ_i+BW_i-A\hat{Z}_{i|i-1}-A\Sigma_\infty\big(E^{tr}_\infty{H}_\infty{E}_\infty{C}\big)^{tr}M^{-1}_\infty\big(Y_i-C\hat{Z}_{i|i-1}\big)\nonumber\\
&\stackrel{(b)}=A\bar{e}_i+BW_i-A\Sigma_\infty\big(E^{tr}_\infty{H}_\infty{E}_\infty{C}\big)^{tr}M^{-1}_\infty\big(E^{tr}_\infty{H}_\infty{E}_\infty{C}\bar{e}_i+E^{tr}_\infty{H}_\infty{E}_\infty{N}V_i+E^{tr}_\infty{\cal B}_\infty{V}^c_i\big)\nonumber\\
&=\big(A-A\Sigma_\infty\big(E^{tr}_\infty{H}_\infty{E}_\infty{C}\big)^{tr}M^{-1}_\infty{E}^{tr}_\infty{H}_\infty{E}_\infty{C}\big)\bar{e}_i+BW_i-A\Sigma_\infty\big(E^{tr}_\infty{H}_\infty{E}_\infty{C}\big)^{tr}M^{-1}_\infty{E}^{tr}_\infty{\cal B}_\infty\big({\cal A}_\infty{E}_\infty{N}V_i+V_i^c\big)\nonumber\\
&\equiv\tilde{A}\bar{e}_i+\tilde{B}_1W_i+\tilde{B}_2V_i+\tilde{B}_3V_i^c,~\bar{e}_0=Z_0-\hat{Z}_{0|-1}\label{error_exponent2}
\end{align}
where $(b)$ follows from (\ref{equation_decoder1}) and  $\tilde{B}_2=A\Sigma_\infty\big(E^{tr}_\infty{H}_\infty{E}_\infty{C}\big)^{tr}M^{-1}_\infty{E}^{tr}_\infty{H}_\infty{E}_\infty{N}$, $\tilde{B}_3=A\Sigma_\infty\big(E^{tr}_\infty{H}_\infty{E}_\infty{C}\big)^{tr}M^{-1}_\infty{E}^{tr}_\infty{\cal B}_\infty{V}_i^c$, $\tilde{B}_1=B$, $\tilde{A}=A-A\Sigma_\infty\big(E^{tr}_\infty{H}_\infty{E}_\infty{C}\big)^{tr}M^{-1}_\infty{E}^{tr}_\infty{H}_\infty{E}_\infty{C}$.\\
Clearly, $\{\bar{e}_i: i=0,1,\ldots\}$ is a Gaussian Markov process. Finally,  the computation of $\mathbb{E}\big\{e^{\lambda\sum_{i=0}^{n-1}\rho(K_i,\tilde{K}_i)}\big\}$ is as follows.
\begin{align}
\mathbb{E}\big\{e^{\lambda\sum_{i=0}^{n-1}\rho(K_i,\tilde{K}_i)}\big\}&\stackrel{(c)}=\mathbb{E}\Big\{e^{\lambda\sum_{i=0}^{n-1}||\big(E^{tr}_\infty{H}_\infty{E}_{\infty}-I\big)K_i+E^{tr}_{\infty}{\cal B}_\infty{V}^c_i||_{\mathbb{R}^p}^2}\Big\}\nonumber\\
&\stackrel{(d)}=\mathbb{E}\Big\{e^{\lambda\sum_{i=0}^{n-1}||\big(E^{tr}_\infty{H}_\infty{E}_{\infty}-I\big)(C\bar{e}_i+NV_i)+E^{tr}_\infty{\cal B}_\infty{V}^c_i||_{\mathbb{R}^p}^2}\Big\}\label{error_exponent3}
\end{align}
where $(c)$ follows due to (\ref{error_exponent1}), $(d)$ follows by (\ref{error_exponent1a}), and $\bar{e}_{i+1}$ is given by (\ref{error_exponent2}).
The expectation in (\ref{error_exponent3}) can be computed explicitly due to Gaussianity of $\{\bar{e}_i:~i=0,1,\ldots\}$ and its independence of the Gaussian processes $\{V_i:~i=0,1,\ldots\}$, $\{V_i^c:~i=0,1,\ldots\}$. A simple procedure to compute the continuous time analogue of (\ref{error_exponent3}) is given in \cite[Section III.B]{charalambous-djouadi-denic2005ieeeit}. For any of the JSCC design of scalar sources, the above calculation can be easily carried out. \\

\noi Below we consider the BSMS($p$) and we  compute the excess distortion probability, when  the source is not probabilistically matched to the channel, using uncoded transmission (sub-optimal), ensuring reliable communication. This observation is also pointed out in \cite{kostina-verdu2012itw} for the two well-known examples of JSCC design, the {\it Example-IID-BSS} and the {\it Example-IID-GS}. Recently, in \cite{kourtellaris-charalambous-boutros2015isit} the JSCC design of the BSMS($p$) is shown over a unit memory channel (with certain symmetry) with transmission cost, with and without feedback encoding. We do not present the optimal JSCC design, due to space limitations. However, the calculations presented below apply directly to the JSCC design in \cite{kourtellaris-charalambous-boutros2015isit}. \\

\noi{\it Example 2: Excess Distortion Probability of BSMS($p$).} Consider  the BSMS($p$) and its nonanticipative rate $R^{na}(D)$ computed in Theorem~\ref{marex1}. For a given $D>0$, using uncoded transmission as in Theorem~\ref{achievability_noisy_coding_theorem}, {\bf Part B.}  the exact calculation of the excess distortion probability
 ${\mathbb P}\Big\{d_{0,n-1}(X^{n-1},Y^{n-1})>n d\Big\}\leq\epsilon, \ \epsilon\in(0,1), \ d>D$, is not as straightforward as it in 
 the case of the IID Bernoulli source \cite{kostina-verdu2012}.  However, if we can show that the joint process $\{(X_i, Y_i):i=0,1,\ldots\}$  induced by the optimal reproduction distribution and the BSMS($p$) with alphabet  $\overline{\Sigma} \tri \big\{(x,y): x\in\{0,1\}, y\in\{0,1\}\big\}$ is Markov, then we can find an upper bound for the excess distortion probability, for finite but large enough $n$ so the bound is non-trivial.   Express the distortion function as follows.
\bes
z_i\tri(x_i, y_i), \hst S_n\tri\sum_{i=0}^{n-1}f(z_i),~f(z) \tri \rho(x,y)=x\oplus{y}, \ \ i=0, 1, \ldots, n-1
\ees
 Introduce the mapping $\phi:  \overline{\Sigma} \mapsto { \Sigma} \tri \{1, 2, 3, 4\}, (x,y) \in \overline{\Sigma}$ such that under $\phi$,  $(0,0) \mapsto 1,  (0,1) \mapsto 2, (1,0) \mapsto 3,  (1,1) \mapsto 4$.

\begin{theorem}
The joint process $\{Z_i\tri(X_i, Y_i):~i=0,1,\ldots\}$ generated by the optimal reproduction distribution $P^*_{Y_i|Y_{i-1},X_{i}}({y}_i|y_{i-1},x_{i})$ and the BSMS($p$), $P_{X_i|X_{i-1}}(x_i|x_{i-1})$ is Markov, that is, $P^*_{Z_i|Z^{i-1}}(z_i|z^{i-1})= P^*_{Z_i|Z_{i-1}}(z_i|z_{i-1}), i=0,1,\ldots$,  and  its  transition probability matrix denoted by ${\bf \Pi}=\{\pi(i,j)\equiv P^*_{Z_i|Z_{i-1}}(i|j): (i,j) \in  \Sigma\times\Sigma\}$  is given by
\begin{align}
{\bf \Pi} 
=
\bbordermatrix{~ &  &  &  &  \cr
                  & \alpha(1-p) & \beta(1-p) &{\alpha}p & {\beta}p \vspace{0.3cm} \cr
                  & (1-\alpha)(1-p) & (1-\beta)(1-p)&(1-\alpha)p & p(1-\beta)  \vspace{0.3cm} \cr
                  & (1-\beta)p & (1-\alpha)p & (1-\beta)(1-p)&  (1-\alpha)(1-p) \vspace{0.3cm} \cr
                  & {\beta}p & {\alpha}p & \beta(1-p)&   \alpha(1-p) \cr} \label{cachjodis}
\end{align}
where each column $\{\pi(\cdot,j): j \in \{1, 2, 3, 4\}\}$ is a probability vector, where $\alpha$ and $\beta$ are defined in (\ref{def-alphabeta}).
\end{theorem}

\begin{proof}
Recall the solution of the BSMS($p$) in Theorem~\ref{marex1}. Then for $i=1, 2, \ldots, $ we have
\begin{align}
P^*_{X_i, Y_i|X^{i-1},Y^{i-1}}({x}_i, y_i|x^{i-1},y^{i-1})
&=Q^*_{Y_i|Y^{i-1},X^{i}}({y}_i|y^{i-1},x^{i})P_{X_i|X^{i-1},Y^{i-1}}({x}_i|x^{i-1},y^{i-1})\nonumber\\
&=Q^*_{Y_i|Y_{i-1},X_{i}}({y}_i|y_{i-1},x_{i})P_{X_i|X_{i-1}}({x}_i|x_{i-1}).\nonumber\\
&=P^*_{X_i,Y_i|X_{i-1},Y_{i-1}}({x}_i,{y}_i|x_{i-1},y_{i-1}).\nonumber
\end{align}
This shows that the joint process is Markov. By simple algebra, we obtain (\ref{cachjodis}).
\end{proof}
\noi Since the process $\{Z_i: i=0,1, \ldots\}$ is Markov we can apply Hoeffding's  inequality \cite{glynn2002}, which bounds the probability of a function of the Markov chain to obtain an upper bound for the excess distortion probability as follows. \\
\begin{align}
\mathbb{P}^*\Big(\frac{S_n-\mathbb{E}[S_n]}{n}\geq\gamma\Big)\leq \exp\Big(-\frac{{\lambda}^2
(n\gamma -2\|f\|m/{\lambda})^2}{2n{\|f\|}^2m^2}\Big), \ \  n>2{\|f\|}m/(\lambda\gamma)  \label{HF1}
\end{align}
where ${\|f\|}=1$, $m=1$ and
$\lambda=\min\{p,1-p\}\min\{\alpha,\beta,1-\alpha,1-\beta\}$. Moreover, since ${\bf\Pi}$ defined by (\ref{cachjodis}) is a Reversible Finite State Markov Chain (R-FSMC), a tighter bound is obtained by applying \cite{leon-perron2004}, and it is given by
\begin{align}
\mathbb{P}^*\Big(\frac{S_n-\mathbb{E}[S_n]}{n}\geq\gamma\Big)\leq \exp\Big(-2\frac{1-{\lambda}_0}{1+{\lambda}_0}n{\gamma}^2\Big) \label{HF1_revers}
\end{align}
\noi where ${\lambda}_0=\max(0,{\lambda}_2)$, and ${\lambda}_2$ is the second largest eigenvalue of the transition probability matrix ${\bf\Pi}$. The fact that ${\bf\Pi}$ is R-FSMC is easily shown by verifying that $\{\pi(i,j):~(i,j)\in\Sigma\times\Sigma\}$ and its steady state probability $\{\pi(i):~i\in\Sigma\}$ satisfy $\pi(i){\pi}(j,i) =\pi(j){\pi}(i,j), \ \forall i,j$, or by applying Kolmogorov's necessary and sufficient criterion of reversibility given by \cite{kelly}
\begin{align*}
{\pi} (j_1,j_2){\pi} (j_2,j_3)\ldots{\pi} (j_n,j_1)=
{\pi} (j_1,j_n)\ldots{\pi} (j_3,j_2){\pi} (j_2,j_1).
\end{align*}
\noi This bound is illustrated in Fig.~\ref{fig_RFSMCbound_1}, $p=0.3$, $D=0.1$, $\gamma=0.1$, and illustrates how the upper bound (\ref{HF1_revers}) of the excess distortion probability changes as a function of the number of transmissions.
\begin{figure}[t]
\begin{center}
\psfrag{yl}[b]{\fontsize{9}{10}$\mathbb{P}^*\Big(\dfrac{S_n-\mathbb{E}[S_n]}{n}\geq\gamma\Big)$}
\includegraphics[scale=0.58]{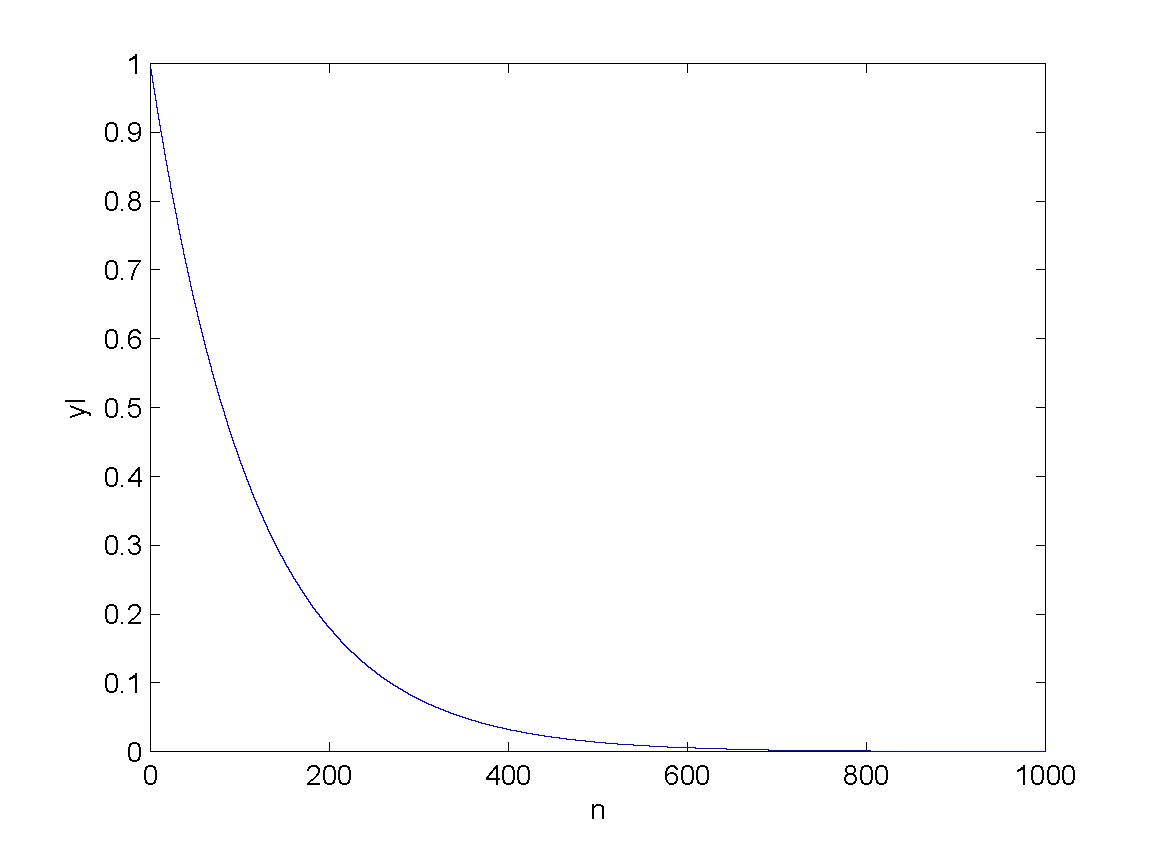}
\caption{Excess distortion probability bound (\ref{HF1_revers}) for $p=0.3,D=0.1$ and $\gamma=0.1$ ($d=D+\gamma$).}
\label{fig_RFSMCbound_1}
\end{center}
\end{figure}
\noi Next, we use results from large deviations theory to compute the error exponent of the excess distortion probability, as $n \longrightarrow \infty$.  Let $\mathbb{P}_{\sigma}^{\pi}$ denote the  joint probability associated with ${\bf \Pi}$, having an initial state $Z_0=\sigma\in\Sigma$ given by
\begin{align*}
\mathbb{P}_{\sigma}^{\pi}(Z_1=z_1,\ldots,Z_{n-1}=z_{n-1})=\pi( z_1,\sigma)\otimes_{i=2}^{n-1}\pi(z_{i},z_{i-1})
\end{align*}
and denote the expectation with respect to $\mathbb{P}_{\sigma}^{\pi}$ by $\mathbb{E}_{\sigma}^{\pi}(\cdot)$. It can be verified that ${\bf\Pi}$ is irreducible, therefore we can describe the error exponent of the excess distortion probability by the Perron-Frobenius eigenvalue of a certain non-negative matrix \cite{dembo-zeitouni1998}.\\
Let $\lambda\in\mathbb{R}$ and define the non-negative matrix ${\bf\Pi}_\lambda$ with eigenvalues
\begin{align*}
\pi_\lambda(j,i)=\pi(j,i) \exp\{{\lambda{f}(j)}\},~i,j\in\Sigma.
\end{align*}
Let $\rho({\bf \Pi}_\lambda)$ denote the Perron-Frobenius eigenvalue of ${{\bf \Pi}_\lambda}$. Then by \cite[Theorem 3.1.2]{dembo-zeitouni1998}, we have the following. For any set $\Gamma\in\mathbb{R}$ and initial state $\sigma\in\Sigma$,
\begin{align*}
-\inf_{\theta\in\Gamma^o}I(\theta)\leq\liminf_{n\longrightarrow\infty}\frac{1}{n}\log\mathbb{P}_{\sigma}^{\pi}\Big(\frac{S_n}{n}\in\Gamma\Big)\leq\limsup_{n\longrightarrow\infty}\frac{1}{n}\log\mathbb{P}_{\sigma}^{\pi}\Big(\frac{S_n}{n}\in\Gamma\Big)\leq-\inf_{\theta\in\bar{\Gamma}}I(\theta)
\end{align*}
where $\Gamma^{o}$ is the interior of $\Gamma$ and $\bar{\Gamma}$ is the closure of $\Gamma$, and $I:\mathbb{R}\longmapsto[0,\infty]$ is a convex rate function defined by
\begin{align}
I(\theta)\tri\sup_{\lambda\in\mathbb{R}}\Big\{\lambda{\theta}-\log\rho({\bf \Pi}_{\lambda})\Big\}.
\label{rate-functionck}
\end{align}
As $n$ becomes very large (i.e., $n\longrightarrow\infty$) then
\begin{align}
\mathbb{P}^{\pi}_{\sigma}\Big(\frac{S_n}{n}\geq{d}\Big)\sim\exp\Big({-n\inf_{\theta\in[d,\infty]}I(\theta)}\Big), \quad d\geq D\label{error-exponent}
\end{align}
and the exponential decay of this probability is obtained by minimizing the  rate function over $\theta \in [d, \infty]$.  Since for the evaluation of (\ref{error-exponent}), $I(\theta)$ is convex, non-decreasing function of $\theta\in[d,\infty]$, then
\begin{align*}
\lim_{n\longrightarrow\infty}\frac{1}{n}\log\mathbb{P}^{\pi}_{\sigma}\Big(\frac{S_n}{n}\geq{d}\Big)=-I(d).
\end{align*}
The graph of $I(\theta)$, for $\theta\geq{D}$ is shown in Fig.~\ref{rate_function}, and indicates the rate of exponential decay of the excess distortion probability as a function of $\theta \geq D$.

\begin{figure}
\begin{center}
{\includegraphics[scale=0.58]{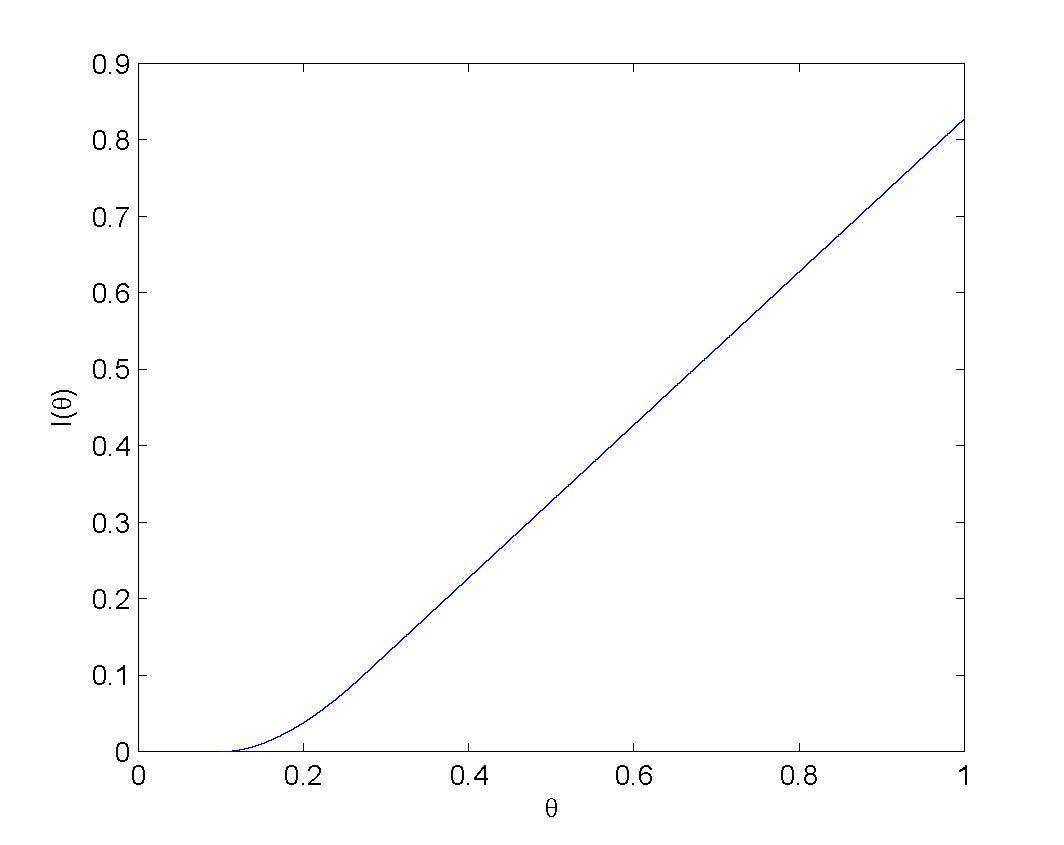}}
\caption{ Rate function (\ref{rate-functionck}) for $p=0.3$ and $D=0.1$.}
\label{rate_function}
\end{center}
\end{figure}

\subsection{Noiseless Coding Theorem}\label{sequential_causal_rdf}

\par It is straight forward to verify by invoking Lemma~\ref{equivalent_statements} and Theorem~\ref{equivalent_rdf}, that the   the coding theorem derived in \cite[Chapter 5]{tatikonda2000} for two dimensional sources $X^{n,s}\tri\{X_{i,j}:i=0,\ldots,n,j=0,\ldots,s\}\in\times_{i=0}^n\times_{j=0}^s{\cal X}_{i,j}$, where $i$ represents time index and $j$ represents spatial index, gives an operational meaning to ${R}^{na}_{0,n}(D)$ (even for finite $n$, when the process is IID in the spatial component). The connection of the coding theorem for the so-called sequential RDF derived in \cite[Chapter 5]{tatikonda2000} and nonanticipative RDF is described in \cite[Appendix G]{stavrou-kourtellaris-charalambous2014arxiv}. The aforementioned coding theorem for two-dimensional sources is recently utilized in video coding applications \cite{ma-ishwar2011}, where the authors derived a coding theorem giving an operational meaning to $R^{na}_{0,n}(D)$.


\section{Conclusion}

\par A variant of the classical RDF of the OPTA by noncausal codes, called information nonanticipative RDF, which imposes a nonanticipative or causality constraint on the optimal reproduction conditional distribution is investigated in the context of its applications in JSCC design using nonanticipative transmission schemes with respect to average and excess distortion probability, and in evaluating the RL of the OPTA by zero-delay and causal codes with respect to noncausal codes. These applications are employed to two working examples, the BSMS($p$) with Hamming distance distortion, and the multidimensional partially observed Gaussian-Markov source. It is our belief that the results derived in this paper provide a crucial step towards the complete investigation of two fundamental problems in information theory, the evaluation of nonanticipative RDF in systems where sources with memory are considered, and in nonanticipative JSCC system design.

\section*{Acknowledgement}

The authors wish to acknowledge the contribution of the associate editor and the reviewers in improving the presentation and results of the paper.

\appendices

\section{Proof of Lemma~\ref{equivalent_statements}}\label{appendix_equivalent_statements}
\par The equivalence of {\bf MC1}, {\bf MC2} and {\bf MC3} is straightforward hence it is omitted. To this end, we show equivalence of {\bf MC4} to any of {\bf MC1}, {\bf MC2} and {\bf MC3}. We proceed with the derivation, by often assuming existence of densities, which are denoted by lower case letters $\bar{p}(\cdot|\cdot)$, to avoid lengthy measure theoretic arguments.
\vspace*{0.2cm}\\
\noi{\bf MC4} $\Longrightarrow$ {\bf MC3}: Since for $i=0,\ldots,n-1$, by {\bf MC4} we have
\begin{align*}
P_{X_{i+1}^n|X^i,Y^i}(dx_{i+1}^n|x^i,y^i)=P_{X_{i+1}^n|X^i}(dx_{i+1}^n|x^i)
\end{align*}
then by integrating over ${\cal X}_{i+2,n}$ both sides of the previous identity we obtain {\bf MC3}.
\vspace*{0.2cm}\\
{\bf MC4} $\Longleftarrow$ {\bf MC3}: Since {\bf MC3} $\Longleftrightarrow$ {\bf MC2}, we show that if $X_{i+1}^n\leftrightarrow({X^i},Y^{i-1})\leftrightarrow{Y_i}$ forms a MC for $i=0,1,\ldots, n-1$, then $X_{i+1}^n\leftrightarrow{X^i}\leftrightarrow{Y^i}$ forms a MC for $i=0,1,\ldots, n-1.$ We show this by induction. First, we show that $(X_{i+1},X_{i+2})\leftrightarrow{X^i}\leftrightarrow{Y^i}$ forms a MC, or equivalently, $\bar{p}(x_{i+1},x_{i+2}|x^i,y^i)=\bar{p}(x_{i+1},x_{i+2}|x^i)$. Since
\begin{align*}
\bar{p}(x_{i+1},x_{i+2}|x^i,y^i)&=\frac{\bar{p}(x^i,x_{i+1},x_{i+2},y^i)}{\bar{p}(x^i,y^i)}=\frac{\bar{p}(y_i|y^{i-1},x^{i+2})\bar{p}(y^{i-1},x^{i+2})}{\bar{p}(x^i,y^i)}\\
&=\frac{\underbrace{\bar{p}(y_i|y^{i-1},x^i)}_{(a)}\bar{p}(x_{i+2}|x^{i+1},y^{i-1})\bar{p}(x^{i+1},y^{i-1})}{\bar{p}(x^i,y^i)}\\
&=\frac{\bar{p}(y_i|y^{i-1},x^i)\underbrace{\bar{p}(x_{i+2}|x^{i+1})}_{(b)}\bar{p}(x_{i+1}|x^i,y^{i-1})\bar{p}(x^i,y^{i-1})}{\bar{p}(y_i|y^{i-1},x^i)\bar{p}(x^i,y^{i-1})}\\
&=\bar{p}(x_{i+2}|x^{i+1})\underbrace{\bar{p}(x_{i+1}|x^i)}_{(c)}=\bar{p}(x_{i+2},x_{i+1}|x^i)
\end{align*}
where $(a)$ is implied from {\bf MC2}, while $(b)$, $(c)$ follows from {\bf MC3} $\Longleftrightarrow$ {\bf MC2}. 
Hence, {\bf MC4} holds for $n=i+2$.\\
Suppose $X_{i+1}^k\leftrightarrow{X^i}\leftrightarrow{Y^i}$ forms a MC, for some $i+2\leq{k}<n-1$. We show that it holds for $k\longrightarrow{k+1}$.
\begin{align*}
\bar{p}(x_{i+1}^{k+1}|x^i,y^i)&=\frac{\bar{p}(x_{i+1}^{k+1},x^i,y^i)}{\bar{p}(x^i,y^i)}=\frac{\bar{p}(x_{k+1}|x_{i+1}^k,x^i,y^i)\bar{p}(x^k_{i+1},x^i,y^i)}{\bar{p}(x^i,y^i)}\\
&=\frac{\underbrace{\bar{p}(x_{k+1}|x^k)}_{(d)}\bar{p}(x^k_{i+1}|x^i,y^i)\bar{p}(x^i,y^i)}{\bar{p}(x^i,y^i)}=\bar{p}(x_{k+1}|x^k)\underbrace{\bar{p}(x^k_{i+1}|x^i)}_{(e)}=\bar{p}(x_{i+1}^{k+1}|x^i)
\end{align*}
where $(d)$, $(e)$ follow from {\bf MC3} $\Longleftrightarrow$ {\bf MC2}. This completes the derivation.\qed
\section{Proof of Theorem~\ref{property3}}\label{appendix_property3}

\par Let $s$ be the the Lagrange multiplier which is part of the optimal solution which solves the information nonanticipative RDF. Then, by Theorem~\ref{th6} the optimal reproduction distribution is expressed as
\begin{align}
Q_{Y_i|Y^{i-1},X^i}^*(F_i|y^{i-1},x^i) = \int_{F_i}e^{s\rho(T^ix^n,T^iy^n)}\lambda_i(x^i,y^{i-1})P_{Y_i|Y^{i-1}}^*(dy_i|y^{i-1}),~\forall F_i \in {\cal{B}}({\cal Y}_{i}),~\forall{i}\in\mathbb{N}^n. \label{f10i}
\end{align}
By integrating (\ref{f10i}) with respect to $P^*_{0,i}(dx^i|y^{i-1})$ we obtain the expression
\begin{align*}
&\int_{{{\cal X}_{0,i}}}Q_{Y_i|Y^{i-1},X^i}^*(F_i|y^{i-1},x^i)\otimes{P}_{X^i|Y^{i-1}}^*(dx^{i}|y^{i-1})= {P}_{X^i,Y_i|Y^{i-1}}^*({{\cal X}_{0,i}}\times F_i|y^{i-1})\\
&=\int_{F_i\times {\cal X}_{0,i}} e^{s\rho(T^ix^n,T^iy^n)}\lambda_i(x^i,y^{i-1})P_{Y_i|Y^{i-1}}^*(dy_i|y^{i-1})\otimes{P}_{X^i|Y^{i-1}}^*(dx^{i}|y^{i-1}),~\forall{F_i}\in{\cal B}({\cal Y}_i),~i\in\mathbb{N}^n.
\end{align*}
Moreover, $\forall F_i \in {\cal{B}}({\cal Y}_{i}),~i\in\mathbb{N}^n$,
\begin{align}
P_{Y_i|Y^{i-1}}^*(F_i|y^{i-1})&=\int_{F_i}P_{Y_i|Y^{i-1}}^*(dy_i|y^{i-1})={P}_{X^i,F_i|Y^{i-1}}^*({{\cal X}_{0,i}}\times F_i|y^{i-1})\nonumber\\
 &= \int_{F_i \times{\cal X}_{0,i}} e^{s\rho(T^ix^n,T^iy^n)}\lambda_{i}(x^i,y^{i-1}){P}_{0,i}^*(dx^{i}|y^{i-1})\otimes {P}_{Y_i|Y^{i-1}}^*(dy_i|y^{i-1}).\label{equation800}
\end{align}
Utilizing (\ref{equation800}) we finally obtain
\begin{align*}
\int_{{{\cal X}_{0,i}}} e^{s
\rho(T^ix^n,T^iy^n)} \lambda_i(x^i,y^{i-1}){P}_{X^i|Y^{i-1}}^*(dx^{i}|y^{i-1}) =1,~~~P_{Y_i|Y^{i-1}}^*-a.s.,~\forall{i}\in\mathbb{N}.
\end{align*}
This completes the derivation.\qed
\section{Proof of Theorem~\ref{alternative_expression}}\label{appendix_alternative_expression}

\par Let $s\leq{0}$, $\lambda\in\Psi_s$ and $\overrightarrow{Q}_{Y^n|X^n}(\cdot|x^n)\in\overrightarrow{\cal Q}_{0,n}(D)$ be given. Then, using the fact that 
\begin{align*}
\frac{1}{n+1}\sum_{i=0}^n\int_{{\cal X}_{0,i}\times{\cal Y}_{0,i}}\rho(T^ix^n,T^iy^n)(P_{X^i}\otimes\overrightarrow{Q}_{Y^i|X^i})(dx^i,dy^i)\leq{D}
\end{align*}
gives
\begin{align}
&\mathbb{I}_{X^n\rightarrow{Y^n}}(P_{X^n},\overrightarrow{Q}_{Y^n|X^n})-sD(n+1)-\sum_{i=0}^n\int_{{\cal X}_{0,i}\times{\cal Y}_{0,i-1}}\log\Big(\lambda_i(x^i,y^{i-1})\Big)(P_{X^i}\otimes\overrightarrow{Q}_{Y^{i-1}|X^{i-1}})(dx^i,dy^{i-1})\nonumber\\
&\geq\sum_{i=0}^n\int_{{\cal X}_{0,i}\times{\cal Y}_{0,i}}\log\Big(\frac{Q_{Y_i|Y^{i-1},X^i}(dy_i|y^{i-1},x^i)}{P_{Y_i|Y^{i-1}}(dy_i|y^{i-1})}\Big)(P_{X^i}\otimes\overrightarrow{Q}_{Y^i|X^i})(dx^i,dy^{i})\nonumber\\
&\qquad-s\sum_{i=0}^n\int_{{\cal X}_{0,i}\times{\cal Y}_{0,i}}\rho(T^ix^n,T^iy^n)(P_{X^i}\otimes\overrightarrow{Q}_{Y^i|X^i})(dx^i,dy^i)\nonumber\\
&\qquad-\sum_{i=0}^n\int_{{\cal X}_{0,i}\times{\cal Y}_{0,i}}\log\Big(\lambda_i(x^i,y^{i-1})\Big)(P_{X^i}\otimes\overrightarrow{Q}_{Y^i|X^i})(dx^i,dy^{i})\nonumber\\
&=\sum_{i=0}^n\int_{{\cal X}_{0,i}\times{\cal Y}_{0,i}}\log\bigg(\frac{Q_{Y_i|Y^{i-1},X^i}(dy_i|y^{i-1},x^i)e^{-s\rho(T^ix^n,T^iy^n)}}{P_{Y_i|Y^{i-1}}(dy_i|y^{i-1})\lambda_i(x^i,y^{i-1})}\bigg)(P_{X^i}\otimes\overrightarrow{Q}_{Y^i|X^i})(dx^i,dy^{i})\nonumber
\end{align}
\begin{align}
&=\sum_{i=0}^n\int_{{\cal X}_{0,i-1}\times{\cal Y}_{0,i-1}}\Bigg\{\int_{{\cal X}_i\times{\cal Y}_i}\log\bigg(\frac{Q_{Y_i|Y^{i-1},X^i}(dy_i|y^{i-1},x^i)e^{-s\rho(T^ix^n,T^iy^n)}}{P_{Y_i|Y^{i-1}}(dy_i|y^{i-1})\lambda_i(x^i,y^{i-1})}\bigg)\nonumber\\
&\qquad\qquad Q_{Y_i|Y^{i-1},X^i}(dy_i|y^{i-1},x^i)\otimes{P}_{X_i|X^{i-1}}(dx_i|x^{i-1})\Bigg\}\otimes{P}_{X^{i-1},Y^{i-1}}(dx^{i-1},dy^{i-1})\nonumber\\
&\stackrel{(a)}\geq\sum_{i=0}^n\int_{{\cal X}_{0,i-1}\times{\cal Y}_{0,i-1}}\Bigg\{\int_{{\cal X}_i\times{\cal Y}_i}\bigg(1-\frac{e^{s\rho(T^ix^n,T^iy^n)}P_{Y_i|Y^{i-1}}(dy_i|y^{i-1})\lambda_i(x^i,y^{i-1})}{Q_{Y_i|Y^{i-1},X^i}(dy_i|y^{i-1},x^i)}\bigg)\nonumber\\
&\qquad\qquad Q_{Y_i|Y^{i-1},X^i}(dy_i|y^{i-1},x^i)\otimes{P}_{X_i|X^{i-1}}(dx_i|x^{i-1})\Bigg\}\otimes{P}_{X^{i-1},Y^{i-1}}(dx^{i-1},dy^{i-1})\nonumber\\
&=\sum_{i=0}^n\Bigg\{1-\int_{{\cal X}_{0,i-1}\times{\cal Y}_{0,i-1}}\int_{{\cal X}_i\times{\cal Y}_i}{e^{s\rho(T^ix^n,T^iy^n)}\lambda_i(x^i,y^{i-1})}P_{Y_i|Y^{i-1}}(dy_i|y^{i-1})\nonumber\\
&\qquad\qquad\otimes{P}_{X_i|X^{i-1}}(dx_i|x^{i-1})\otimes{P}_{X^{i-1},Y^{i-1}}(dx^{i-1},dy^{i-1})\Bigg\}\nonumber\\
&=\sum_{i=0}^n\Bigg\{1-\int_{{\cal Y}_i}P_{Y_i|Y^{i-1}}(dy_i|y^{i-1})\int_{{\cal Y}_{0,i-1}}P_{Y^{i-1}}(dy^{i-1})\bigg(\int_{{\cal X}_{0,i}}{e^{s\rho(T^ix^n,T^iy^n)}\lambda_i(x^i,y^{i-1})}\otimes{P}_{X^i|Y^{i-1}}(dx^{i}|y^{i-1})\bigg)\Bigg\}\nonumber\\
&\stackrel{(b)}\geq\sum_{i=0}^n\Big(1-\int_{{\cal Y}_{0,i}}P_{Y^i}(dy^i)\Big)=0\nonumber
\end{align}
where $(a)$ follows from the inequality $\log{x}\geq1-\frac{1}{x},~x>0$, and $(b)$ follows from (\ref{equation-ls}). Hence,
\begin{align}
&R^{na}_{0,n}(D)\stackrel{(c)}\geq\sup_{s\leq{0}}\sup_{\lambda\in\Psi_s}\Big\{sD(n+1)+\sum_{i=0}^n\int_{{\cal X}_{0,i}\times{\cal Y}_{0,i-1}}\log\Big(\lambda_i(x^i,y^{i-1})\Big)P_{X^{i-1},Y^{i-1}}(dx^{i-1},dy^{i-1})\otimes{P}_{X_i|X^{i-1}}(dx_i|x^{i-1})\Big\}.\nonumber
\end{align}
However, equality in $(c)$ holds if 
\begin{align*}
\lambda_i(x^i,y^{i-1})\tri\Bigg(\int_{{\cal Y}_i}e^{s\rho(T^ix^n,T^iy^n)}P^*_{Y_i|Y^{i-1}}(dy_i|y^{i-1})\Bigg)^{-1},~i=0,1,\ldots,n.
\end{align*}
This completes the derivation.\qed

\section{Proof of Theorem~\ref{solution_gaussian}}\label{appendix_solution_gaussian}

\par The derivation is based on showing that $R^{na}(D)$ is bounded above and below by the RHS of (\ref{equa.13}). The lower bound is obtained by using Theorem~\ref{alternative_expression}, to derive a lower bound analogous to Shannon's lower bound. Define
\begin{align*}
H_\infty=\lim_{t\longrightarrow\infty}H_t,~H_t\tri{d}iag\{\eta_{t,1},\ldots,\eta_{t,p}\},~\eta_{t,i}=1-\frac{\delta_{t,i}}{\lambda_{t,i}},~i=1,\ldots,p,~t\in\mathbb{N}.
\end{align*}
\noi Consider the additive noisy channel with feedback of the form
\begin{align}
\tilde{K}_t&=E_t^{tr}H_tE_t\Big(X_t-\mathbb{E}\big\{X_t|\sigma\{Y^{t-1}\}\big\}\Big)+E_t^{tr}{\cal B}_tV^c_t=E_t^{tr}H_tE_tK_t+E_t^{tr}{\cal B}_tV^c_t,~t\in\mathbb{N}^n\label{equation95ii}
\end{align}
where $\{V^c_t:t\in\mathbb{N}\}$ is an independent Gaussian zero mean process with covariance $cov(V^c_t)=Q=diag\{q_1,\ldots,q_p\}$, and $\{{\cal B}_t:~t\in\mathbb{N}\}$ is to be determined.\\
Next, we show by letting ${\cal B}_\infty=\lim_{t\longrightarrow\infty}{\cal B}_t$, where ${\cal B}_t=\sqrt{H_t\Delta_t{Q}^{-1}}$, and $\Delta_t\tri{d}iag\{\delta_{t,1},\ldots,\delta_{t,p}\}$, that $\Lambda_\infty=\lim_{t\longrightarrow\infty}\Lambda_t
=\lim_{t\longrightarrow\infty} \mathbb{E}\Big\{K_tK^{tr}_t\Big\}$, and also $\lim_{n\longrightarrow\infty}\frac{1}{n+1}\mathbb{E}\Big\{\sum_{t=0}^n||X_t-Y_t||_{\mathbb{R}^p}^2\Big\}=\lim_{n\longrightarrow\infty}\frac{1}{n+1}\mathbb{E}\Big\{\sum_{t=0}^n||K_t-\tilde{K}_t||_{\mathbb{R}^p}^2\Big\}=D$. Clearly, by (\ref{equation52}), (\ref{eq.12i}), (\ref{equation95ii})
\begin{align}
&\lim_{t\longrightarrow\infty}\mathbb{E}\Big{\{}(X_t-{Y}_t)^{tr}(X_t-{Y}_t)\Big{\}}=\lim_{t\longrightarrow\infty}Trace~\mathbb{E}\Big{\{}(K_t-\tilde{K}_t)(K_t-\tilde{K}_t)^{tr}\Big{\}}\nonumber\\
&=\lim_{t\longrightarrow\infty}Trace~\mathbb{E}\Big{\{}(K_t-E_t^{tr}H_tE_tK_t-E_t^{tr}{\cal B}_tV^c_t)(K_t-E_t^{tr}H_tE_tK_t-E_t^{tr}{\cal B}_tV^c_t)^{tr}\Big{\}}\nonumber\\
&=\lim_{t\longrightarrow\infty} Trace\Big{\{}E_t^{tr}\Big{(}(I-H_t)diag(\lambda_{t,1},\ldots,\lambda_{t,p})(I-H_t)^{tr}+({\cal B}_tQ{\cal B}_t^{tr})\Big{)}E_t\Big{\}}\nonumber\\
&\stackrel{(a)}=\lim_{t\longrightarrow\infty}Trace\Big{\{}diag(\delta_{t,1},\ldots,\delta_{t,p})\Big{\}}=Trace\Big{\{}diag(\delta_{\infty,1},\ldots,\delta_{\infty,p})\Big{\}}=D~~~~\nonumber
\end{align}
where $(a)$ holds by setting ${\cal B}_{\infty}$ and ${\cal B}_t$ as in (\ref{equation11mmma}). Also, by (\ref{equation.2}),
\begin{align}
&\lim_{n\longrightarrow\infty}\frac{1}{n+1}R_{0,n}^{na,K^n,\tilde{K}^n}(D)\leq\lim_{n\longrightarrow\infty}\frac{1}{n+1}\mathbb{I}_{X^n\rightarrow{Y^n}}(P_{K^n},\overrightarrow{Q}_{\tilde{K}^n|K^n})\nonumber\\
&=\lim_{n\longrightarrow\infty}\frac{1}{n+1}\sum_{t=0}^n{I}(K_t;\tilde{K}_t|\tilde{K}^{t-1})=\lim_{n\longrightarrow\infty}\frac{1}{n+1}\sum_{t=0}^n\Big(H(\tilde{K}_t|\tilde{K}^{t-1})-H(\tilde{K}_t|\tilde{K}^{t-1},K_t)\Big)\nonumber\\
&\stackrel{(b)}\leq\lim_{n\longrightarrow\infty}\frac{1}{n+1}\sum_{t=0}^n\Big(H(\tilde{K}_t)-H(\tilde{K}_t|\tilde{K}^{t-1},K_t)\Big)\stackrel{(c)}\leq\lim_{n\longrightarrow\infty}\frac{1}{n+1}\sum_{t=0}^n\Big(H(\tilde{K}_t)-H(\tilde{K}_t|K_t)\Big)\nonumber\\
&\stackrel{(d)}\leq\lim_{n\longrightarrow\infty}\frac{1}{n+1}\sum_{t=0}^n\Big(H(\tilde{K}_t)-H(E_t^{tr}{\cal B}_tV^c_t)\Big)\label{equa.18}
\end{align}
where $(b)$ follows from the fact that conditioning reduces entropy, $(c)$ follows from the fact that $\tilde{K}_t=E_t^{tr}H_tE_tK_t+E_t^{tr}{\cal B}_tV^c_t$ is a memoryless Gaussian channel, and $(d)$ follows from the orthogonality of $K_t$ and $V_t^c$,~$\forall{t}\in\mathbb{N}$. Next, we compute the entropy rates appearing in (\ref{equa.18}). The covariance of the Gaussian zero mean, noise process $\{E_t^{tr}{\cal B}_tV^c_t,~t\in\mathbb{N}\}$ is obtained as follows.
\begin{align}
\lim_{t\longrightarrow\infty}\mathbb{E}\Big\{(E_t^{tr}{\cal B}_tV^c_t)(E_t^{tr}{\cal B}_tV_t^{c})^{tr}\Big\}&=\lim_{t\longrightarrow\infty}\mathbb{E}\Big\{E_t^{tr}{\cal B}_tV^c_tV_t^{c,tr}{\cal B}_t^{tr}E_t\Big\}\nonumber\\
&=\lim_{t\longrightarrow\infty}\Big\{E_t^{tr}\sqrt{H_t\Delta_tQ^{-1}}Q\sqrt{H_t\Delta_tQ^{-1}}E_t\Big\}\nonumber\\
&=\lim_{t\longrightarrow\infty}E_t^{tr}H_t\Delta_tE_t=\lim_{t\longrightarrow\infty}E_t^{tr}diag\{\eta_{t,1}\delta_{t,1},\ldots,\eta_{t,p}\delta_{t,p}\}E_t\nonumber\\
&=E^{tr}_{\infty}diag\{\eta_{\infty,1}\delta_{\infty,1},\ldots,\eta_{\infty,p}\delta_{\infty,p}\}E_\infty.\label{equa.19}
\end{align}
The covariance of the process $\{\tilde{K}_t:~t\in\mathbb{N}\}$ is obtained as follows.
\begin{align}
\lim_{t\longrightarrow\infty}\mathbb{E}\Big\{\tilde{K}_t\tilde{K}^{tr}_t\Big\}
&=\lim_{t\longrightarrow\infty}\mathbb{E}\Big\{E_t^{tr}H_tE_tK_tK_t^{tr}E_t^{tr}H_tE_t+E_t^{tr}{\cal B}_tV^c_tV_t^{c,tr}{\cal B}_t^{tr}E_t\Big\}\nonumber\\
&=\lim_{t\longrightarrow\infty}\Big\{E_t^{tr}H_tE_t\Lambda_tE_t^{tr}H_tE_t+E_t^{tr}\sqrt{H_t\Delta_tQ^{-1}}Q
\sqrt{H_t\Delta_tQ^{-1}}E_t\Big\}\nonumber\\
&=\lim_{t\longrightarrow\infty}\Big\{E_t^{tr}\big(diag\{\eta^2_{t,1}\lambda_{t,1},\ldots,\eta^2_{t,p}\lambda_{t,p}\}+diag\{\eta_{t,1}\delta_{t,1},\ldots,\eta_{t,p}\delta_{t,p}\}\big)E_t\Big\}\nonumber\\
&=\lim_{t\longrightarrow\infty}E_t^{tr}diag\{\lambda_{t,1}-\delta_{t,1},\ldots,\lambda_{t,p}-\delta_{t,p}\}E_t,~\lambda_{t,i}-\delta_{t,i}\geq{0},\forall{t}\nonumber\\
&=E_\infty^{tr}diag\{\lambda_{\infty,1}-\delta_{\infty,1},\ldots,\lambda_{\infty,p}-\delta_{\infty,p}\}E_\infty.\label{equa.20}
\end{align}
Using (\ref{equa.20}) we obtain the first term of (\ref{equa.18}) as follows
\begin{align}
\lim_{n\longrightarrow\infty}\frac{1}{n+1}\sum_{t=0}^nH(\tilde{K}_t)
&=\lim_{n\longrightarrow\infty}\frac{1}{2}\frac{1}{n+1}\sum_{t=0}^n\log\Big\{\big(2{\pi}e\big)\times_{i=1}^p\big(\lambda_{t,i}-\delta_{t,i}\big)^{+}\Big\}\nonumber\\
&=\lim_{n\longrightarrow\infty}\frac{1}{2}\frac{1}{n+1}\sum_{t=0}^n\sum_{i=1}^p\log\Big\{\big(2{\pi}e\big)\big(\lambda_{t,i}-\delta_{t,i}\big)^{+}\Big\}\nonumber\\
&=\frac{1}{2}\sum_{i=1}^p\log\Big\{\big(2{\pi}e\big)\big(\lambda_{\infty,i}-\delta_{\infty,i}\big)^{+}\Big\}.\label{equa.21}
\end{align}
Also, by (\ref{equa.19}), we obtain the second term in (\ref{equa.18}) as follows.
\begin{align}
\lim_{n\longrightarrow\infty}&\frac{1}{n+1}\sum_{t=0}^nH(E_t^{tr}{\cal B}_tV^c_t)=\lim_{n\longrightarrow\infty}\frac{1}{2}\frac{1}{n+1}\sum_{t=0}^n\log\Big\{\big(2{\pi}e\big){d}iag\{\eta_{t,1}\delta_{t,1},\ldots,\eta_{t,p}\delta_{t,p}\}\Big\}\nonumber\\
&=\frac{1}{2}\sum_{t=0}^n\log\Big\{\big(2{\pi}e\big)\times_{i=1}^p\big(\eta_{\infty,i}\delta_{\infty,i}\big)\Big\}=\frac{1}{2}\sum_{i=1}^p\log\Big\{\big(2{\pi}e\big)\big(\eta_{\infty,i}\delta_{\infty,i}\big)\Big\}.\label{equa.22}
\end{align}
By using (\ref{equa.21}) and (\ref{equa.22}) in (\ref{equa.18}) we have the following upper bound
\begin{align}
\lim_{n\longrightarrow\infty}R_{0,n}^{na,K^n,\tilde{K}^n}(D)&\leq\frac{1}{2}\sum_{i=1}^p\log\Big\{\big(2{\pi}e\big)\big(\lambda_{\infty,i}-\delta_{\infty,i}\big)^{+}\Big\}-\frac{1}{2}\sum_{i=1}^p\log\Big\{\big(2{\pi}e\big)\big(\eta_{\infty,i}\delta_{\infty,i}\big)\Big\}\nonumber\\
&=\frac{1}{2}\sum_{i=1}^p\log\Big\{\frac{\big(\lambda_{\infty,i}-\delta_{\infty,i}\big)^{+}}{\eta_{\infty,i}\delta_{\infty,i}}\Big\}
=\frac{1}{2}\sum_{i=1}^p\log\frac{\lambda_{\infty,i}}{\delta_{\infty,i}}\nonumber
\end{align}
where $\delta_{\infty,i}=\min\{\xi_\infty,\lambda_{\infty,i}\}$ and $\sum_{i=1}^p\delta_{\infty,i}=D$.\\
\noi{\it Achievable Lower Bound~(Analogous to Shannon's Lower Bound).} Next, we apply Theorem~\ref{alternative_expression} to obtain a lower bound for the nonanticipative RDF $R^{na}(D)=\lim_{n\longrightarrow\infty}\frac{1}{n+1}R^{na}_{0,n}(D)=\lim_{n\longrightarrow\infty}\frac{1}{n+1}R_{0,n}^{na, K^n,\tilde{K}^n}(D)$.\\
\noi Applying Theorem~\ref{alternative_expression} to $R_{0,n}^{na, K^n,\tilde{K}^n}(D)$, the set $\Psi_s$ is defined by 
\begin{align}
\Psi_s\tri\Big\{&\{\lambda_t(k_t,\tilde{k}^{t-1}):~t=0,1,\ldots,n\}:\lambda_t(k_t,\tilde{k}^{t-1})\geq{0},\nonumber\\
&\int{e}^{s||K_t-\tilde{K}_t||_{\mathbb{R}^p}^2}{\lambda(k_t,\tilde{k}^{t-1})}\bar{p}(k_t|\tilde{k}^{t-1})dk_t\leq{1},~t=0,1,\ldots,n\Big\}\label{equa.24}
\end{align}
where $\bar{p}(k_t|\tilde{k}^{t-1})$ denotes the conditional density of $k_t$ given $\tilde{k}^{t-1}$. Choose $s\leq{0}$ and take $\lambda_t(k_t,\tilde{k}^{t-1})\in\Psi_s$ such that
\begin{align}
\lambda_t(k_t,\tilde{k}^{t-1})=\frac{\alpha_t}{\bar{p}(k_t|\tilde{k}^{t-1})}\label{equa.23}
\end{align}
for some $\alpha_t$ not depending on $k_t$. Substituting (\ref{equa.23}) into the reduced integral inequality in (\ref{equa.24}) we obtain
\begin{align*}
{\alpha_t}\int{e}^{s||K_t-\tilde{K
}_t||_{\mathbb{R}^p}^2}dk_t\leq{1},~t=0,1,\ldots,n.
\end{align*}
By changing the variable of integration we also obtain
\begin{align}
{\alpha_t}\int{e}^{s||z_t||_{\mathbb{R}^p}^2}dz_t=\alpha_t\sqrt{(-\frac{\pi}{s})^{p}}=\alpha_t(-\frac{\pi}{s})^{\frac{p}{2}}\leq{1},~t=0,1,\ldots,n.\label{equa.25}
\end{align}
\noi By setting $\alpha_t=(-\frac{s}{\pi})^{\frac{p}{2}},~t=0,1,\ldots,n$, the inequality of (\ref{equa.25}) holds with equality. Then, by Theorem~\ref{alternative_expression}, we have
\begin{align}
\lim_{n\longrightarrow\infty}\frac{1}{n+1}R_{0,n}^{na,K^n,\tilde{K}^n}(D)&\geq{s}D+\lim_{n\longrightarrow\infty}\frac{1}{n+1}\sum_{t=0}^n\int_{K_t\times\tilde{K}_{0,t-1}}\bar{p}(k_t,\tilde{k}^{t-1})\log\Big(\lambda_t(k_t,\tilde{k}^{t-1})\Big)dk_td\tilde{k}^{t-1}\nonumber\\
&\stackrel{(e)}=sD+\lim_{n\longrightarrow\infty}\frac{1}{n+1}\sum_{t=0}^n\int_{K_t\times\tilde{K}_{0,t-1}}\bar{p}(k_t,\tilde{k}^{t-1})\log\Big(\frac{(-\frac{s}{\pi})^{\frac{p}{2}}}{\bar{p}(k_t|\tilde{k}^{t-1})}\Big)dk_t{d}\tilde{k}^{t-1}\nonumber\\
&=sD+\lim_{n\longrightarrow\infty}\frac{1}{n+1}\sum_{t=0}^n\log(-\frac{s}{\pi})^{\frac{p}{2}}+\lim_{n\longrightarrow\infty}\frac{1}{n+1}\sum_{t=0}^nH(K_t|\tilde{K}^{t-1})\nonumber\\
&\stackrel{(f)}=sD+\lim_{n\longrightarrow\infty}\frac{1}{n+1}\sum_{t=0}^n\log(-\frac{s}{\pi})^{\frac{p}{2}}+\lim_{n\longrightarrow\infty}\frac{1}{n+1}\sum_{t=0}^nH(K_t)\label{equa.26}
\end{align}
where $(e)$ follows from (\ref{equa.23}), and $(f)$ follows from the orthogonality of $K_t$ and $\tilde{K}^{t-1}$. Next, we need to find the Lagrangian multiplier $``s"$ so that the lower bound (\ref{equa.26}) equals $\frac{1}{2}\sum_{i=1}^p\log\frac{\lambda_{\infty,i}}{\delta_{\infty,i}}$. To this end, we need to ensure existence of some $s<0$ such that the following identity holds.
\begin{align}
&sD+\lim_{n\longrightarrow\infty}\frac{1}{n+1}\sum_{t=0}^n\log(-\frac{s}{\pi})^{\frac{p}{2}}+\lim_{n\longrightarrow\infty}\frac{1}{2}\frac{1}{n+1}\sum_{t=0}^n\log2\pi{e}|\Lambda_t|=\lim_{n\longrightarrow\infty}\frac{1}{2}\frac{1}{n+1}\sum_{t=0}^n\sum_{i=1}^p\log\frac{\lambda_{t,i}}{\delta_{t,i}}\nonumber\\
&\Longrightarrow{s}D+\lim_{n\longrightarrow\infty}\frac{1}{2}\frac{1}{n+1}\sum_{t=0}^n\sum_{i=1}^p\log(-\frac{s}{\pi})+\frac{1}{2}\sum_{i=1}^p\log2\pi{e}(\lambda_{\infty,i})=\frac{1}{2}\sum_{i=1}^p\log\frac{\lambda_{\infty,i}}{\delta_{\infty,i}}\nonumber\\
&\Longrightarrow\lim_{n\longrightarrow\infty}\frac{1}{2}\frac{1}{n+1}\log{e}^{2(n+1)s\sum_{i=1}^p\delta_{t,i}}+\lim_{n\longrightarrow\infty}\frac{1}{2}\frac{1}{n+1}\sum_{t=0}^n\sum_{i=1}^p\log(-\frac{s}{\pi})\nonumber\\
&=\frac{1}{2}\sum_{i=1}^p\log\frac{\lambda_{\infty,i}}{\delta_{\infty,i}}-\frac{1}{2}\sum_{i=1}^p\log2\pi{e}(\lambda_{\infty,i})\nonumber\\
&\Longrightarrow\lim_{n\longrightarrow\infty}\frac{1}{2}\frac{1}{n+1}\sum_{t=0}^n\sum_{i=1}^p\log{e}^{2s\delta_{t,i}}+\lim_{n\longrightarrow\infty}\frac{1}{2}\frac{1}{n+1}\sum_{t=0}^n\sum_{i=1}^p\log(-\frac{s}{\pi})=\frac{1}{2}\sum_{i=1}^p\log\frac{1}{2\pi{e}\delta_{\infty,i}}\nonumber\\
&\Longrightarrow\lim_{n\longrightarrow\infty}\frac{1}{2}\frac{1}{n+1}\sum_{t=0}^n\sum_{i=1}^p\log{e}^{2s\delta_{t,i}}(-\frac{s}{\pi})=\frac{1}{2}\sum_{i=1}^p\log\frac{1}{2\pi{e}\delta_{\infty,i}}\nonumber\\
&\Longrightarrow\frac{1}{2}\sum_{i=1}^p\log{e}^{2s\delta_{\infty,i}}(-\frac{s}{\pi})=\frac{1}{2}\sum_{i=1}^p\log\frac{1}{2\pi{e}\delta_{\infty,i}}\Longrightarrow{e}^{2s\delta_{\infty,i}}(-\frac{s}{\pi})=\frac{1}{2\pi{e}\delta_{\infty,i}}\Longrightarrow{s}=-\frac{1}{2\delta_{\infty,i}}\nonumber
\end{align}
where $\delta_{\infty,i}=\{\xi_\infty,\lambda_{\infty,i}\}$. Now, if $\delta_{\infty,i}=\xi_\infty$ then ${s}=-\frac{1}{2\delta_{\infty,i}}$ and the nonanticipative RDF is bounded below by the following expression
\begin{align}
\lim_{n\longrightarrow\infty}R_{0,n}^{na,K^n,\tilde{K}^n}(D)\geq\lim_{n\longrightarrow\infty}\frac{1}{2}\frac{1}{n+1}\sum_{t=0}^n\sum_{i=1}^p\log\frac{\lambda_{t,i}}{\delta_{t,i}}=\frac{1}{2}\sum_{i=1}^p\log\frac{\lambda_{\infty,i}}{\delta_{\infty,i}}
\end{align}
which is the desired lower bound with $\sum_{i=1}^p\delta_{\infty,i}=D$. When $\delta_{\infty,i}=\lambda_{\infty,i}$, no encoding is performed and there is no need to prove a lower bound on $R_{0,n}^{na,K^n,\tilde{K}^n}(D)$. This completes the proof of (\ref{equa.13}).\\
\noi Next, we determine the expression of $\Lambda_{\infty}$. By definition, $\Lambda_{\infty}=\lim_{t\longrightarrow\infty}\Lambda_t$, where $\Lambda_t=cov\big(X_t-\mathbb{E}\big\{X_t|\sigma\{Y^{t-1}\}\big\}\big)$. Since $X_t-\mathbb{E}\big\{X_t|\sigma\{Y^{t-1}\}\big\}=CZ_t+NV_t-C\mathbb{E}\big\{Z_t|\sigma\{Y^{t-1}\}\big\}$ then $\Lambda_t=C\Sigma_tC^{tr}+NN^{tr}$. Let $\widehat{Z}_{t|t-1}=\mathbb{E}\big{\{}Z_t|\sigma\{{Y}^{t-1}\}\big{\}}$. Clearly, $\Sigma_t\tri\mathbb{E}\big\{\big(Z_t-\mathbb{E}\{Z_t|{Y}^{t-1}\}\big)\big(Z_t-\mathbb{E}\{Z_t|{Y}^{t-1}\}\big)^{tr}\big\}$. Moreover, $\Lambda_\infty=C\lim_{t\longrightarrow\infty}\Sigma_tC^{tr}+NN^{tr}$. Therefore, to determine $\Sigma_\infty\tri\lim_{t\longrightarrow\infty}\Sigma_t$, we need the equation of the error $e_t\tri{Z}_t-\widehat{Z}_{t|t-1}$, hence the equation of the least-squares filter of $Z_t$ given all the previous outputs $Y^{t-1}$, namely $\widehat{Z}_{t|t-1}$. From Fig.~\ref{communication_system}, we deduce that ${Y}_t=\tilde{K}_t+\widehat{X}_{t|t-1}$, where $\{\widehat{X}_{t|t-1}:~t\in\mathbb{N}\}$ is obtained from the modified Kalman filter as follows. Recall that
\begin{align*}
{Y}_t=\tilde{K}_t+\widehat{X}_{t|t-1}&=E_\infty^{tr}H_{\infty}E_{\infty}(X_t-\widehat{X}_{t|t-1})+E_\infty^{tr}{\cal B}_{\infty}V^c_t+\widehat{X}_{t|t-1}\nonumber\\
&=E_\infty^{tr}H_{\infty}E_{\infty}(CZ_t+NV_t-C\widehat{Z}_{t|t-1})+E_\infty^{tr}{\cal B}_{\infty}V^c_t+C\widehat{Z}_{t|t-1}\nonumber\\
&=E_\infty^{tr}H_{\infty}E_{\infty}C(Z_t-\widehat{Z}_{t|t-1})+C\widehat{Z}_{t|t-1}+E_\infty^{tr}H_{\infty}E_{\infty}N V_t+E_{\infty}^{tr}{\cal B}_{\infty}V^c_t\nonumber
\end{align*}
where $\{V_t:~t\in\mathbb{N}\}$ and $\{V^c_t:~t\in\mathbb{N}\}$ are independent Gaussian vectors. Then $\widehat{Z}_{t|t-1}$ is given by the modified Kalman filter \cite{caines1988} (\ref{10}),~(\ref{11}). Notice that the filter is ergodic with initial condition $\hat{Z}_{0|-1} =\mathbb{E}\{ Z_0|Y^{-1}\}$ and $\Sigma_0$ the covariance of $Z_0- \hat{Z}_{0|-1}$ which is Gaussian $N(0,\Sigma_\infty)$.\qed

\bibliographystyle{IEEEtran}
\bibliography{photis_references_nonanticipative_RDF,photis_references_DI_properties}

\end{document}